\pdfoutput=1
\documentclass[a4paper,12pt]{article}
\usepackage[a4paper, top=1.6cm, bottom=1.6cm, left=1.6cm, right=1.6cm]%
{geometry}
\usepackage[utf8]{inputenc}
\usepackage{setspace}
\usepackage{epsfig,epsf,psfrag, url, setspace, graphicx,mathtools,subfigure,comment}
\usepackage{a4wide,amsmath, amsthm, amssymb, multirow}
\usepackage{xcolor,colortbl}
\usepackage{verbatim}
\usepackage{arydshln} 
\usepackage{placeins} 

\usepackage{bm}
\usepackage{makecell}
\usepackage[flushleft]{threeparttable} 
\usepackage{graphicx}
\usepackage{amsmath}
\usepackage{amsfonts} 
\usepackage{commath} 
\usepackage{stackengine}
\usepackage[T1]{fontenc} 
\stackMath
\newcommand\tsup[2][2]{
 \def\useanchorwidth{T}%
  \ifnum#1>1%
    \stackon[-.5pt]{\tsup[\numexpr#1-1\relax]{#2}}{\scriptscriptstyle\sim}%
  \else%
    \stackon[.5pt]{#2}{\scriptscriptstyle\sim}%
  \fi%
}
\usepackage{tabularx, booktabs}
\usepackage{multirow}
\usepackage[flushleft]{threeparttable}
\usepackage{makecell}
\usepackage{adjustbox}
\usepackage{appendix}
\usepackage[figuresright]{rotating} 
\usepackage{hyperref}
\usepackage{cleveref}
\usepackage{times} 
\usepackage{authblk}
\usepackage[round]{natbib}

\usepackage[plain,noend]{algorithm2e}
\newcommand{\ve}[1]{{\mbox{\boldmath ${#1}$}}}
\usepackage{caption}
\usepackage{tikz-cd}
\usetikzlibrary{arrows}
\usepackage{tikz}
\usetikzlibrary{arrows.meta}

\tikzset{commutative diagrams/.cd, arrow style=tikz,diagrams={>= triangle 45}}

\newcommand\independent{\protect\mathpalette{\protect\independenT}{\perp}}
\def\independenT#1#2{\mathrel{\rlap{$#1#2$}\mkern2mu{#1#2}}}

\makeatletter
\renewcommand{\algocf@captiontext}[2]{#1\algocf@typo. \AlCapFnt{}#2} 
\def\@algocf@capt@plain{top}
\renewcommand{\algocf@makecaption}[2]{%
  \addtolength{\hsize}{\algomargin}%
  \sbox\@tempboxa{\algocf@captiontext{#1}{#2}}%
  \ifdim\wd\@tempboxa >\hsize
    \hskip .5\algomargin%
    \parbox[t]{\hsize}{\algocf@captiontext{#1}{#2}}
  \else%
    \global\@minipagefalse%
    \hbox to\hsize{\box\@tempboxa}
  \fi%
  \addtolength{\hsize}{-\algomargin}%
}
\makeatother

\usepackage{amsmath,amsthm}
      \theoremstyle{plain}
      \newtheorem{assumption}{Assumption}
      \newtheorem{theorem}{Theorem}
      \newtheorem{remark}{Remark}

\font\myfont=cmr12 at 19pt

\usepackage[symbol]{footmisc}

\newcommand*\samethanks[1][\value{footnote}]{\footnotemark[#1]}

\usepackage[pagewise]{lineno}

\begin{document}

\title{\myfont T2-DAG: a powerful test for differentially expressed gene pathways via graph-informed structural equation modeling}

\author[1,\thanks{These authors contributed equally: Jin Jin and Yue Wang.}]{Jin Jin }
\affil[1]{Department of Biostatistics, Bloomberg School of Public Health, Johns Hopkins University, 615 N Wolfe St, Baltimore, Maryland 21205, U.S.A.
\protect\\ jjin31@jhu.edu}

\author[2,\samethanks]{Yue Wang }
\affil[2]{School of Mathematical and Natural Sciences, Arizona State University, 4701 W Thunderbird Rd, Glendale, AZ 85306, U.S.A. 
\protect\\ yue.wang.stat@asu.edu}
\date{}
\maketitle

\doublespacing

\begin{abstract}
\noindent
\textbf{Motivation:} A major task in genetic studies is to identify genes related to human diseases and traits to understand functional characteristics of genetic mutations and enhance patient diagnosis.
Compared to marginal analyses of individual genes, identification of gene pathways, i.e., a set of genes with known interactions that collectively contribute to specific biological functions, can provide more biologically meaningful results. Such gene pathway analysis can be formulated into a high-dimensional two-sample testing problem. {Given the typically limited sample size of gene expression datasets, most existing two-sample tests tend to have compromised powers because they ignore or only inefficiently incorporate the auxiliary pathway information on gene interactions.}\\
\textbf{Results:} We propose T2-DAG, a Hotelling's $T^2$-type test for detecting differentially expressed gene pathways, which efficiently leverages the auxiliary pathway information on gene interactions from existing pathway databases through a linear structural equation model. We further establish its asymptotic distribution under pertinent assumptions. Simulation studies under various scenarios show that T2-DAG outperforms several representative existing methods with well-controlled type-I error rates and substantially improved powers, even with incomplete or inaccurate pathway information or unadjusted confounding effects. We also illustrate the performance of T2-DAG in an application to detect differentially expressed KEGG pathways between different stages of lung cancer. \\
\textbf{Availability:} The R package T2DAG is available on Github at
https://github.com/Jin93/T2DAG.\\
\textbf{Contact:} \href{jjin31@jhu.edu}{jjin31@jhu.edu}; \href{yue.wang.stat@asu.edu}{yue.wang.stat@asu.edu}\\
\textbf{Supplementary information:} Supplementary data are available at \textit{Bioinformatics}
online.
\end{abstract}



\section{Introduction}\label{sec1}
With the continuous development of high-throughput sequencing technologies, there is an increasing need for powerful statistical methods for large-scale quantitative comparison between populations/conditions. 
For example, in gene expression data analysis, a major task is to identify genes from a large gene list that express differentially between different conditions, such as different disease statuses.
Various approaches have been proposed for this task, providing critical insights into the molecular basis of many complex human traits/diseases; see \cite{robinson2010edger}, \cite{love2014moderated} and the references therein.

Besides separately analyzing individual genes, recently there is a growing interest in detecting differentially expressed gene pathways. 
Here, a gene pathway refers to a set of correlated genes with known interactions collectively contributing to specific biological functions. 
Many biologically meaningful gene pathways have been identified via experiments or statistical modeling, along with identified gene interactions, such as activation/expression or inhibition/repression of one gene by another.
Public databases containing such pathways include Kyoto Encyclopedia of Genes and Genomes (KEGG, \citealp{kanehisa2000,kanehisa2019}), Gene Ontology (GO, \citealp{ashburner2000gene,gene2019gene}), BioCarta databases \citep{biocarta}, etc. {Fig. \ref{kegg.example} shows two directed acyclic graphs (DAGs) illustrating two KEGG gene pathways, which will be revisited in our data analysis in Section \ref{sec4}. 
	$a$ \begin{tikzpicture}[>={LaTeX[width=1.5mm,length=1.5mm]},->]
  \draw[line width=0.2mm] (0,1) -- (0.5,1);
\end{tikzpicture} $b$ and $a$ 
\begin{tikzpicture}[>={LaTeX[width=1.5mm,length=1mm]},-<]
  \draw[line width=0.2mm] (0,1) -- (0.5,1);
\end{tikzpicture} $b$
both indicate activation from gene $a$ to gene $b$; $a \dashrightarrow b$ and $a$ \begin{tikzpicture}[>={LaTeX[width=1.5mm,length=1mm]},-|]
  \draw[line width=0.2mm] (0,1) -- (0.5,1);
\end{tikzpicture} $b$ indicate expression and inhibition, respectively, from gene $a$ to gene $b$.
A common and notable feature of the two pathways is the high sparsity of the edges in the corresponding DAGs; that is, only a small proportion of the gene pairs are connected. 
Such sparsity pattern is also present in most identified gene pathways \citep{wille2004sparse,li2012sparse,chun2015gene}.}

\begin{figure}[ht!]
\centering
    \includegraphics[scale=0.42]{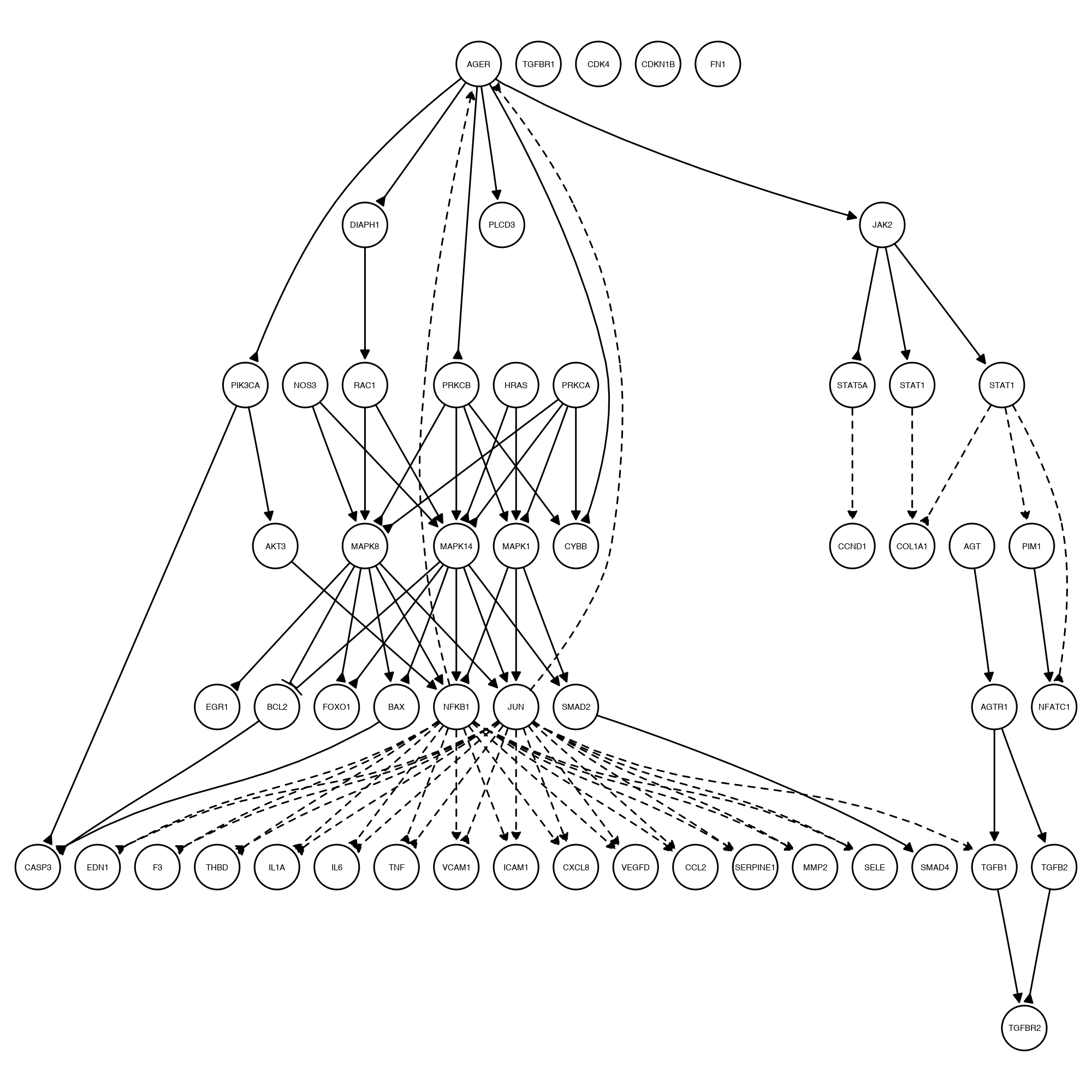}
    \hspace{25pt}
    \includegraphics[scale=0.42]{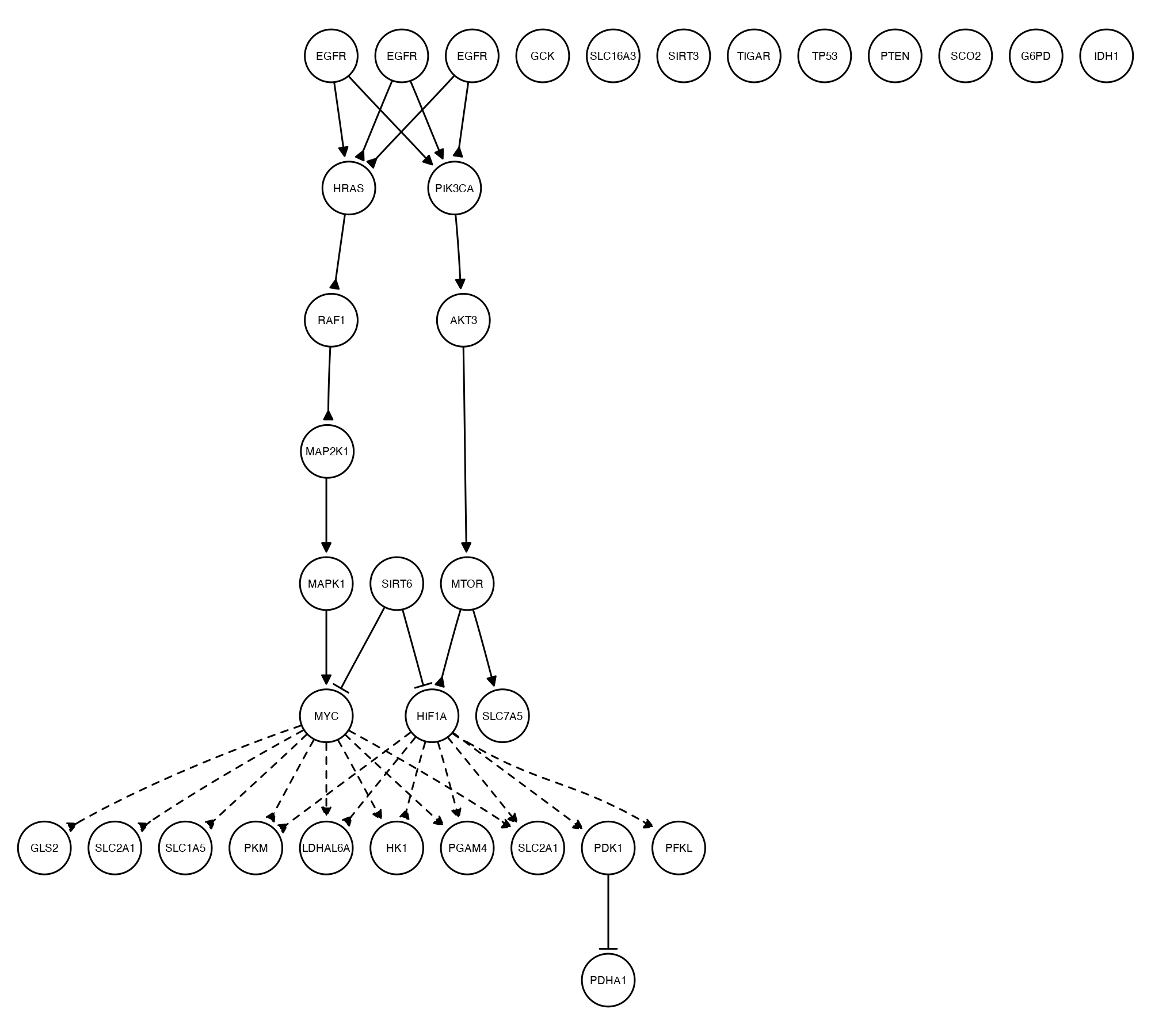}
	\caption{
	Directed acyclic graphs (DAGs) describing gene interactions in an AGE-RAGE signaling pathway in diabetic complications (left panel) and a pathway for central carbon metabolism in cancer (right panel) from the KEGG database. 
\label{kegg.example}}
\end{figure}

There are two main reasons for investigating differentially expressed gene pathways over individual genes. First, marginal tests for individual genes tend to have low statistical power due to limited sample sizes of typical gene expression datasets. On the other hand, gene pathway tests can achieve enhanced power by aggregating concerted signals from multiple related genes in the same pathway. Second, 
marginal analyses of individual genes often ignore interactions among genes, which  potentially lead to spurious conclusions inconsistent with the underlying biological processes. 
Since gene interactions can provide critical insights into the molecular mechanism of
the progression of traits/diseases, efficiently leveraging the auxiliary information on gene interactions can facilitate the detection of disease- and trait-specific genes and yield biologically more meaningful results.

To formalize the problem, let $\mathcal{G} = (\mathcal{V}, \mathcal{D})$ denote the pathway of interest, where $\mathcal{V} = \{1, \ldots, p\}$ and $\mathcal{D}$ denote the gene set and edge set between genes, respectively. We assume that $\mathcal{G}$ is a DAG; i.e., all edges in $\mathcal{D}$ are directed edges with no loops. 
A DAG $\mathcal{G}$ can be summarized by an adjacency matrix  $A = (a_{jk})\in \mathbb{R}^{p \times p}$, where $a_{jk} = 1$ if the edge $j \rightarrow k$ is present in $\mathcal{D}$, in which case we call gene $j$ a parent of gene $k$ and gene $k$ a child of gene $j$, and $a_{jk} = 0$ otherwise. 
Although some identified gene pathways cannot be described by DAGs, they tend to have a small number of loops and/or undirected edges (see Table S1 for a summary of the KEGG pathways included in our data analysis). Thus, for such gene pathways, analyses based on their largest directed acyclic subgraphs are still feasible and informative.
Approaches to directly handling non-DAG pathways are feasible but will require techniques different from the ones in this study, which we will briefly discuss in Section \ref{sec5}.

Assume we have observed expression levels of the genes in $\mathcal{V}$ for individuals in two different conditions. Specifically, let $x^{(l)}_{ij}$ denote the expression level of gene $j\in \mathcal{V}$ for individual $i\in\{1,\ldots,n_l\}$ in group $l\in\{1,2\}$.
We assume that the data are appropriately transformed so that ${\bf x}^{(l)}_i = (x^{(l)}_{i1}, \ldots, x^{(l)}_{ip})^\intercal$, $i=1,\ldots,n_l$, are independent and identically distributed ({\it i.i.d.}) random vectors from a multivariate Gaussian distribution with mean $\ve \mu_{l}$ ($l = 1, 2$) and covariance $\Sigma$. 
Besides observed gene expression data, 
there is auxiliary information on the presence/absence of directed relationships among genes {(contained in $\mathcal{D}$)}, such as activation/inhibition of one gene by another. 
Our goal is then to test for the inequality of the two population mean, i.e., $H_0: \ve \mu_1 = \ve \mu_2\text{ versus } H_a: \ve \mu_1 \neq \ve \mu_2$, based on ${\bf x}^{(l)}_i$s, while efficiently leveraging the auxiliary information on gene interactions.

Conventional hypothesis testing for {the inequality}  of two mean vectors has been thoroughly investigated. A classical approach is the Hotelling's $T^2$ test \citep{hotelling1931} with the test statistic
\begin{equation}\label{hotelling} 
T^2 = \frac{n_1 n_2 }{n_1 + n_2}(\overline{\bf x}^{(1)} - \overline{\bf x}^{(2)})^\intercal \widehat{\Sigma}^{-1}(\overline{\bf x}^{(1)} - \overline{\bf x}^{(2)}),
\end{equation}
where 
\(\overline{\bf x}^{(l)} = n_l^{-1}\sum_{i=1}^{n_l}{\bf x}^{(l)}_i
\)
and 
\(
\widehat{\Sigma} = (n_1 + n_2 -1)^{-1}  \sum_{l = 1}^2\sum_{i=1}^{n_l}({\bf x}^{(l)}_i - \overline{\bf x}^{(l)})({\bf x}^{(l)}_i - \overline{\bf x}^{(l)})^\intercal .\)
When the number of genes $(p)$ is smaller than the total sample size $(n_1 + n_2)$, $T^2$ is well-defined; 
under $H_0$, 
$T^2\left(n_1+n_2-p-1\right)/\left(p\left(n_1+n_2-2\right)\right)$
follows an $F$-distribution with parameters $p$ and $n_1 + n_2 - 1- p$. 
The $T^2$ is a sum-of-squares type of statistic powerful against ``dense alternatives", where the signals of differential expression are scattered among a large number of genes. This can often be true for gene pathway analyses, as genes in the same pathway tend to have concerted signals of expression that together form certain biological functions.

While Hotelling's $T^2$ is well known as the uniformly most powerful invariant (UMPI) test, it is not applicable when $p$ is larger than $(n_1+n_2)$, referred to as the high-dimensional setting hereafter. This is simply because in the high-dimensional setting the sample covariance matrix $\widehat{\Sigma}$ becomes  singular, making $T^2$
ill-posed. 
Research has been conducted extensively on extending $T^2$ to the high-dimensional setting. 
For example, 
\cite{bai1996} and \cite{srivastava2008test} directly modified the ill-posed $T^2$ by replacing $\widehat{\Sigma}$ in (\ref{hotelling}) with the identity matrix $I$ and the diagonal of $\widehat{\Sigma}$, respectively. 
\cite{lopes2011more} proposed to first project the high-dimensional data onto a randomly constructed low-dimensional space and then perform the Hotelling's $T^2$ test. \cite{chen2011regularized} replaced $\widehat{\Sigma}$ with $(\widehat{\Sigma} + \lambda I)$ for some regularization parameter $\lambda > 0$ and proposed a nonparametric testing procedure. 
Recently, \cite{li2020adaptable} extended the work of \cite{chen2011regularized} and proposed an Adaptive Regularized Hotelling's $T^2$ (ARHT) test that combines the test statistics corresponding to a set of $\lambda$s optimally chosen by a data-adaptive selection mechanism. 
Methods that are not direct modifications of Hotelling's $T^2$ are also extensive.
For example, \cite{cai2014} proposed the CLX test with a supremum-type test statistic powerful against ``sparse'' alternatives. \cite{gregory2015} proposed the GCT test that avoids the estimation of $\widehat{\Sigma}^{-1}$ by assuming an ordering of the variables so that the dependence between variables can be modeled according to their displacement. 
\cite{xu2016} proposed the aSPU test, a robust adaptive test powerful against a wide range of alternatives.

While these existing methods can, to varying degrees, handle potentially high-dimensional gene expression data, they cannot incorporate the auxiliary information on gene interactions available from pathway databases. 
Incorporating such auxiliary information could substantially improve the test power, especially considering the inadequate sample size of a typical gene expression dataset. 
Previous efforts to address this issue include the NetGSA test (\citealp{shojaie2009analysis} and \citealp{shojaie2010network}), which assumes that the {\it magnitude} of gene interactions 
is priorly known and incorporates such information using a linear mixed model. This approach has limited applicability in real applications because the magnitude of gene interactions is often unknown or only partially known and needs to be estimated using external data. 
Another approach proposed by \cite{jacob2012}, referred to as the Graph T2 test hereafter, 
projects data onto a low-dimensional space spanned by top eigenvectors of the graph Laplacian of the gene pathway and then performs the dimension reduction test in \cite{lopes2011more}. 
Unlike the NetGSA test, Graph T2 can simply incorporate information on the presence/absence of gene interactions, 
but it may have compromised power if the signal is not coherent with the selected eigenvectors of the corresponding graph Laplacian.

To address these limitations of the existing methods, we propose T2-DAG, a novel test for identifying differentially expressed gene pathways, which efficiently leverages auxiliary information on gene interactions.
Unlike the NetGSA test that requires additional information on the magnitude of gene interactions, the T2-DAG test only requires information on the presence/absence of gene interactions, which is more easily accessible in practice. Although the Gragh T2 test also requires only the information on the presence/absence of gene interactions, it simply incorporates a few eigenvectors of the corresponding graph Laplacian. On the contrary, T2-DAG directly incorporates all gene interactions into the analysis with minimal 
information loss. 
Our main idea is to link the auxiliary information on gene interactions with (partial) correlations among genes through a linear structural equation model (SEM). Based on the SEM, we construct a novel pathway-informed estimator of the precision matrix of the genes via the LDL decomposition \citep{krishnamoorthy2013matrix}, where the ordering of the genes is naturally determined by the topological ordering of the pathway. 
We then construct a test statistic by replacing the ill-posed $\widehat{\Sigma}^{-1}$ in (\ref{hotelling}) with the novel estimator of the precision matrix followed by a $Z$-transformation, which is denoted by T2-DAG $Z$. 
Under mild conditions, we prove that the proposed test statistic converges in distribution to the standard normal distribution. 
An alternative T2-DAG $\chi^2$ test can be constructed without the $Z$-transformation, which converges in distribution to $\chi^2_p$ under more stringent conditions.  
We also derive the asymptotic power of the tests under local alternatives that characterize the true signals of differential expression in the gene pathway between two populations/conditions.
Under various simulated data scenarios, 
the proposed T2-DAG tests show well-controlled type I error rates and substantially higher powers compared to multiple existing methods, even with partially missing or incorrect information on gene interactions, unadjusted confounding effects, or non-Gaussian error terms. Although theoretically limited, the T2-DAG $\chi^2$ test compares favorably to the T2-DAG $Z$ test in the simulated scenarios, where it has slightly higher but still well-controlled type I error rates and higher powers. 
We also demonstrate the efficiency of the T2-DAG tests in a lung cancer-related application to detect
differentially expressed gene pathways between different cancer stages.

\section{Methods}\label{sec2}

Recall that our goal is to test for inequality of the mean expression levels of a gene pathway in two populations, $H_0: \ve \mu_1 = \ve \mu_2\text{ versus } H_a: \ve \mu_1 \neq \ve \mu_2$, based on observed gene expression data, ${\bf x}^{(l)}_i = (x^{(l)}_{i1}, \ldots, x^{(l)}_{ip})^\intercal$, $i=1,\ldots,n_l$, $l=1,2$, while leveraging available auxiliary information on gene interactions described by a DAG $\mathcal{G} = (\mathcal{V}, \mathcal{D})$. 
Throughout the article, we use $\xi_1 \independent \xi_2$ to denote the independence between two random variables $\xi_1$ and $\xi_2$. We write $g(n) = O(f(n))$ if and only if $\lim_{n \rightarrow \infty} g(n)/f(n) < \infty$, and use $g(n) = o(f(n))$ to denote $\lim_{n \rightarrow \infty} g(n)/f(n) = 0$. We also write $f(n) \asymp g(n)$ if and only if $\lim_{n \rightarrow \infty} |f(n)/g(n)| = C$ for some constant $C > 0$. Let $I(\mathcal{A})$ denote the indicator function of the event $\mathcal{A}$; i.e., $I(\mathcal{A}) = 1$ if $\mathcal{A}$ is true, and $I(\mathcal{A}) = 0$ otherwise.
We use lowercase letters, bold lowercase letters, and uppercase letters to denote scalars, vectors, and matrices, respectively.

We first define a new random variable, ${\bf z}_i = (z_{i1}, \ldots, z_{ip})^\intercal$, based on the original gene expression data, ${\bf x}^{(l)}_i$, $i=1,\ldots,n_l$, $l=1,2$:
\[{\bf z}_i =
    \begin{cases}
      {\bf x}^{(1)}_i - \ve \mu_1 & i =1,2,\ldots, n_1\\
      {\bf x}^{(2)}_{i-n_1} - \ve \mu_2 & i = n_1+1,\ldots,n_1+n_2.
    \end{cases}  
\]
Given that ${\bf x}^{(l)}_i$  are {\it i.i.d.} random vectors that follow a multivariate normal distribution  $\mathcal{N}(\ve \mu_{l}, \Sigma)$, we can easily check that ${\bf z}_i \stackrel{i.i.d.}{\sim} \mathcal{N}(\ve 0,\Sigma)$, $i=1,2,\ldots,n_1+n_2$. 
We propose to model each ${\bf z}_i$ through the following SEM with a vector of Gaussian noises:
\begin{equation}\label{test:1}
{\bf z}_i = Q^\intercal{\bf z}_i + \ve \epsilon_i, \end{equation}
where $Q = (q_{jk}) \in \mathbb{R}^{p \times p}$ is an unknown coefficient matrix with $q_{jk}$ characterizing the {\it magnitude} of the effect of gene $j$ on gene $k$, and
$\ve \epsilon_i = (\epsilon_{i1}, \ldots, \epsilon_{ip})^\intercal$ are $i.i.d.$ random vectors following $\mathcal{N}(\ve 0,R)$. 
We make the following assumption on $R$:
\begin{assumption}
 $R = \text{diag}(r_1, \ldots, r_p)$.
\end{assumption}
\textit{Assumption 1} states that $\epsilon_{ij_1} \independent \epsilon_{ij_2}$ for all $j_1 \neq j_2$ and  $i = 1, \ldots, n_1 + n_2$. This
amounts to the assumption that there is no latent confounders, which is common in the existing SEM literature; see \cite{loh2014high}, \cite{peters2014identifiability}, and the references therein.
We next link the coefficient matrix $Q$ in model (\ref{test:1}) with the DAG $\mathcal{G}$ through 
the following assumption.
\begin{assumption}
  The support of $Q$, i.e., the indices of nonzero entries in $Q$, is informed by $\mathcal{G}$ through the corresponding adjacency matrix $A=(a_{jk})$, such that $q_{jk} \neq 0$ if and only if $a_{jk} \neq 0$.
\end{assumption}
{\it Assumption} 2 is only required for theoretical analysis. Simulation studies in Section \ref{sec3} will allow the auxiliary DAG to have missing or mis-specified edges.
Since $\mathcal{G}$ is a DAG, we can label the variables in $\mathcal{V}$ according to the topological ordering of $\mathcal{G}$ so that $Q$ is a strictly upper triangular matrix, i.e., an upper triangular matrix with the diagonal entries being 0 \citep{loh2014high}.
As a result, model (\ref{test:1}) can be rewritten as the following $p$ linear models:
\[
z_{ij} = f_j(\epsilon_{i1}, \ldots, \epsilon_{ij}) \mbox{, } j = 1, \ldots, p\mbox{, } i = 1, \ldots, n_1 + n_2,
\]
where $\{f_j(\cdot)\}_{j=1}^p$ are linear functions that are determined by $Q$. Thus, for $1 \leq j_2 < j_1 \leq p$ and each $i$,  one can check that 
$\mbox{Cov}(\epsilon_{ij_1}, z_{ij_2}) = 0$, indicating that $\epsilon_{ij_1} \independent z_{ij_2}$ under Gaussianity.  
This property will play an important role in establishing the asymptotic null distribution of our test statistic.

We next construct DAG-informed estimators of $\Sigma$ and $\Sigma^{-1}$. 
From model (\ref{test:1}), we can obtain $\ve \epsilon_i = (I - Q^\intercal){\bf z}_i$, $i = 1, \ldots, n_1 + n_2$. Applying the covariance function to both sides of this equation yields 
\(
    R = (I - Q^\intercal)\Sigma(I - Q).
\)
Since $Q$ is strictly upper triangular, we know that $I - Q$ is invertible, leading to
\begin{equation}\label{add:0325:1}
    \Sigma = (I - Q^\intercal)^{-1} R (I - Q)^{-1}.
\end{equation}
Taking the inverse of both sides of (\ref{add:0325:1}) gives
 \begin{equation}\label{test:3}
    \Sigma^{-1} = (I -  Q) {R}^{-1}(I -  Q)^\intercal.
 \end{equation}
Equations (\ref{add:0325:1}) and (\ref{test:3}) serve as the basis for the construction of the DAG-informed estimators of $\Sigma$ and $\Sigma^{-1}$, which are critical to the proposed test statistic. 
Note that (\ref{add:0325:1}) and (\ref{test:3}) are, respectively, LDL-type decompositions of $\Sigma$ and $\Sigma^{-1}$; similar decompositions have been employed in the literature for estimating high-dimensional covariance/precision matrices \citep{wu2003nonparametric, huang2006covariance}. 
The rationale behind the LDL decomposition is that there exists some natural ordering of the variables,  
which, in the case of gene pathways, is determined by the topological ordering of the corresponding graph $\mathcal{G}$.

We observe from equations (\ref{add:0325:1}) and (\ref{test:3}) that $\Sigma$ and $\Sigma^{-1}$ are functions of $Q$ and $R$.
We then construct estimators for $Q$ and $R$. 
Letting ${\bf q}_j$ denote the $j$-th column of $Q$, we decompose model (\ref{test:1}) into $p$ sub-models, 
\begin{equation}\label{submodel}
    z_{ij} = {\bf z}_i^\intercal {\bf q}_j + \epsilon_{ij}, \mbox{ } j = 1, \ldots, p.
\end{equation}
To estimate the unknown parameters in (\ref{submodel}), it is natural to consider the following square loss function for each sub-model $j$:
\begin{linenomath}
\begin{align}\label{leastsquare}
    &~  \sum_{l=1}^2 \sum_{i = 1}^{n_l} \left\{x^{(l)}_{ij} - \mu_{lj} - \left({\bf x}^{(l)}_i - {\ve \mu}_l \right)^\intercal  {\bf q}_j\right\}^2 \nonumber \\
    = &~ \sum_{l=1}^2 \sum_{i = 1}^{n_l} \left\{x^{(l)}_{ij} - \gamma_{lj} - {\bf x}^{(l)\intercal}_i  {\bf q}_j\right\}^2,
\end{align}
\end{linenomath}
where $\gamma_{lj} = \mu_{lj} + {\ve \mu}_l^\intercal {\bf q}_j$. 
Letting $S_j = \{k: a_{kj} = 1\}$, one can see that under {\it Assumption} 2, 
${\bf q}_j$ has $|S_j|$ nonzero entries with $|S_j| \leq j-1$ for $j > 1$ and that ${\bf q}_1 = {\ve 0}$. Since real gene pathways are far less dense than fully connected networks \citep{wille2004sparse,li2012sparse,chun2015gene}, we make the following assumption on $S_j$:
\begin{assumption}
$d = o(n_1 + n_2)$ as $n_1, n_2 \rightarrow \infty$, where $d = \max_j |S_j|$. 
\end{assumption}
This assumption appears to be reasonable for all 206 pathways considered in our data analysis in Section \ref{sec4} (see Table S1 for details). 
Under \textit{Assumptions} 1-3, we can estimate $\gamma_{lj}$ and the nonzeros entries of ${\bf q}_j$ in (\ref{leastsquare}) using the ordinary least squares (OLS). 
Specifically, denote by ${\bf 1}_{n_1+n_2}$ the $(n_1 + n_2)$-dimensional vector of all ones and let ${\bf g} = (g_1, \ldots, g_{n_1 + n_2})^\intercal$ denote the group indicator; that is, $g_i$ equals 1 for $i = 1, \ldots, n_1$ and equals 0 for $i = n_1 + 1, \ldots, n_1 + n_2$. Then, letting $X = \left[ {\bf x}_1^{(1)} ~\cdots ~ {\bf x}_{n_1}^{(1)} ~ {\bf x}_1^{(2)} ~\cdots~ {\bf x}_{n_2}^{(2)} \right]^\intercal$, we rewrite (\ref{leastsquare}) in the following matrix form
\begin{equation}\label{rev:1}
    \left\|X_j - \theta_{1j}{\bf 1}_{n_1+n_2} - \theta_{2j}{\bf g} - X_{S_j}{\bf q}_{j,S_j} \right\|_2^2,
\end{equation}
where $X_j$ denotes the $j$-th column of $X$, $X_{S_j}$ denotes the sub-matrix of $X$ with columns indexed by $S_j$, and ${\bf q}_{j,S_j}$ denotes the sub-vector of ${\bf q}_j$ with entries indexed by $S_j$ that are non-zero according to the DAG and need to be estimated.  Comparing (\ref{rev:1}) with (\ref{leastsquare}), one can see that $\gamma_{1j} = \theta_{1j} + \theta_{2j}$ and $\gamma_{2j} = \theta_{1j}$.
Thus, letting $W = [{\bf 1}_{n_1 + n_2} ~ {\bf g}]$,
we can obtain the OLS estimator of ${\bf q}_{j,S_j}$ and
$(\theta_{1j}, \theta_{2j})^\intercal$:
\begin{linenomath}
\begin{align}\label{rev:2}
    \widehat{\bf q}_{j,S_j} = &~ \left(X_{S_j}^\intercal (I_{n_1+n_2} - \mathcal{P}_{W}) X_{S_j}\right)^{-1} X_{S_j}^\intercal (I_{n_1 + n_2} - \mathcal{P}_{W}) X_{j}, \nonumber \\
    (\widehat{\theta}_{1j}, \widehat{\theta}_{2j})^\intercal = &~ (W^\intercal W)^{-1}W^\intercal \left(X_j - X_{S_j}\widehat{\bf q}_{j,S_j}\right),
\end{align}
\end{linenomath}
where $\mathcal{P}_W = W(W^\intercal W)^{-1}W^\intercal$ is the orthogonal projection matrix onto the column space of $W$. 
For ease of notation, we also denote the OLS estimator of $q_{kj}$ by $\widehat{q}_{kj}$ which equals 0 if $k \notin S_j$, and define the estimator of $Q$ as $\widehat{Q} = (\widehat{q}_{jk})_{j,k = 1, \ldots, p}$.  
Given $\widehat{Q}$, we then construct a ``plug-in'' estimator of $R$ denoted by $\widehat{R} = \mbox{diag}(\widehat{r}_1, \ldots, \widehat{r}_p)$, with 
\[\widehat{r}_j = \frac{1}{n_1 + n_2 - |S_j| - 4} 
\left\|X_j - \widehat{\theta}_{1j}{\bf 1}_{n_1+n_2} - \widehat{\theta}_{2j}{\bf g} - X_{S_j}\widehat{\bf q}_{j,S_j} \right\|_2^2.
\] 
The denominator, $(n_1 + n_2 - |S_j| - 4)$, makes $\widehat{r}_j$ slightly different from the classical OLS estimator of $r_j$. Although it may seem non-intuitive, we show in Lemma 1.3 in the supplementary material that $\mathbb{E}[\widehat{r}_j^{-1}] = r_j^{-1}$; this property is critical for establishing the asymptotic null distribution of our test statistic. 

By replacing $Q$ and $R$ in
(\ref{add:0325:1}) and (\ref{test:3}) with $\widehat{Q}$ and $\widehat{R}$, we have the following DAG-informed empirical estimators of $\Sigma$ and $\Sigma^{-1}$:
\[
\widehat{\Sigma}_{\text{DAG}} = (I - \widehat{Q}^\intercal)^{-1} \widehat{R} (I - \widehat{Q})^{-1} ~~\mbox{and}~~ 
\widehat{\Sigma}_{\text{DAG}}^{-1} = (I -  \widehat{Q}) \widehat{R}^{-1}(I -  \widehat{Q})^\intercal.
\]

Finally, we propose our DAG-informed $Z$ test statistic: 
\begin{equation}\label{t2dagz}
T_{\text{DAG-}Z} = (2p)^{-1/2}\left(T^2_{\text{DAG-}\chi^2} - p\right),
\end{equation}
where $T^2_{\text{DAG-}\chi^2}$ is a DAG-informed Hotelling's $T^2$-type test statistic,
\begin{equation}\label{add:0326:2}
T^2_{\text{DAG-}\chi^2} = N(\overline{\bf x}^{(1)} - \overline{\bf x}^{(2)})^\intercal \widehat{\Sigma}^{-1}_{\text{DAG}} (\overline{\bf x}^{(1)} - \overline{\bf x}^{(2)}),
\end{equation}
with $N = n_1n_2/(n_1 + n_2)$, which is an alternative test statistic for testing $H_0$ that will be discussed in \textbf{Remark 2}. 
As such, we call the proposed testing procedure T2-DAG hereafter. 

Let $N_e = \sum_{j=1}^p |S_j|$ and $p_0 = \sum_{j=1}^p I(|S_j| > 0)$ denote the number of edges in $\mathcal{D}$ and the number of genes with at least one parent gene, respectively. 
The following result characterizes the asymptotic null distribution of $T_{\text{DAG-}Z}$.
\begin{theorem}\label{thm1}
Suppose Assumptions 1-3 hold.
As $n_1, n_2 \rightarrow \infty$, if $d\log p_0 = O\left(n_1 + n_2\right)$ and $N_e = o\left(\sqrt{p}(n_1 + n_2)\right)$, then under $H_0$, we have
\begin{equation}\label{eq:thm:3}
   T_{\text{DAG-}Z} {\rightarrow} N(0,1) ~ \mbox{ in distribution,  as } ~ n_1, n_2 \rightarrow \infty. 
\end{equation}
\end{theorem}
Theorem \ref{thm1} directly yields an asymptotically valid two-sided p-value for $T_{\text{DAG-}Z}$ for testing $H_0: \ve \mu_1 = \ve \mu_2$,
that is, $p_{\text{DAG-}Z} = \mathbb{P}(|z| > |T_{\text{DAG-}Z}|)$, where $z$ is a random variable following the standard normal distribution. 
The assumptions in Theorem \ref{thm1} often hold for typical gene pathways. In particular, $N_e = o\left(\sqrt{p}(n_1 + n_2)\right)$ implies that the total number of edges in the gene pathway is small compared to the square root of the number of genes multiplied by the total sample size, which is a weak assumption in real applications. See Table S1 for information on the real pathways considered in our data analysis.

To examine the asymptotic power of the proposed test, we first define the following two distance measures between the mean vectors of the two groups: 

 \textit{1.}   Kullback-Leibler divergence: $D_{\text{KL}}(\ve \mu_1,\ve \mu_2)
=  (\ve \mu_1 - \ve \mu_2)^\intercal(I - Q)R^{-1}(I - Q)^\intercal (\ve \mu_1 - \ve \mu_2)$. 

 \textit{2.} Pathway-specific divergence:  $D_{\text{DAG}}(\ve \mu_1,\ve \mu_2) = \sum_{j: |S_j| > 0} (\ve \mu_1 - \ve \mu_2)_{S_j}^\intercal \{\Sigma_{S_j, S_j}\}^{-1} (\ve \mu_1 - \ve \mu_2)_{S_j}$, where $(\ve \mu_1 - \ve \mu_2)_{S_j}$ denotes the sub-vector of $\ve \mu_1 - \ve \mu_2$ with entries indexed by $S_j$, and $\Sigma_{S_j, S_j}$ denotes the sub-matrix of $\Sigma$ with rows and columns both indexed by $S_j$.

The following result characterizes the asymptotic power of the proposed test under local alternatives regarding $D_{\text{KL}}(\ve \mu_1,\ve \mu_2)$ and $D_{\text{DAG}}(\ve \mu_1,\ve \mu_2)$. 

\begin{theorem}\label{thm2}
Suppose all the assumptions in Theorem \ref{thm1} hold. If $D_{\text{DAG}}(\ve \mu_1,\ve \mu_2) \asymp p^\beta$ and $D_{\text{KL}}(\ve \mu_1,\ve \mu_2) \asymp p^{1/2} N^{-\gamma}$  for some $0 < \beta < 1/2$ and $\gamma > 1/2$, 
then for all $\alpha \in (0,1)$,
\[
\lim_{n_1, n_2 \rightarrow \infty} \left[ \mathbb{P}(p_{\text{DAG-}Z} \leq \alpha) - \mathbb{P}( |z| \leq -z_{\alpha/2} + \frac{ND_{\text{KL}}(\ve \mu_1,\ve \mu_2)}{\sqrt{2p}}) \right] \geq 0,
\]
where $z$ is a standard normal random variable. In particular, if $1/2 < \gamma < 1$, then we have  $\lim_{n_1, n_2 \rightarrow \infty} \mathbb{P}(p_{\text{DAG-}Z} \leq \alpha) = 1$.
\end{theorem}

\begin{remark}
The proposed T2-DAG test is built upon the OLS estimators of $Q$ and $R$. 
When the OLS estimators are ill-posed (e.g., when Assumption 2 is violated), one can instead use penalized estimators such as the lasso estimator \citep{tibshirani1996regression} to estimate $Q$ and $R$, and the test statistic can be obtained by replacing $\widehat{Q}$ and $\widehat{R}$ in (\ref{add:0326:2}) with such penalized estimators. 
However, we may need alternative techniques to establish the asymptotic distribution of the resulting test statistic. 
This is because the current proof heavily depends on the theoretical property specific to the OLS estimator; that is, the OLS coefficient estimator is independent of the OLS variance estimator. Please refer to Lemma 1.3 in Section 1 of the supplementary material for more details. 
\end{remark}

\begin{remark}
The statistic $T^2_{\text{DAG-}\chi^2}$ in (\ref{add:0326:2}) defines an alternative test statistic for testing $H_0$. We show in the proof of Theorem 1 (see Section 1 of the supplementary material) that $T^2_{\text{DAG-}\chi^2}$ converges to the chi-squared distribution with $p$ degrees of freedom if
\begin{equation}\label{chi:cond}
p = o(n_1 + n_2), N_e = o(n_1 + n_2), ~\mbox{and } d\log p_0 = O(n_1 + n_2).
\end{equation}
Note that the conditions in (\ref{chi:cond}) are more stringent than those required in Theorem 1 for $T_{\text{DAG-}Z}$; that is,
\begin{equation}\label{z:cond}
N_e = o(\sqrt{p}(n_1 + n_2)) \mbox{ and } d\log p_0 = O(n_1 + n_2). 
\end{equation}
Compared to (\ref{chi:cond}), the conditions in (\ref{z:cond}) allow the number of genes in the pathway $(p)$ to be asymptotically larger than the total sample size $(n_1 + n_2)$, which is theoretically important in high-dimensional statistics. Moreover, the conditions in (\ref{z:cond}) allow the total number of edges in the pathway ($N_e$) to grow with the number of genes ($p$), which is a reasonable situation in real applications, while the conditions in (\ref{chi:cond}) don't.

Due to these theoretical limitations of $T^2_{\text{DAG-}\chi^2}$, we adopt the $Z$-transformed statistic $T_{\text{DAG-}Z}$ in this manuscript. However, as shown in Section 3, $T^2_{\text{DAG-}\chi^2}$ compares favorably to $T_{\text{DAG-}Z}$ in terms of the finite-sample performance under various simulated scenarios: it yields well-controlled type-I error rates and higher powers than $T_{\text{DAG-}Z}$. It would be interesting to further study the gap between the theoretical and numerical properties of the statistic $T^2_{\text{DAG-}\chi^2}$.
\end{remark}

\section{Simulation Studies}\label{sec3}
\subsection{Set-up}\label{sim.sec1}
We illustrate the performance of the proposed T2-DAG test with simulation studies under various settings of sample size ($n_1=n_2$), dimension ($p$), number of children nodes, i.e., genes that have at least one parent gene ($p_0$), and the maximum number of directed edges toward one gene ($d$). 
We applied both $T_{\text{DAG-}Z}$ in (\ref{t2dagz}) and $T^2_{\text{DAG-}\chi^2}$ in ((\ref{add:0326:2})) for the proposed T2-DAG test, and denote the corresponding tests by T2-DAG $Z$ and T2-DAG $\chi^2$, respectively. The data was simulated as follows. 
We generated two groups of observations 
according to
\(
{\bf x}_i^{(l)} = \ve \mu_l + (I - Q^\intercal)^{-1}\ve \epsilon_{i} \mbox{, } l = 1,2\mbox{, } i = 1, \ldots, n_l.
\)
Specifically, we first characterized the support of $Q$ using a  
binary adjacency matrix $A=(a_{ij})_{p \times p}$ generated according to the characteristics (e.g., $p, p_0$, and $d$) of the real KEGG pathways included in our real data analysis (see Table S1 for details). 
We considered two scenarios regarding $|S_j| = \sum_{k=1}^p I(a_{kj} = 1)$ and $p_0 = \sum_{j=1}^p I(|S_j| > 0)$. In the first scenario, we set 
$p_0 = 0.6p$ and
$|S_j| = 1 + \zeta_j$, where $\zeta_j \stackrel{i.i.d.}{\sim} \text{NB}(3,0.6)$ with $\text{NB}(3,0.6)$ denoting the negative binomial distribution assuming the experiment terminates after 3 failures with the probability of success equal to 0.6. 
This resulted in a value of $d$
between 10 and 15. 
This reflects the scenario where there are condensed gene interactions towards a relatively small number of genes.
In the second scenario, 
we set $p_0 = 0.8p$ and $\zeta_j \sim \text{NB}(3,0.8)$ which resulted in a value of $d$ between 3 and 7. This is aligned with real scenarios where gene pathways with larger $p_0$ tend to have smaller $d$, and it reflects the scenario of less condensed gene interactions where more genes in the gene set are affected by other genes (larger $p0$) but each of them are affected by fewer genes (smaller $d$).
Next, we generated an initial coefficient matrix, $Q^0$, with the $(i,j)$-th entry equal to $(-1)^{\upsilon_{ij}}a_{ij}$ where 
$\upsilon_{ij} \stackrel{i.i.d.}{\sim} \text{Bernoulli}(0.5)$. We further rescaled $Q^0$ to obtain the final coefficient matrix, $Q = Q^0/(\kappa \|Q^0\|_2)$ with $\kappa=1.5$. 
We then generated $\bm{\epsilon}_i$ independently from $N(\bm{0},R)$, $i=1,\ldots,n_1+n_2$, where $R=\text{diag}(r_0)$ 
with $r_0 = 0.2$. 
Without loss of generality, we set $\bm{\mu}_1={\bf 0}$ and $\bm{ \mu}_2=(\mu_{2,1},\ldots,\mu_{2,p})^\intercal$ with $\mu_{2,j} = \delta \text{I}(j\leq q)$, where $\delta$ characterizes the signal strength, i.e., the magnitude of differential gene expression, and $q$ controls the proportion of genes in the pathway that express differentially between the two groups.

We compared the proposed T2-DAG test with several existing tests, including the Hotelling's $T^2$ test (\citealp{hotelling1931}), 
Graph T2 test 
\citep{jacob2012}, 
BS \citep{bai1996}, Ch-Q \citep{chen2010}, SK \citep{srivastava2013}, CLX \citep{cai2014}, GCT \citep{gregory2015}, aSPU \citep{xu2016}, and ARHT \citep{li2020adaptable}. 
The Hotelling's $T^2$ test was only applied under the scenarios where $p < n_1 + n_2$. Besides information on the existence/absence of gene interactions, the Graph T2 test can also account for the signs of the interactions, such as activation (+) or inhibition (-). 
We therefore implemented Graph T2 based on top $k$ eigenvectors of the graph Laplacian matrix of the signed graph. 
We set the number of eigenvectors used to $k = 0.2p$ as suggested in \cite{jacob2012}. It has also been noted that Graph T2 is robust to the choice of $k$ \citep{lopes2011more, jacob2012}. 
For GCT, we conducted the moderate-$p$ version of the test given the dimensions of the simulated datasets, and set the lag window size to $L=(2/3)p^{1/2}$ as suggested in \cite{gregory2015}. For aSPU, the index of the sum-of-powers test, $\gamma$, was chosen adaptively from $\{1,2,\ldots,6,\infty\}$ \citep{xu2016}. 
The ARHT test was implemented using a probabilistic prior model specified by the canonical weights, $(1,0,0)$, $(0,1,0)$, and $(0,0,1)$, for $(I_p, \Sigma, \Sigma^2)$, with the type 1 error rate calibrated using cube-root transformation, which are the default settings of ARHT in the ``ARHT'' R package. All tests were conducted under a significance level of $\alpha=0.05$.

\subsection{Results given fully informative pathway information}\label{sim.sec2}
We first implemented the T2-DAG test leveraging complete and correct edge information of the pathways summarized by the true adjacency matrix $A$. 
Fig. \ref{fig.noerror.50} shows the empirical rejection percentages (i.e., type I error rates when $\delta=0$ and powers when $|\delta|>0$) of all the tests when $n_1=n_2=50$ and $p=40$ or $100$. 
Additional simulation results for $n_1=n_2=50$ and $p=20$, $n_1=n_2=100$ and $p=100$ or 300 have similar patterns and are summarized in Figs. S1-S3. 
All tests yield well-controlled type I error rates, except for GCT which has substantial inflation in type I error rates in the lower-dimensional case ($n_1=n_2=50$, $p=40$). {This is possibly because GCT assumes an ordering of the genes such that the dependency between genes can be modeled according to their displacement \citep{gregory2015}, 
whereas in our simulated data we did not introduce this unique structure.} 
This violation may lead to inflated type I error rates especially when $p$ is small. \begin{figure}
\centering
    \includegraphics[width=\textwidth]{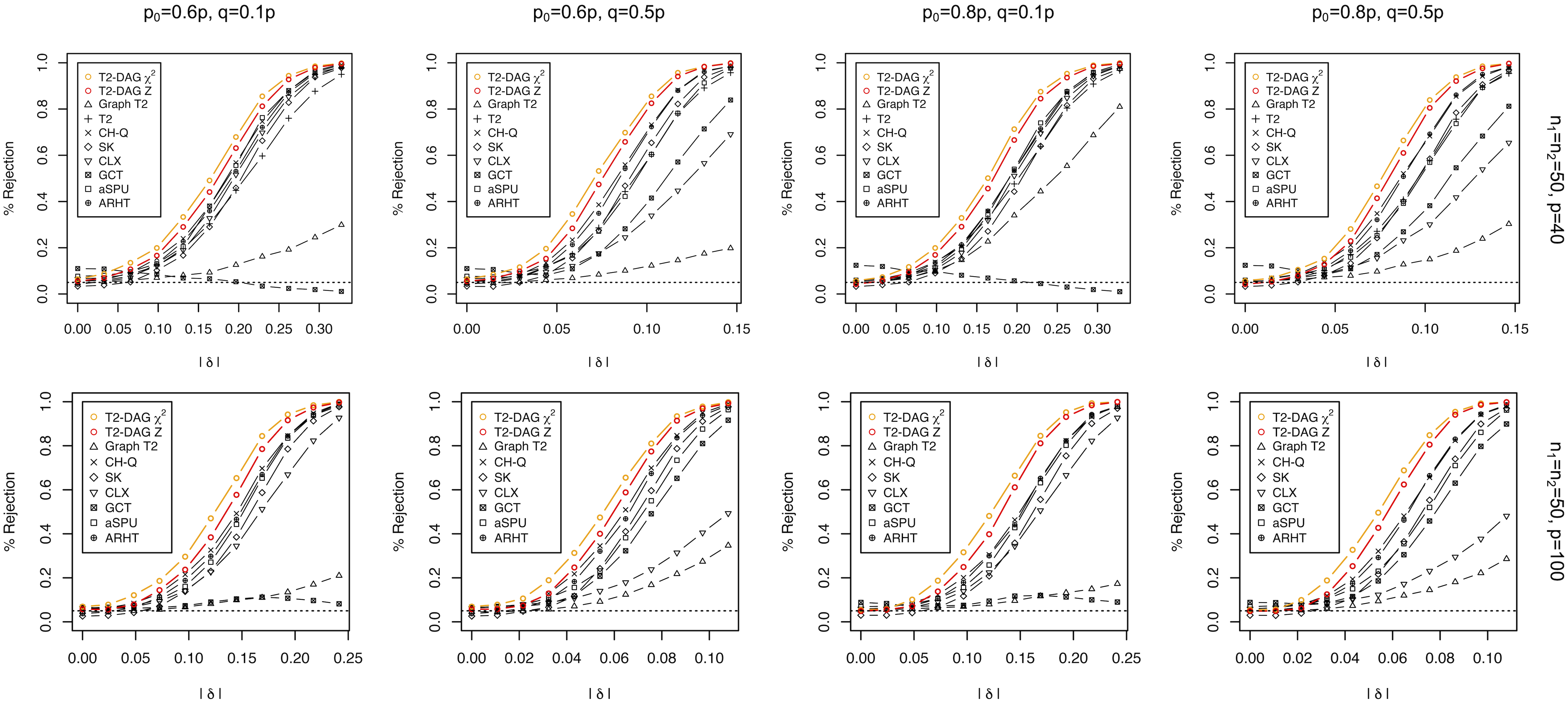}
	\caption{Empirical type I error rates ($\delta = 0$) and powers ($|\delta| > 0)$ under $n_1=n_2=50$ and $p=40$ (top row) or 100 (bottom row), with the number of children nodes in the DAG set to $p_0=0.6p$ (columns 1 and 2) or $p_0=0.8p$ (columns 3 and 4), and the number of non-zero signals set to $q = 0.1p$ (columns 1 and 3) or $0.5p$ (columns 2 and 4).
}
\label{fig.noerror.50}
\end{figure}

Among the existing tests, CH-Q gives the best overall performance with a power higher than or equal to that of the Hotelling's $T^2$ test (when $n_1 + n_2 > p$). 
The ARHT test shows a power that is similar to or slightly lower than that of the CH-Q test but higher than that of all other existing tests under all simulated scenarios, except for one scenario where aSPU has a slightly higher power (see Fig. \ref{fig.noerror.50}, row 1, column 3). 
The Graph T2 test shows relatively low power in all scenarios. This is possibly because the simulated signal distribution is less coherent with the top eigenvectors of the graph Laplacian matrix.
The power of the supremum-based CLX test is competitive to that of CH-Q and SK in sparse-signal scenarios (Fig. \ref{fig.noerror.50}, columns 1 and 3), but is outperformed by most of the other tests under dense signals (Fig. \ref{fig.noerror.50}, columns 2 and 4). 
The GCT test has an extremely low power in the sparse-signal scenarios (Fig. \ref{fig.noerror.50}, columns 1 and 3), where the power is below 0.15 and even decreases as signal strength ($|\delta|$) increases.
The proposed T2-DAG tests have substantially higher powers compared to the other tests in all simulated scenarios, except when signals are weak such that all tests have powers below 0.15. 
Compared to the T2-DAG $Z$ test that has a type I error rate strictly controlled at around $\alpha=0.05$, T2-DAG $\chi^2$ has a slightly higher but still well-controlled type I error rate and a higher power under all simulated scenarios.
Furthermore, 
compared to the scenarios of $p_0=0.6p$ where there are condensed gene interactions towards a relatively small number of genes (Fig. \ref{fig.noerror.50}, columns 1-2), in the scenarios of $p_0=0.8p$ where there are scattered gene interactions towards a larger number of genes (Fig. \ref{fig.noerror.50}, columns 3-4), the two T2-DAG tests show more obvious advantages compared to the existing tests with strict type I error control and more improvement in power.

\subsection{Results given partially incorrect edge information}\label{sim.sec3}
We next discuss a more realistic scenario where the auxiliary edge information comes with errors, such as having missing and/or redundant edges. 
Recall that $N_e = \sum_{i=1}^p |S_j|$ denotes the total number of edges in the gene pathway. 
{Besides} {the true adjacency matrix $A$}, we consider two types of mis-specified auxiliary DAGs with the following adjacency matrices: (1) $A_1$ with a random set of $0.4N_e$ edges missing, and (2) $A_2$ with a random set of $0.4N_e$ redundant edges that do not exist in $A$. 
We implemented the proposed T2-DAG test based on $A$, $A_1$ or $A_2$, 
and implemented Graph T2 only based on the signed graph corresponding to the true adjacency matrix $A$.
\begin{figure}[ht!]
\centering
  \includegraphics[width=0.7\textwidth]{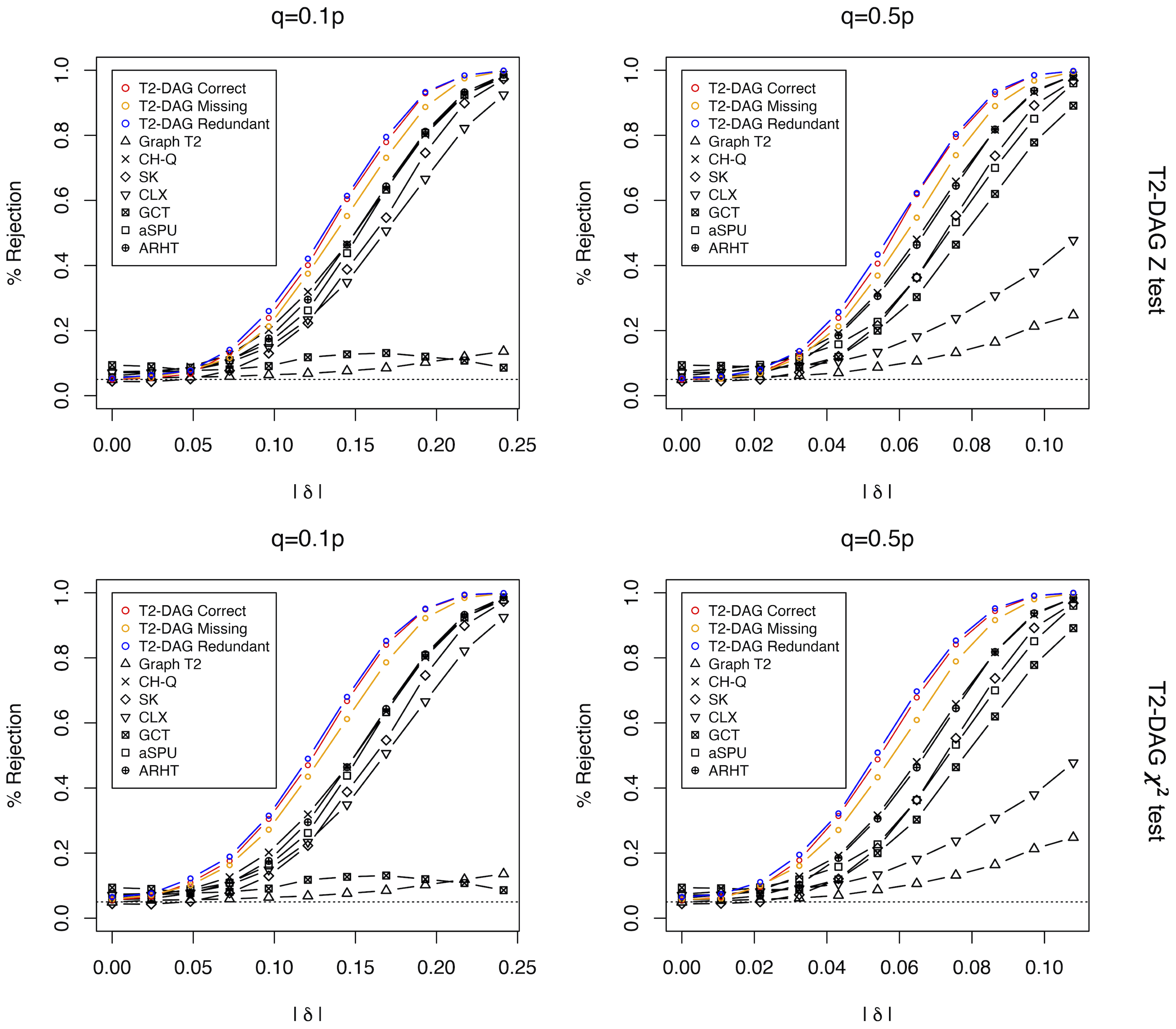}
	\caption{Empirical type I error rates ($\delta = 0$) and powers ($|\delta| > 0)$ of T2-DAG $Z$ (top panel) and T2-DAG $\chi^2$ (bottom panel) under $n_1=n_2=50$ and $p=100$, with the number of children nodes set to $p_0=0.8p$ and the number of non-zero signals set to $q = 0.1p$ (column 1) or $0.5p$ (column 2).  
	T2-DAG was applied based on $A$ (``T2-DAG Correct''), $A_1$ (``T2-DAG Missing''), or $A_2$ (``T2-DAG Redundant''). Graph T2 was applied based on the true adjacency matrix $A$.
\label{fig.error.50.100}}
\end{figure}

Fig. \ref{fig.error.50.100} shows the empirical type I error rates and powers of all the tests given $n_1=n_2=50$ and $p=100$.
Additional simulation results for $n_1=n_2=50$ and $p=50$, $n_1=n_2=100$ and $p=100$ or 300 
are summarized in Figs. S4-S6. 
Although having incorporated the correct edge information as well as the signs of interactions, Graph T2 still has a much lower power compared to T2-DAG. 
Among the other existing tests that do not incorporate external graphical information, the ARHT test shows the highest power which is almost the same as that of CH-Q under relatively weak signals ($|\delta|$) and slightly higher than that of CH-Q under stronger signals. 
Similar to what we observed in Section \ref{sim.sec2}, T2-DAG $\chi^2$ has well-controlled but slightly higher type I error rates and substantially higher power than T2-DAG $Z$. 
Using either of the two test statistics, the T2-DAG test based on pathway information with redundant edges (T2-DAG redundant) has slightly higher inflation in type I error rates, and the one based on partially missing edge information (T2-DAG Missing) has slightly compromised powers. 
Nonetheless, the overall performance of T2-DAG is similar across the three scenarios. 
This indicates the robustness of the proposed T2-DAG test with respect to potential errors in the auxiliary pathway information.

\begin{figure}[ht!]
\centering
    \includegraphics[width=\textwidth]{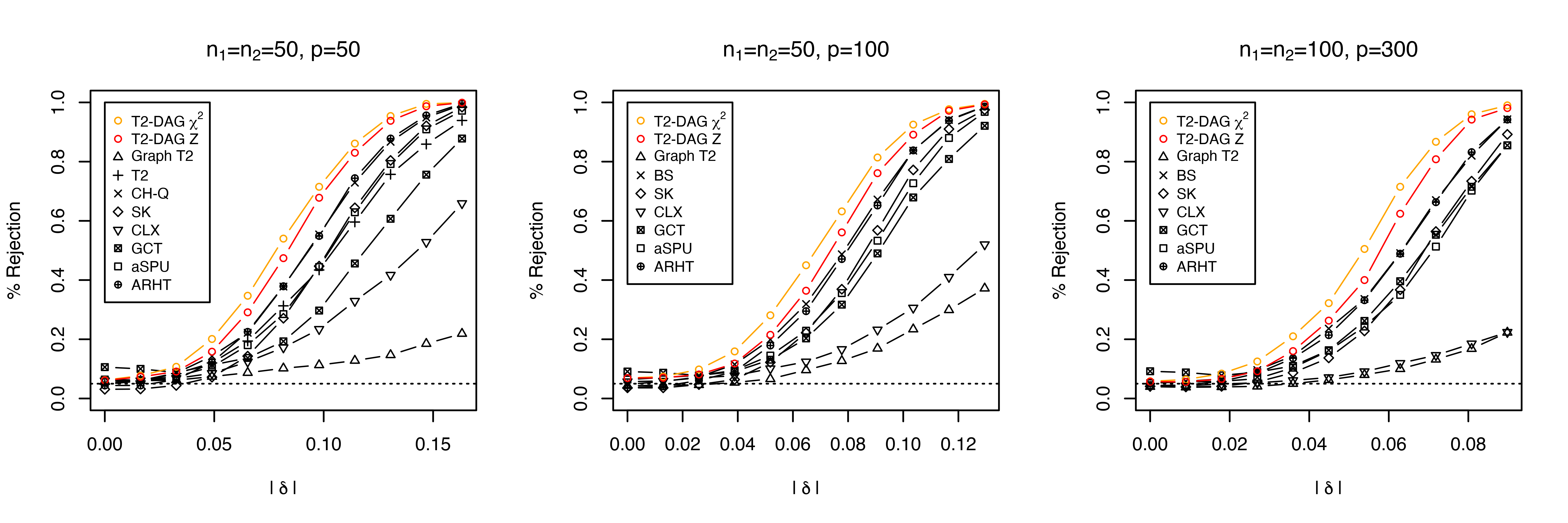}
\caption{Empirical type I error rates ($\delta = 0$) and powers ($|\delta| > 0)$ of the various test assuming two continuous confounders, with the number of children nodes in the DAG set to $p_0=0.8p$, and the number of non-zero signals set to $q = 0.5p$. Please note that the Hotelling's $T^2$ test (T2) is not presented in columns 2 and 3 because it is not applicable when $n_1+n_2 \leq p$.}
\label{fig.confounder}
\end{figure}

\subsection{Results in the presence of unadjusted confounding effects}\label{sim.sec4}
So far, we have considered independent error terms $\epsilon_j: j = 1, \ldots, p$, with a diagonal covariance matrix $R$. 
However, when there exist latent confounders affecting gene interactions, {the error terms} may be correlated, leading to a non-diagonal $R$. 
To evaluate the sensitivity of T2-DAG with respect to the mis-specification of $R$, we conducted additional simulations with two continuous confounders. 
Specifically, we generated two groups of gene expression data according to ${\bf x}_i^{(l)} = (I-Q^\intercal)^{-1}(C^\intercal\bm{w}_i + {\ve \epsilon}_i)+\bm{\mu}_l$, $l=1, 2$, $i = 1, \ldots, n_l$. 
Here, $\bm{\mu}_1$, $\bm{\mu}_2$, $Q$ and ${\ve \epsilon}_i$ were generated following the same procedure described in Section \ref{sim.sec1}. The confounding variables $\bm{w}_i = (w_{i,1},w_{i,2})^\intercal$, $i=1,\ldots,n_1+n_2$, were generated independently from $\mathcal{N} (0,\Sigma_{w})$ with $\Sigma_{w}=\text{diag}((0.2)^2 /\|Q\|_{\max})$. Twenty percent of the entries in each row of $C$ were generated independently from $\mathcal{N}(0,\sqrt{0.2})$, and the remaining entries were set to 0. As such, 
a random set of $0.2p$ genes were conditionally associated with each confounder. 
For T2-DAG test, we assumed that the confounders $\ve w_i$ were latent, and T2-DAG was still implemented based on model (\ref{test:1}). 

We observe from Fig. \ref{fig.confounder} that all tests have well-controlled type I error rates expect for GCT, which shows high inflation in type I error rates. 
When the signals are not too weak (i.e., $|\delta| > 0.05$), 
T2-DAG is more powerful than all other tests implemented. 
Additional simulation results based on $n_1=n_2=50$, $p=50$ and $n_1=n_2=100$, $p=100$  
have similar patterns and are summarized in Fig. S7. 
These results show that T2-DAG can still outperform the existing tests even with unadjusted confounders.

\begin{figure}[ht!]
\centering
    \includegraphics[width=\textwidth]{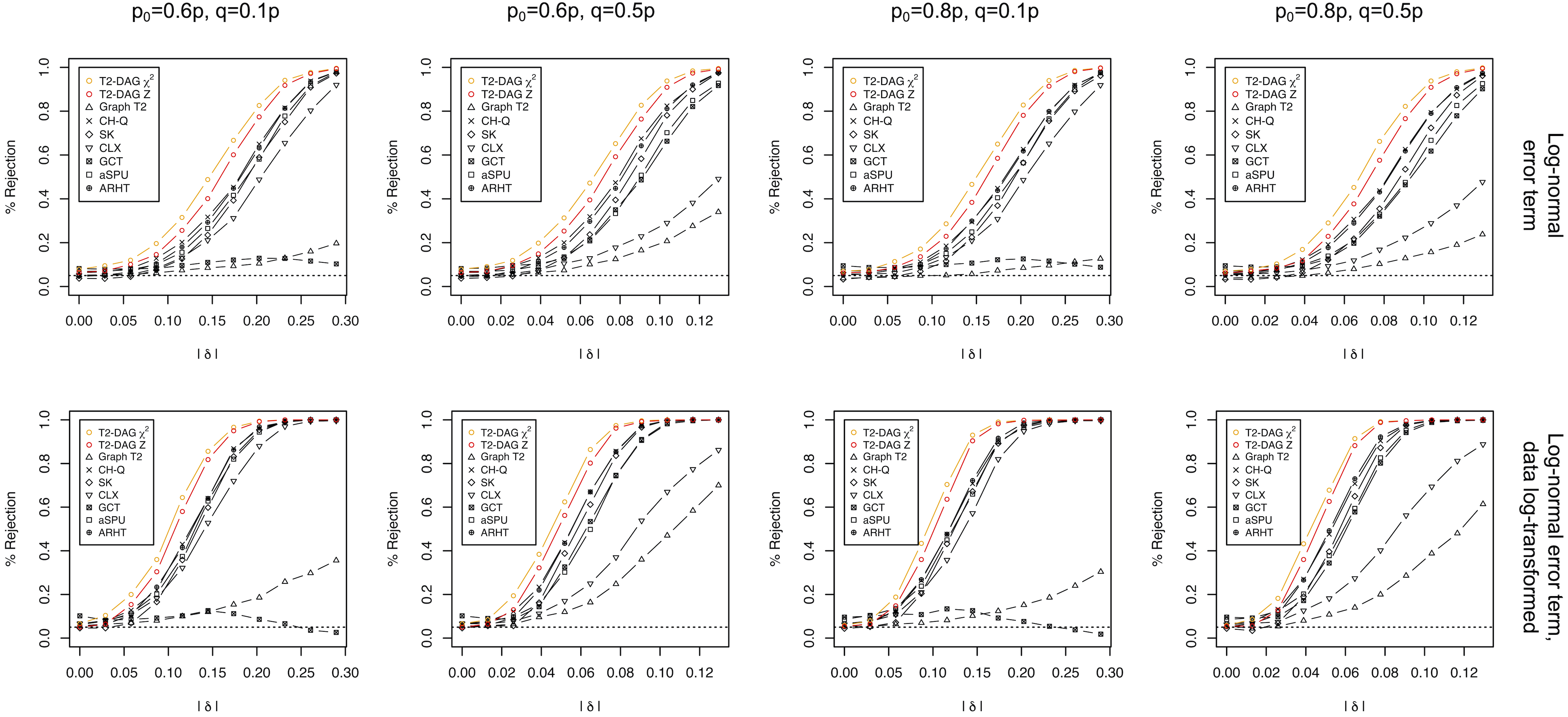}
\caption{Empirical type I error rates ($\delta = 0$) and powers ($|\delta| > 0)$ of the various tests under ${n_1=n_2=50}$ and ${p=100}$ with error terms generated from log-normal distributions based on the raw data (row 1) or log-transformed data (row 2).
The number of children nodes in the DAG is set to $p_0=0.6p$ (columns 1 and 2) or $p_0=0.8p$ (columns 3 and 4), and the number of non-zero signals is set to $q = 0.1p$ (columns 1 and 3) or $0.5p$ (columns 2 and 4). Additional results under different settings of $n_1$, $n_2$, $p$ and results assuming error terms of gamma or uniform distributions are summarized in Figs. S8 - S18.}
\label{lognormal}
\end{figure}

\subsection{Results given non-Gaussian error terms}\label{sim.sec5}
The theory of T2-DAG requires the assumption of Gaussian error terms, which can be violated in reality. 
We therefore conducted an additional sensitivity analysis on T2-DAG where we considered three alternative distributions: log-normal distribution, gamma distribution, and uniform distribution. Detailed simulation settings are summarized in Section 2.4 of the supplementary material. We can observe from the first row of Fig. \ref{lognormal} and Figs. S8 - S18 that T2-DAG tests perform similarly as in Section \ref{sim.sec2} with well-controlled type I error rates and higher powers than the existing tests, showing the robustness of T2-DAG to these distributions of the error terms. We further conducted a log transformation on the data under the scenarios of log-normal and gamma error terms where the simulated gene expression data were right skewed. Interestingly, this additional data transformation step 
lead to a substantial gain in power for all tests including T2-DAG in all scenarios, except for GCT in a few scenarios where it had powers lower than 0.15 under all signal strengths (Figs. \ref{lognormal} and S8 - S14, bottom row). Furthermore, while there is no obvious change in the type I error rates of the existing tests, T2-DAG consistently has a similar or more strict type I error control after the log transformation. These findings suggest that for right skewed gene expression data which is commonly seen in genetics studies we can conduct a log transformation 
to improve the performance of T2-DAG.

\section{A real data example}\label{sec4}
In lung cancer diagnosis and prognosis, extensive research has been conducted utilizing gene expression profiling to identify lung cancer-related gene pathways for therapeutic research \citep{weiss2008,shao2010,cai2014revealing,dancik2014robust,chang2015impact,long2019identification,venugopal2019differences,kong2020dynamical}. 
In this application, we aim to identify differentially expressed gene pathways between different stages of lung cancer, the findings of which could provide insights into the underlying genetic mechanisms of lung cancer progression.

We applied T2-DAG $Z$ and T2-DAG $\chi^2$, as well as the existing tests considered in the simulation studies, to a gene expression dataset for lung cancer patients in The Cancer Genome Atlas (TCGA) Program \citep{cancer2014}. The dataset was obtained from the Lung Cancer Explorer (LCE) \citep{LCE} with standard quality control, where the expression values of zero were set to the overall minimum value and all data were then $\log_2$-transformed so that the gene expression data are approximately normal for all lung cancer stages \citep{cai2014,cai2019}. 
Expression levels of 20429 genes were measured for lung cancer tissue samples 
collected from 513 patients, among whom $N_1=278$, $N_2=124$, $N_3=84$ and $N_4=27$ were diagnosed with stage I, II, III and IV lung cancer, respectively. 
Additionally, $N_0=59$ normal tissue samples were collected from 59 of the 513 patients. When comparing cancer tissues with normal tissues, we excluded the tumor samples of these 59 patients to avoid dependent samples.

We considered $H=206$ common human pathways related to signaling, metabolic and human diseases from the KEGG database \citep{kanehisa2000,kanehisa2019}. Pathway information on previously identified gene interactions 
is categorized into activation/expression or inhibition/repression of one gene by another \citep{kanehisa2000,kanehisa2019}; see Fig. 1 for an illustrative example. Detailed characteristics of the pathways are summarized in Table S1. Among the 206 pathways, 34 have circles (including self loops) in the corresponding graphs, among which 18 has 1 circle only, 8 have 2 circles, 4 have 3 circles, and 4 have more than 3 circles each. 
We removed these circles  when implementing the proposed T2-DAG test but used the complete graph information when implementing the Graph T2 test. Other than this, all the tests were implemented in the same way as in Section \ref{sec3}. 
We further applied Bonferroni correction \citep{bonferroni1936teoria} with a significance level $\alpha=\alpha_0/H$ for each test to control the family-wise error rate (FWER) at $\alpha_0=0.05$.

\begin{table}[ht!]
\footnotesize{}
\def~{\hphantom{0}}
\centering
\caption{Results of the hypothesis tests on 206 KEGG pathways for differential gene expression between stage I and stage III lung cancer (top panel, $N_1=278$, $N_3=84$) and between stages II and stage III lung cancer (bottom panel, $N_2=124$, $N_3=84$). 
In each panel, the numbers on the diagonal line are the numbers of significant pathways identified by the corresponding methods; the numbers in the $(i,j)$-th cell in the upper and lower triangle are the number of significant pathways detected by both method $i$ and method $j$ and the estimated correlation of $-{\log}_{10}(p)$ between the two methods, respectively. ``T2'' denotes the Hotelling's $T^2$ test.
\label{table.I.III}}
\begin{tabular}{lcccccccccc}
\multicolumn{11}{c}{Stage I versus Stage III}\\
& T2-DAG $\chi^2$ & T2-DAG $Z$ & Graph T2 & T2 & CH-Q & SK & CLX & GCT & aSPU & ARHT\\
T2-DAG $\chi^2$ & 127 & 127 &  10 &   6 &  90 &  59 & 1 &  79 &  84 & 37\\ 
T2-DAG $Z$ &0.98 & 144 &  10 &   8 &  93 &  59 &   1 &  82 &  88 & 37\\
Graph T2 & 0.43 & 0.45 &   10 &   2 &  8 &  8 &   0 &  7 &  8 & 5\\ 
T2  & 0.23 & 0.23 & 0.54 &  8 &   5 &   4 &   0 &   3 &   7 & 6\\ 
CH-Q   & 0.57 & 0.59 & 0.27 & 0.19 & 96 &  55 &   1 &  64 &  81 & 36\\ 
SK   & 0.75 & 0.77 & 0.48 & 0.43 & 0.58 & 59 &   0 &  41 &  50 & 28\\
CLX   & 0.54 & 0.56 & 0.42 & 0.34 & 0.54 & 0.60 & 1 &   0 &  1 & 1\\
GCT   & 0.81 & 0.76 & 0.40 & 0.17 & 0.36 & 0.63 & 0.29 & 82 &  59 & 27\\  
aSPU   & 0.51 & 0.55 & 0.38 & 0.27 & 0.68 & 0.56 & 0.58 & 0.32 &  95 & 35\\
ARHT   & 0.56 & 0.56 & 0.26 & 0.20 & 0.70 & 0.48 & 0.46 & 0.37 & 0.57 & 37 \\\\
\multicolumn{11}{c}{Stage II versus Stage III}\\
& T2-DAG $\chi^2$ & T2-DAG $Z$ & Graph T2 & T2 & CH-Q & SK & CLX & GCT & aSPU & ARHT\\ 
  T2-DAG $\chi^2$ & 0 & 0 & 0 & 0 & 0 & 0 & 0 & 0 & 0 & 0\\ 
  T2-DAG $Z$ &-0.21 & 0 & 0 & 0 & 0 & 0 & 0 & 0 & 0 & 0\\ 
  Graph T2 & 0.42 & -0.24 & 0 & 0 & 0 &  0 &  0 & 0 &  0 & 0\\ 
  T2 & 0.34 & -0.30 & 0.54 & 0 & 0 &  0 &   0 &   0 &  0 & 0\\ 
  CH-Q & 0.48 & -0.27 & 0.23 & 0.21 & 0 & 0 &  0 & 0 &  0 & 0\\ 
  SK & -0.27 & 0.84 & -0.31 & -0.35 & -0.34 & 0 & 0 & 0 & 0 & 0\\ 
  CLX & 0.63 & -0.24 & 0.49 & 0.42 & 0.38 & -0.26 & 0 & 0 & 0 & 0\\ 
  GCT & -0.20 & 0.49 & -0.21 & -0.21 & -0.18 & 0.58 & -0.22 & 21 & 0 & 0\\ 
  aSPU & 0.43 & -0.17 & 0.20 & 0.17 & 0.68 & -0.20 & 0.44 & -0.15 & 1 & 0\\
  ARHT & 0.33 & -0.31 & 0.53 & 0.94 & 0.41 & -0.35 & 0.44 & -0.20 & 0.30 & 0\\
\end{tabular}
\end{table}

All tests except aSPU conclude that all 206 gene pathways are differentially expressed between normal tissues and tissues of each cancer stage (Tables S7 - S10). Table \ref{table.I.III} summarizes the number of differentially expressed gene pathways detected by each test between stage I and stage III lung cancer and between stage II and stage III lung cancer. 
For the comparison between stage I and stage III lung cancer, the sample sizes ($N_1=278$, $N_2=124$) are large relative to the corresponding $p$, $p_0$ and $d$ for all pathways. 
The Hotelling's $T^2$ test and Graph T2 test identify 8 and 10 significant pathways, respectively, while
CH-Q, aSPU, GCT, SK and ARHT identify much more pathways (i.e., 96, 95, 82, 59 and 37, respectively).
On the other hand, the CLX test, which is powerful for detecting sparse signals, only identifies 1 significant pathway. This may indicate that the signals of differential gene expression in real pathways tend to be dense, i.e., 
scatter among many correlated genes. 
The existing tests in total identify 129 unique pathways,
while the T2-DAG $\chi^2$ test and $Z$ test identify 127 and 144 significant pathways, respectively, which cover 115 and 122, respectively, of the 129 pathways identified by the existing tests. 
Furthermore, the $-{\log}_{10}$-transformed p-values ($-{\log}_{10}p$) of the T2-DAG tests are positively correlated with those of all the other tests; in particular, results of T2-DAG tests are in line with those of CH-Q, SK, CLX, GCT, aSPU and ARHT, with correlations higher than 0.5. 
The two T2-DAG tests yield highly correlated results (the correlation of $-{\log}_{10}p$ equals 0.98). T2-DAG $Z$ detects a slightly larger number of significant pathways than T2-DAG $\chi^2$, and for the pathways detected by both tests, T2-DAG $Z$ tends to give smaller p-values than T2-DAG $\chi^2$. This is different from what was observed from simulations where T2-DAG $Z$ tends to have lower rejection rates than T2-DAG $\chi^2$, which could be due to the difference in data structures between the simulated data and real data. This inconsistency in relative performance is observed not only for the two T2-DAG tests, but also for other tests such as CH-Q and ARHT, which perform similarly on the simulated data but yield quite different results on the lung cancer data where they detect 96 and 37 differentially expressed pathways, respectively, between stage I and stage III lung cancer.

For the comparison between stage II and stage III lung cancer, all tests except GCT and aSPU detect no differentially expressed gene pathway. 
This may be due to the limited sample sizes ($N_3 = 124$, $N_4 = 84$) and/or relatively small differences between stage II and stage III lung cancer. 
On the other hand, aSPU identifies 1 significant pathway, while GCT identifies 21 significant pathways that do not include the one identified by aSPU. These 21 pathways are possibly false positives because of two reasons. First, none of the pathways is detected by any other test. Second, in our simulation studies, 
GCT could have substantially inflated type I error rates under small $p$ settings (Fig. \ref{fig.noerror.50}, Fig. \ref{fig.confounder}, Fig. S1 and Fig. S7). 

\begin{figure}[ht!]
\centering
	\includegraphics[width=\textwidth]{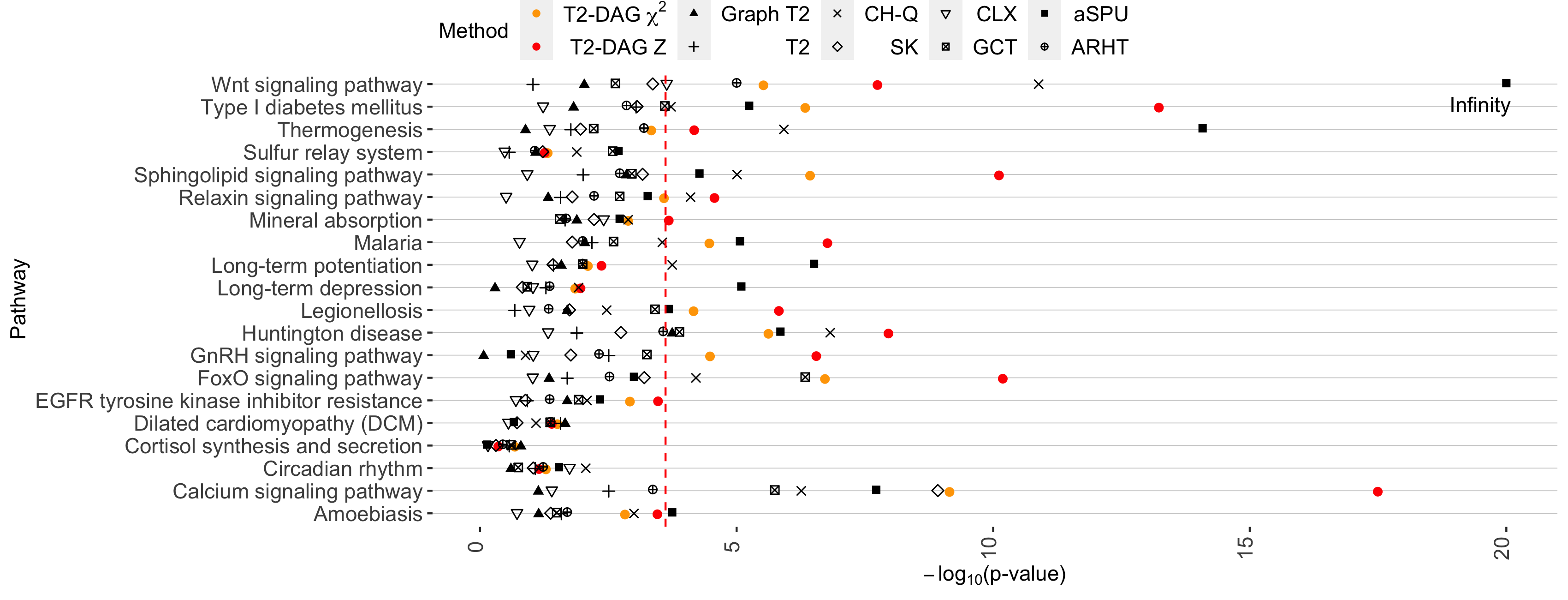}
	\caption{$-{\log}_{10}(p)$ of the various tests for differential expression between stage I and stage III lung cancer on 20 randomly selected human pathways. 
\label{fig.I.III}}
\end{figure}

To take a closer look at these results, we randomly selected 20 of the 206 pathways and report the corresponding $-{\log}_{10}(p)$s (Fig. \ref{fig.I.III}: stage I versus stage III; Fig. S18: stage II versus stage III).
Overall, the p-values for comparison between stage I and stage III lung cancer show consistency across different tests. 
Specifically, for pathways on which all tests yield insignificant results with large p-values, the T2-DAG tests give similar and sometimes even larger p-values compared to the other tests; see the Sulfur relay system pathway and the Circadian rhythm pathway in Fig. \ref{fig.I.III}. For pathways on which T2-DAG yields highly significant results with the smallest p-values, many of the other tests also yield significant results; see the Calcium signaling pathway and the Type I diabetes mellitus in Fig. \ref{fig.I.III}. 
On the contrary, for the pathways identified by GCT to be significantly differentially expressed between stage II and stage III cancer, all the other tests yield insignificant results with large $p$-values (Fig.  \ref{fig.II.III.significant}).

We also compared other pairs of cancer stages and summarized the results in Section 3 of the supplementary material. Results for the comparison of stage I versus II, stage I versus IV, and stage II versus IV cancer have similar patterns as the comparison between stage I and III cancer. Similar to the comparison between stage II and III cancer, when comparing stage III to stage IV cancer, none of the tests identifies any significant pathway except for GCT which identifies 5 significant pathways. 
Detailed results for all the tests are summarized in Tables S1 - S10.

\begin{figure}[ht!]
\centering
	\includegraphics[width=\textwidth]{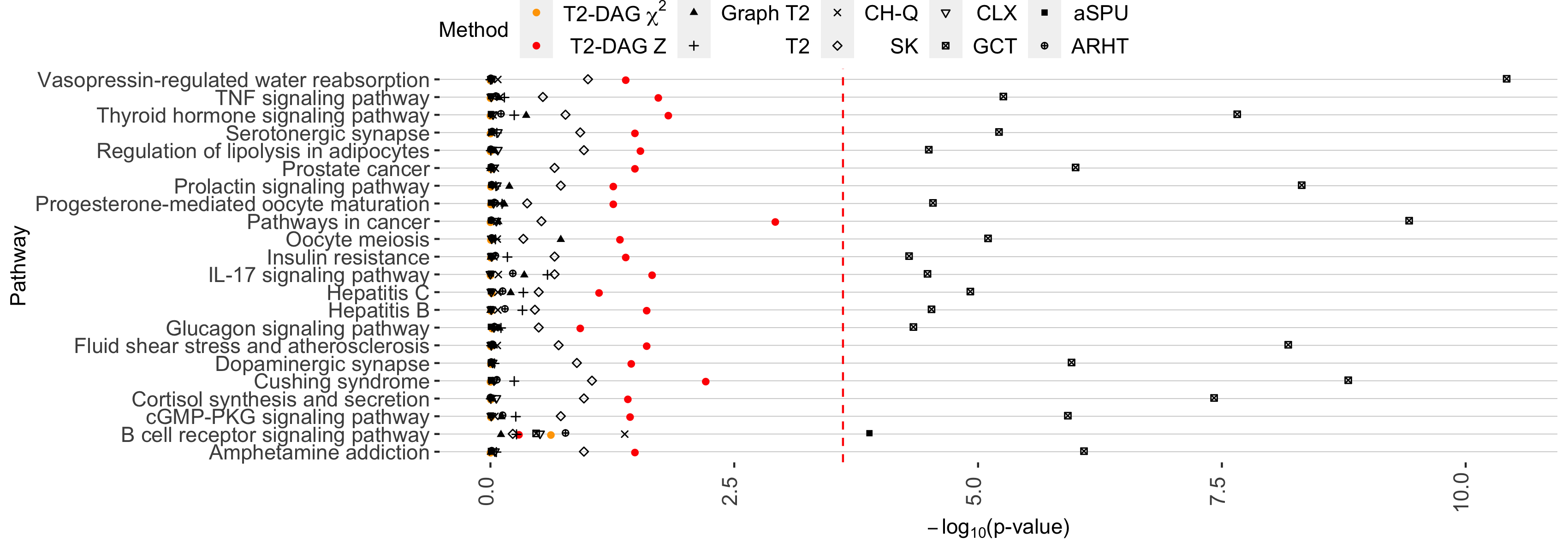}
	\caption{The $-{\log}_{10}(p)$ of the comparison between stage II and stage III lung cancer ($N_2=124$, $N_3=84$) for the pathways identified by any test to be significantly differentially expressed between the two cancer stages.
\label{fig.II.III.significant}}
\end{figure}

\section{Discussion}\label{sec5}
The proposed T2-DAG $Z$ test and the alternative T2-DAG $\chi^2$ test have illustrated their advantages over the existing tests through our simulation studies and application to a lung cancer dataset. 
When establishing asymptotic properties, the T2-DAG $Z$ test is favored over the alternative T2-DAG $\chi^2$ test since it requires less stringent conditions and thus is presented as our main result. 
Although facing theoretical limitations, the alternative T2-DAG $\chi^2$ test has excellent finite-sample performance where it has slightly higher but still well-controlled type I error rates and higher powers than T2-DAG $Z$, as shown by simulations. In the lung cancer application, however, T2-DAG $\chi^2$ appears to be slightly more conservative than T2-DAG $Z$ but still identifies a much larger number of differentially expressed pathways than the existing tests for stage I versus stage II, III, or IV cancer. 
In applications, the choice between T2-DAG $Z$ and T2-DAG $\chi^2$ can be made based on the theoretical properties and simulated finite-sample results under the specific data scenarios of sample size ($n_1$ and $n_2$), dimension ($p$), and sparsity of gene interactions ($p_0$ and $d$) inferred from the auxiliary pathway information.

The proposed T2-DAG test was motivated by gene pathway analyses,  
but it is generally applicable to problems in many other areas 
where similar auxiliary graph information is available. For example, in human microbiome studies for identifying host-microbe associations, the auxiliary graph information is provided by the phylogenetic tree that characterizes the evolutionary relationships among the microbes. 


We have illustrated the robustness of T2-DAG with respective to model mis-specifications due to missing/inaccurate edge information or unadjusted confounding effects. 
While our simulation results are encouraging in these scenarios, the current estimation procedure may no longer be optimal. Thus, we consider the following two extensions of T2-DAG as our future work. 
First, while T2-DAG requires that the auxiliary pathway information can be represented by a DAG, 
it is feasible to extend it to
non-DAG pathways with more complex gene interactions. 
For example, the recently proposed linear SEM for cyclic mixed graphs \citep{amendola2020structure} shows a possible way to account for feedback loops and undirected edges in the gene pathway.
Second, the proposed test assumes a diagonal covariance matrix $R$ for the error term $\ve \epsilon$ in (\ref{test:1}). As discussed in Section 3.4, this assumption may fail due to latent confounders. 
In this case, we may consider a non-diagonal but sparse $R$,
i.e., only a small number of off-diagonal entries of $R$ are non-zero, and estimate $R$ via graphical lasso \citep{friedman2008sparse} or other penalized estimation approaches for high-dimensional precision matrix  \citep{zhang2014sparse,kuismin2017precision,fan2019precision}. Furthermore, when measurements are available for some potential confounders, we can directly incorporate them into the SEM as covariates. We have in fact observed promising performance of this strategy through additional simulation studies, but the theoretical properties of the strategy require further investigation.

While the current T2-DAG test only incorporates information on the presence/absence of gene interactions, it is possible to further incorporate the signs of the interactions, such as activation ($+$) or inhibition ($-$) of one gene by another. 
One possible strategy is to conduct modeling in a Bayesian framework, where the signs of gene interactions can be accounted for by inducing truncated normal priors on the corresponding parameters in the coefficient matrix $Q$. The asymptotic properties of the test, however, need to be re-investigated.

Finally, we note that the key component of the proposed T2-DAG test is the novel DAG-informed estimator for the precision matrix (see Equations (\ref{add:0325:1}) and (\ref{test:3})),  
which could provide substantially improved power through incorporating external pathway information that describes gene interactions. 
Besides being incorporated into our T2-DAG test, this novel estimator has plenty of other potential usage. For example, as illustrated in Section \ref{sec4}, T2-DAG, with a Hotelling's $T^2$-type test statistic, is powerful against ``dense'' alternatives which are usually 
true for differentially expressed gene pathways. 
In applications where the signals are sparse, one may instead consider combining the proposed DAG-informed estimator 
with tests that are powerful against ``sparse'' alternatives, such as the CLX test \citep{cai2014}, to further boost power. Another example is that the proposed DAG-informed estimator 
can be used for high-dimensional linear discriminant analysis (LDA) \citep{bouveyron2007high} to improve the classification accuracy. Specifically, consider the Gaussian case where one wants to classify a random vector $\bf x$ drawn with equal probability from one of the two Gaussian distributions, $\mathcal{N}(\ve \mu_1, \Sigma)$ (class 1) and $\mathcal{N}(\ve \mu_2, \Sigma)$ (class 2). 
When $\ve \mu_1, \ve \mu_2$ and $\Sigma$ are known, 
Fisher's LDA rule, given by
\[
C(\boldsymbol{x})=\Bigg\{\begin{array}{ll}
1, & ~~~\boldsymbol{\delta}^{\top} \Omega\left(\boldsymbol{x}-\frac{\boldsymbol{\mu}_{1}+\boldsymbol{\mu}_{2}}{2}\right)<0 \\
2, & ~~~\boldsymbol{\delta}^{\top} \Omega\left(\boldsymbol{x}-\frac{\boldsymbol{\mu}_{1}+\boldsymbol{\mu}_{2}}{2}\right) \geq 0,
\end{array}
\]
where $\ve \delta = \ve \mu_2 - \ve \mu_1$, and $\Omega = \Sigma^{-1}$ denotes the precision matrix, is known to be optimal. In practice, one needs to estimate $\ve \mu_1, \ve \mu_2$, and $\Omega$. The proposed DAG-informed estimator of $\Omega$ can thus be applied here to facilitate the classification.

\section*{Conflicts of Interest}
The authors declare no conflict of interest.

\section*{Data Availability}
The data underlying this article are available at https://github.com/Jin93/T2DAG.

\section*{Supplementary Material}\label{SM}
Supplementary material available at \textit{Bioinformatics} online includes technical results and proofs, detailed information on the pathways and gene expression data used in the data analysis, detailed simulation settings in the scenario of non-Gaussian error terms, as well as additional results on the simulated and real datasets. 


\section*{Software}
The R (R Development Core Team, 2021) package T2DAG which implements the proposed T2-DAG test is available on Github at
https://github.com/Jin93/T2DAG.

\nolinenumbers

\bibliographystyle{abbrvnat}
\bibliography{c-ref}

\end{document}


\maketitle
\doublespacing
\pagestyle{plain} 

\section{Proof of Theorems}\label{supp1}
In this proof, we consider the high-dimensional regime where $(p, d, p_0) \rightarrow \infty$ as $(n_1, n_2) \rightarrow \infty$. 
We begin with some additional notations. For a $p$-dimensional random vector ${\bf v} = \left( v _ { 1 } , \dots , v _ { p } \right) ^ { \intercal }$, 
we define the following vector norms,
\[
\| {\bf v} \| _ { q } : = \left( \sum _ { j = 1 } ^ { p } \left| v _ { j } \right| ^ { q } \right) ^ { 1 / q } \mbox{ for } 0 < q < \infty , \quad \| {\bf v} \| _ { \infty } : = \max _ { 1 \leq j \leq p } \left| v _ { j } \right|.\]
For a $p \times p$ matrix $M = (m_{jk})_{j,k = 1, \ldots, p}$, we define
\[\| M \| _ { q } : = \max _ { \| v \| _ { q } = 1 } \| M v \| _ { q } , \quad \| M \| _ { \max } : = \max _ { (j, k) } \left| m _ { j k } \right| ,\quad \| M \| _ { \mathrm { F } } : = \left( \sum _ { j , k } \left| m _ { j k } \right| ^ { 2 } \right) ^ { 1 / 2 },\]
for $0 < q < \infty$. 
We use $I_p$ to denote a $p\times p$ identity matrix. For any index set $S \subseteq \{1, \ldots, p\}$, let ${\bf v}_{S}$ denote the subvector of $\bf v$ with entries indexed by $S$. Similarly, $M_{S, S}$ denotes the submatrix of $M$ with rows and columns indexed by $S$.

We first cite the following result (Theorem 6.1, \citealp{wainwright2019high}).
\begin{lemma} \label{lemma:wainwright}
Let ${\bf x}_1, \ldots, {\bf x}_n$ be $p \times 1$ random vectors drawn {\it i.i.d} from a multivariate normal distribution with mean $\bf 0$ and covariance $\Sigma$, 
and ${X} = [{\bf x}_1 ~ \cdots ~ {\bf x}_n]^\intercal$. 
 Let $\sigma_{\max}({X}), \sigma_{\min}({ X})$, $\gamma_{\max}(\Sigma)$ and $\gamma_{\min}(\Sigma)$, respectively, denote the largest and smallest singular value of $X$ and the largest and smallest eigenvalue of $\Sigma$.  For all $\delta > 0$, we then have
\begin{equation}\label{eq:lemma:cite:1}
\mathbb{P}\left[\frac{\sigma_{\max }(X)}{\sqrt{n}} \geq \sqrt{\gamma_{\max }(\Sigma)}(1+\delta)+\sqrt{\frac{\operatorname{tr}(\Sigma)}{n}}\right] \leq e^{-n \delta^{2} / 2}, 
\end{equation}
and
\begin{equation}
\mathbb{P}\left[\frac{\sigma_{\min }(X)}{\sqrt{n}} \leq \sqrt{\gamma_{\min }(\Sigma)}(1-\delta)-\sqrt{\frac{\operatorname{tr}(\Sigma)}{n}}\right] \leq e^{-n \delta^{2} / 2}.
\end{equation}
\end{lemma}


Recall that $S_j = \{i: a_{ij} = 1\}$, where $a_{ij}$ is the $(i,j)$-th entry of the adjacency matrix $A$, 
and  
\[\widehat{\Sigma} = (n_1 + n_2)^{-1}\left( \sum_{k=1}^{n_1} ({\bf x}^{(1)}_k - \overline{\bf x}^{(1)})^{\otimes 2} + \sum_{k=1}^{n_2} ({\bf x}^{(2)}_k - \overline{\bf x}^{(2)})^{\otimes 2} \right),\]
where for any vector $\bf a$, ${\bf a}^{\otimes 2} = {\bf aa}^\intercal$.


The next result bounds the spectral norm distance between $\widehat{\Sigma}_{S_j, S_j}$  and $\Sigma_{S_j, S_j}$ for all $j \in \mathcal{J}$, where $\mathcal{J} = \{j: |S_j| > 0\}$. 
\begin{lemma}\label{covariance}
As $n_1, n_2 \rightarrow \infty$, we have 
\begin{align}\label{lemma1:eq2}
\max_{j \in \mathcal{J}}\frac{\|\Sigma_{S_j, S_j}^{-1/2}\widehat{\Sigma}_{S_j, S_j}\Sigma_{S_j, S_j}^{-1/2} - I_{|S_j|}\|_2}{|S_j|} =  O_p\left( \frac{ \log p_0}{n_1 + n_2}\right).
\end{align}
Furthermore, if $d\log p_0 = o(n_1 + n_2)$, then
\begin{align}\label{lemma1:eq3}
\max_{j \in \mathcal{J}}\frac{\|\Sigma_{S_j, S_j}^{1/2}\{\widehat{\Sigma}_{S_j, S_j}\}^{-1}\Sigma_{S_j, S_j}^{1/2} - I_{|S_j|}\|_2}{|S_j|} =  O_p\left( \frac{ \log p_0}{n_1 + n_2}\right).
\end{align}
\end{lemma}
\begin{proof}[Proof of Lemma \ref{covariance}]
Recalling that 
\begin{equation*}
{\bf z}_i =
    \begin{cases}
      {\bf x}^{(1)}_k - \ve \mu_1 & k =1,2,\ldots, n_1\\
      {\bf x}^{(2)}_{k-n_1} - \ve \mu_2 & k = n_1+1,\ldots,n_1+n_2
    \end{cases},        
\end{equation*}
we write 
\begin{align*}
    \widehat{\Sigma} = &~ \frac{1}{n_1 + n_2}\sum_{k=1}^{n_1 + n_2}{\bf z}_k^{\otimes 2} + \frac{n_1}{n_1 + n_2}(\overline{\bf x}^{(1)} - \ve \mu_1)^{\otimes 2} + \frac{n_2}{n_1 + n_2}(\overline{\bf x}^{(2)} - \ve \mu_2)^{\otimes 2} \nonumber. 
\end{align*}
Letting ${\bf w}_k = \Sigma^{-1/2}{\bf z}_k$, it is easy to check that ${\bf w}_k \stackrel{i.i.d}{\sim} N(0,{I}_p)$. Thus, we have
\begin{align*}
   & \left\|\Sigma_{S_j, S_j}^{-1/2}\widehat{\Sigma}_{S_j, S_j}\Sigma_{S_j, S_j}^{-1/2} - I_{|S_j|}\right\|_2 \leq \left\| \frac{1}{n_1 + n_2}\sum_{k=1}^{n_1 + n_2} ({\bf w}_k^{\otimes 2})_{S_j, S_j} - {I}_{|S_j|} \right \|_2 \nonumber \\
    + &~ \frac{n_1}{n_1 + n_2}\left\|\Sigma_{S_j, S_j}^{-1/2}(\overline{\bf x}^{(1)} - \ve \mu_1)_{S_j}\right\|_2^2 + \frac{n_2}{n_1 + n_2}\left\|\Sigma_{S_j, S_j}^{-1/2}(\overline{\bf x}^{(2)} - \ve \mu_2)_{S_j}\right\|_2^2.
\end{align*}
Using Lemma \ref{lemma:wainwright}, we know for all $\delta > 0$, 
\begin{align*}
    \mathbb{P}\left\{ \sqrt{\gamma_{\max}\left( \frac{1}{n_1 + n_2}\sum_{k=1}^{n_1 + n_2} ({\bf w}_k^{\otimes 2})_{S_j, S_j} \right)} \geq 1 + \delta + \sqrt{\frac{|S_j|}{n_1 + n_2}} \right\} \leq e^{-\frac{(n_1 + n_2)\delta^2}{2}}
\end{align*}
and
\begin{align*}
    \mathbb{P}\left\{ \sqrt{\gamma_{\min}\left( \frac{1}{n_1 + n_2}\sum_{k=1}^{n_1 + n_2} ({\bf w}_k^{\otimes 2})_{S_j, S_j} \right)} \leq 1 - \delta - \sqrt{\frac{|S_j|}{n_1 + n_2}} \right\} \leq e^{-\frac{(n_1 + n_2)\delta^2}{2}}.
\end{align*}
These together yield
\begin{align*}
   \mathbb{P} \left\{ \left\| \frac{1}{n_1 + n_2}\sum_{k=1}^{n_1 + n_2} ({\bf w}_k^{\otimes 2})_{S_j, S_j} - {I}_{|S_j|} \right\|_2  \geq \left(\delta + \sqrt{\frac{|S_j|}{n_1 + n_2}}\right)^2 \right\} \leq e^{-\frac{(n_1 + n_2)\delta^2}{2}}.
\end{align*}
By taking 
$\delta = \sqrt{4|S_j| \log p_0/({n_1 + n_2})}$,
we have 
\[
 \mathbb{P} \left\{ |S_j|^{-1}\left\|\frac{1}{n_1 + n_2}\sum_{k=1}^{n_1 + n_2} ({\bf w}_k^{\otimes 2})_{S_j, S_j} - {I}_{|S_j|}  \right\|_2  \geq \frac{1}{n_1 + n_2}\left(\sqrt{2\log p_0} + 1\right)^2 \right\} \leq {p_0}^{-2|S_j|}.
\]
Taking the union over $j  \in \mathcal{J}$, we have
\begin{align*}
     &\mathbb{P} \left\{ \max_{j  \in \mathcal{J}}  \left\{ |S_j|^{-1}\left\| \frac{1}{n_1 + n_2}\sum_{k=1}^{n_1 + n_2} ({\bf w}_k^{\otimes 2})_{S_j, S_j} - {I}_{|S_j|} \right\|_2  \right\}
      \geq \frac{1}{n_1 + n_2}\left(\sqrt{2\log p_0} + 1\right)^2  \right\} \nonumber \\
     & \leq \sum_{j \in \mathcal{J}} p_0^{-2|S_j|} \leq p_0^{-1},
\end{align*}
which goes to $0$ as $n_1, n_2 \rightarrow \infty$.
This indicates that as $n_1, n_2 \rightarrow \infty$, we have
\[
\max_{j \in \mathcal{J}} |S_j|^{-1} \left\| \left(\frac{1}{n_1 + n_2}\sum_{k=1}^{n_1 + n_2} {\bf w}_k^{\otimes 2} - {I}_p \right)_{S_j, S_j} \right\|_2  = O_p \left( \frac{\log p_0}{n_1 + n_2} \right).
\]
Also,
\begin{align*}
    \frac{n_1}{n_1 + n_2}\left\|\Sigma_{S_j, S_j}^{-1/2}(\overline{\bf x}^{(1)} - \ve \mu_1)_{S_j}\right\|_2^2 =  \frac{n_1}{n_1 + n_2}\left\|n_1^{-1}\sum_{k = 1}^{n_1}({\bf w}_{k})_{S_j}\right\|_2^2.
\end{align*}
Note that $\left\|n_1^{-1/2}\sum_{k = 1}^{n_1}({\bf w}_{k})_{S_j}\right\|_2^2 \sim \chi^2_{|S_j|}$, therefore by using the standard tail bound for $\chi^2$ variables, we can get 
\begin{align*}
    \mathbb{P}\left\{ \left| \|n_1^{-1/2}\sum_{k = 1}^{n_1}({\bf w}_{k})_{S_j}\|_2^2 - |S_j| \right| > |S_j|n_1t \right\} \leq 2e^{-|S_j|n_1t/8} \leq 2e^{-n_1t/8}, ~\mbox{for all}~ t > \frac{1}{n_1}.
\end{align*}
Hence, taking the union over $j \in \mathcal{J}$ gives
\begin{align*}
    \mathbb{P}\left\{ \max_{j \in \mathcal{J}} \left| \|n_1^{-1/2}|S_j|^{-1}\sum_{k = 1}^{n_1}({\bf w}_{k})_{S_j}\|_2^2 - 1 \right| > n_1t \right\} \leq e^{\log(2p_0) - n_1t/8}, ~\mbox{for all}~ t > \frac{1}{n_1}.
\end{align*}
By taking $t = 16\log(2p_0)/n_1 > n_1^{-1}$, we get
\begin{align*}
    \mathbb{P}\left\{ \max_{j \in \mathcal{J}} \left| |S_j|^{-1} \|n_1^{-1/2}\sum_{k = 1}^{n_1}({\bf w}_{k})_{S_j}\|_2^2 - 1 \right| > 16\log (2p_0) \right\} \leq (2p_0)^{-1}.
\end{align*}
This indicates that
\[
\frac{n_1}{n_1 + n_2}\max_{j \in \mathcal{J}}|S_j|^{-1}\left\|n_1^{-1}\sum_{k = 1}^{n_1}({\bf w}_{k})_{S_j}\right\|_2^2 = O_p\left( \frac{\log p_0}{n_1 + n_2} \right).
\]
Similarly, we can show that
\begin{align}\label{pf:lemma1:8}
    \frac{n_2}{n_1 + n_2}\left\|\Sigma_{S_j, S_j}^{-1/2}(\overline{\bf x}^{(2)} - \ve \mu_2)_{S_j}\right\|_2^2 = O_p\left( \frac{\log p_0}{n_1 + n_2} \right).
\end{align}
This completes the proof of (\ref{lemma1:eq2}) in Lemma \ref{covariance}. To show (\ref{lemma1:eq3}) in Lemma \ref{covariance}, we only need to note that
\[
\|\Sigma_{S_j, S_j}^{1/2}\{\widehat{\Sigma}_{S_j, S_j}\}^{-1}\Sigma_{S_j, S_j}^{1/2} - I_{|S_j|}\|_2 \leq 
\frac{\|\Sigma_{S_j, S_j}^{-1/2}\widehat{\Sigma}_{S_j, S_j}\Sigma_{S_j, S_j}^{-1/2} - I_{|S_j|}\|_2}{1 - \|\Sigma_{S_j, S_j}^{-1/2}\widehat{\Sigma}_{S_j, S_j}\Sigma_{S_j, S_j}^{-1/2} - I_{|S_j|}\|_2};
\]
here, we use the matrix variant of Taylor expansion. 
It follows from (\ref{lemma1:eq2}) that 
if $d \log p_0 = o(n_1 + n_2)$, we then have $\|\Sigma_{S_j, S_j}^{1/2}\{\widehat{\Sigma}_{S_j, S_j}\}^{-1}\Sigma_{S_j, S_j}^{1/2} - I_{|S_j|}\|_2 = O(\|\Sigma_{S_j, S_j}^{-1/2}\widehat{\Sigma}_{S_j, S_j}\Sigma_{S_j, S_j}^{-1/2} - I_{|S_j|}\|_2)$. This yields (\ref{lemma1:eq3}).
\end{proof}
For any set $E$, let 
$A_{E} = \{ {\bf x}^{(l)}_{i, E}: l = 1,2; i = 1, \ldots, n_l\}$. 
The following lemma characterizes the properties of $\widehat{Q}$ and $\widehat{R}$. 
\begin{lemma}\label{independence}
Under Assumptions 1-3, for $j \in \mathcal{J}$ and any arbitrary set $E$ that satisfies $S_j \subseteq E \subseteq \{1, \ldots, j-1\}$, we have
\begin{enumerate}
    \item[(1)]  Conditioned on $A_E$,
    $\widehat{\bf q}_{j,S_j} \independent \widehat{r}_j$ 
    and $(\overline{\bf x}^{(1)} - \overline{\bf x}^{(2)})^\intercal {\bf e}_j \independent \widehat{r}_j$, 
    where ${\bf e}_j \in \mathbb{R}^p$ denotes the vector where the $j$-th entry equals 1 and 0 elsewhere.  
    \item[(2a)] Conditioned on $A_E$, $\widehat{\bf q}_{j,S_j}$ follows a multivariate normal distribution with mean ${\bf q}_{j,S_j}$ and covariance ${r_j}(n_1 + n_2)^{-1}\{\widehat{\Sigma}_{S_j, S_j}\}^{-1}$.
    \item[(2b)] Conditioned on $A_E$, $\widehat{\bf q}_{j,S_j} \independent (\overline{\bf x}^{(1)} - \overline{\bf x}^{(2)})^\intercal {\bf e}_j$. 
    \item[(2c)] For $k < j$, we have 
    \[
    \mathbb{E}\left[ (\widehat{\bf q}_{j,S_j} - {\bf q}_{j,S_j}) (\widehat{\bf q}_{k,S_k} - {\bf q}_{k,S_k})^\intercal \mid  A_{S_j \cup S_k \cup \{k\}} \right] = 0_{|S_j| \times |S_k|}. 
    \]
    \item[(3a)]  
    For $j = 1, \ldots, p$, we have
    $\mathbb{E}\left[\widehat{r}_j^{-1}\mid A_E \right] = {r}_j^{-1}$ and \(\mbox{Var}[\widehat{r}_j^{-1}\mid A_E] = {r_j^{-2}(n_1 + n_2 - |S_j| - 6)^{-1}}\).
    \item[(3b)] For $k < j$, 
    \[
    \mathbb{E}\left[ (\widehat{r}_j^{-1} - r_j^{-1})(\widehat{r}_k^{-1} - r_k^{-1}) \mid A_{S_j \cup S_k \cup \{k\}}\right] = 0. 
    \]
\end{enumerate}
\end{lemma}

\begin{proof}
We first recall some notations.
We denote by ${\bf 1}_{n_1+n_2}$ the $(n_1 + n_2)$-dimensional vector of all ones and let ${\bf g} = (g_1, \ldots, g_{n_1 + n_2})^\intercal$ denote the group indicator; that is, $g_i$ equals 1 for $i = 1, \ldots, n_1$ and equals 0 for $i = n_1 + 1, \ldots, n_1 + n_2$. Then, letting $X = \left[ {\bf x}_1^{(1)} ~\cdots ~ {\bf x}_{n_1}^{(1)} ~ {\bf x}_1^{(2)} ~\cdots~ {\bf x}_{n_2}^{(2)} \right]^\intercal$, we consider
\begin{equation}\label{rev:1}
    \left\|X_j - \theta_{1j}{\bf 1}_{n_1+n_2} - \theta_{2j}{\bf g} - X_{S_j}{\bf q}_{j,S_j} \right\|_2^2.
\end{equation}
Here, $X_j$ denotes the $j$-th column of $X$; $X_{S_j}$ denotes the sub-matrix of $X$ with columns indexed by $S_j$; ${\bf q}_{j,S_j}$ denotes the sub-vector of ${\bf q}_j$ with entries indexed by $S_j$. 
Thus, letting $W = [{\bf 1}_{n_1 + n_2} ~ {\bf g}]$,
we obtain the OLS estimator of ${\bf q}_{j,S_j}$ and
$(\theta_{1j}, \theta_{2j})^\intercal$:
\begin{align}\label{rev:2}
    \widehat{\bf q}_{j,S_j} = &~ \left(X_{S_j}^\intercal (I_{n_1+n_2} - \mathcal{P}_{W}) X_{S_j}\right)^{-1} X_{S_j}^\intercal (I_{n_1+n_2} - \mathcal{P}_{W}) X_{j}, \nonumber \\
    (\widehat{\theta}_{1j}, \widehat{\theta}_{2j})^\intercal = &~ (W^\intercal W)^{-1}W^\intercal \left(X_j - X_{S_j}\widehat{\bf q}_{j,S_j}\right),
\end{align}
where $\mathcal{P}_W = W(W^\intercal W)^{-1}W^\intercal$ is the orthogonal projection matrix onto the column space of $W$. Then, the ``plug-in" estimator of $r_j$ is
\[\widehat{r}_j = \frac{1}{n_1 + n_2 - |S_j| - 4} 
\left\|X_j - \widehat{\theta}_{1j}{\bf 1}_{n_1+n_2} - \widehat{\theta}_{2j}{\bf g} - X_{S_j}\widehat{\bf q}_{j,S_j} \right\|_2^2;
\]
for shorthand notation, let $\hat{\ve \varepsilon}_j = X_j - \widehat{\theta}_{1j}{\bf 1}_{n_1+n_2} - \widehat{\theta}_{2j}{\bf g} - X_{S_j}\widehat{\bf q}_{j,S_j}$.


\begin{itemize}
    \item[(1)] Note that under Gaussianity, conditioned on the predictors, the OLS coefficient estimator is independent of the OLS variance estimator (note that the group indicator $g_i$ is a fixed quantity). Thus, $\widehat{\bf q}_{j,S_j} \independent \widehat{r}_j$ given $A_{S_j}$. We next prove $(\overline{\bf x}^{(1)} - \overline{\bf x}^{(2)})^\intercal {\bf e}_j = \left(\overline{x}^{(1)}_{j} - \overline{x}^{(2)}_{j}\right) \independent \widehat{r}_j$ given $A_{S_j}$. 
    Under Gaussianity, we only need to show 
    \[
    \mbox{Cov}\left(\widehat{\ve \varepsilon}_j, \overline{x}_j^{(l)} \mid A_{S_j}\right) = {\bf 0}, ~~l = 1,2.
    \]
    We first consider $l = 1$. By (\ref{rev:1}),
    we get 
    \[
    \hat{\ve \varepsilon}_j = \left\{I_{n_1 + n_2} - \mathcal{P}_W(I_{n_1 + n_2} - X_{S_j}F) - X_{S_j}F \right\}X_j,
    \]
    where $F = \left(X_{S_j}^\intercal (I_{n_1 + n_2} - \mathcal{P}_{W}) X_{S_j}\right)^{-1} X_{S_j}^\intercal (I_{n_1 + n_2} - \mathcal{P}_{W})$. Then, since $\overline{x}^{(l)}_j = n_1^{-1}{\bf g}^\intercal X_j$, we have
    \[
    \mbox{Cov}\left(\widehat{\ve \varepsilon}_j, \overline{x}_j^{(1)} \mid A_{S_j}\right) = n_1^{-1}\left\{I_{n_1 + n_2} - \mathcal{P}_W(I_{n_1 + n_2} - X_{S_j}F) - X_{S_j}F \right\}\mbox{Cov}\left(X_j \mid A_{S_j}\right){\bf g}.
    \]
    Note that $X_j - \theta_{1j}{\bf 1}_{n_1 + n_2} - \theta_{2j}{\bf g} - X_{S_j}{\bf q}_{j,S_j} = \widetilde{\ve \epsilon}_j$, where
    $\widetilde{\ve \epsilon}_j = (\epsilon_{1j}, \ldots, \epsilon_{n_1+n_2, j})^\intercal$. Thus, 
    \begin{equation}\label{rev:3}
    \mbox{Cov}\left(X_j \mid A_{S_j}\right) = \mbox{Cov}(\widetilde{\ve \epsilon}_j) = r_j I_{n_1 + n_2}.
    \end{equation}
    This leads to 
    \[
    \mbox{Cov}\left(\widehat{\ve \varepsilon}_j, \overline{x}_j^{(1)} \mid A_{S_j}\right) = n_1^{-1}r_j\left\{I_{n_1+n_2} - \mathcal{P}_W(I_{n_1 + n_2} - X_{S_j}F) - X_{S_j}F \right\}{\bf g} = {\bf 0},
    \]
    because $\mathcal{P}_W {\bf g} = {\bf g}$. For $l = 2$, we notice $\overline{x}_{j}^{(2)} = n_2^{-1} ({\bf 1}_{n_1 + n_2} - {\bf g})^{\intercal} X_j$. Using similar tricks, we get
    \[
    \mbox{Cov}\left(\widehat{\ve \varepsilon}_j, \overline{x}_j^{(2)} \mid A_{S_j}\right) = n_2^{-1}r_j\left\{I_{n_1+n_2} - \mathcal{P}_W(I_{n_1 + n_2} - X_{S_j}F) - X_{S_j}F \right\}({\bf 1}_{n_1 + n_2} - {\bf g}) = {\bf 0},
    \]
    because $\mathcal{P}_W({\bf 1}_{n_1 + n_2} - {\bf g}) = {\bf 1}_{n_1 + n_2} - {\bf g}$. Therefore, we have $\widehat{\varepsilon}_j \independent \overline{x}_j^{(1)} - \overline{x}_j^{(2)}$ given $A_{S_j}$.

    To generalize this result to any arbitrary set $E$ that satisfies $S_j \subseteq E \subseteq \{1, \ldots, j-1\}$, we note that after conditioning on $\{ {\bf x}^{(l)}_{i, S_j}: l = 1,2; i = 1, \ldots, n_l\}$, the only random variable in $\widehat{r}_j$ is $\{\epsilon_{ij}, i = 1, \ldots, n_1 + n_2\}$. Since $\epsilon_{ij} \independent y_{il}$ for all $l < j$ and $i = 1, \ldots, n_1 + n_2$, we have  
    \[
    \left(\widehat{\bf q}_{j,S_j}, \widehat{r}_j, \overline{x}^{(1)}_j - \overline{x}^{(2)}_j\right) \mid A_E \stackrel{d}{=} \left(\widehat{\bf q}_{j,S_j}, \widehat{r}_j, \overline{x}^{(1)}_j - \overline{x}^{(2)}_j\right) \mid A_{S_j}
    \]
    for any $S_j \subseteq E \subseteq \{1, 2, \ldots, j-1\}$. This completes the proof. 
   \item[(2a)] 
   By (\ref{rev:1}) and (\ref{rev:3}), we immediately know that
   \[
   \widehat{\bf q}_{j, S_j} \mid A_{E} \sim N\left({\bf q}_{j,S_j}, r_j FF^\intercal\right).
   \]
   We next find the explicit form of $FF^\intercal$. Note that
   \[
   FF^\intercal = \left(X_{S_j}^\intercal (I_{n_1 + n_2} - \mathcal{P}_W) X_{S_j}\right)^{-1}.
   \]
   Also, 
   \[
   \mathcal{P}_W = W (W^\intercal W)^{-1} W^\intercal = \frac{1}{n_1n_2} W
   \begin{bmatrix}
   n_1 & -n_1 \\
   -n_1 & n_1 + n_2 
   \end{bmatrix}
   W^\intercal = \frac{1}{n_2}({\bf 1}_{n_1 + n_2} - {\bf g})({\bf 1}_{n_1 + n_2} - {\bf g})^\intercal + \frac{1}{n_1}{\bf gg}^\intercal.
   \]
   Note that 
   \[
   \frac{1}{n_2}X_{S_j}^\intercal ({\bf 1}_{n_1 + n_2} - {\bf g}) = \overline{\bf x}^{(2)}_{S_j}, ~~~\frac{1}{n_1}X_{S_j}^\intercal {\bf g} = \overline{\bf x}^{(1)}_{S_j},
   \]
   Then, we have 
   \[
   X_{S_j}^\intercal (I_{n_1 + n_2} - \mathcal{P}_W) X_{S_j} = \sum_{l=1}^2 \sum_{i=1}^{n_l} ({\bf x}_{i,S_j}^{(l)} - \overline{\bf x}^{(l)}_{S_j})({\bf x}_{i,S_j}^{(l)} - \overline{\bf x}^{(l)}_{S_j})^\intercal = (n_1 + n_2)\widehat{\Sigma}_{S_j, S_j}.
   \]
   Here, $\overline{\bf x}_{S_j}^{(l)}$ and ${\bf x}_{i, S_j}^{(l)}$, respectively, denote the sub-vector of $\overline{\bf x}^{(l)}$ and ${\bf x}_{i}^{(l)}$ with entries indexed by $S_j$. This completes the proof.

\item[(2b)] 
Under Gaussianity, we only need to show that
\[
\mbox{Cov}\left(\widehat{\bf q}_{j, S_j}, \overline{x}_{j}^{(l)} \mid A_E \right) = {\bf 0} ~~~ \mbox{ for } l = 1,2.
\]
Using similar tricks to (1), we have 
\[
\mbox{Cov}\left(\widehat{\bf q}_{j, S_j}, \overline{x}_{j}^{(1)} \mid A_E \right) = r_j F {\bf g} = {\bf 0}, ~~~ \mbox{Cov}\left(\widehat{\bf q}_{j, S_j}, \overline{x}_{j}^{(2)} \mid A_E \right) = r_j F ({\bf 1}_{n_1 + n_2} - {\bf g}) = {\bf 0}.
\]
This completes the proof.


\item[(2c)] 
Since $Q$ is strictly upper triangular, thus $j \notin S_k$ but possibly $k \in S_j$. 
Then, conditioned on $A_{S_j \cup S_k \cup \{k\}}$, $\widehat{\bf q}_{k,S_k}$ becomes a fixed quantity; however, since $k < j$ and $S_j \subseteq S_j \cup S_k \cup \{k\} \subseteq \{1, \ldots, j-1\}$, it follows from 2(a) that $\widehat{\bf q}_{j,S_j}$ follows a multivariate normal distribution with mean ${\bf q}_{j, S_j}$ and covariance $r_j(n_1 + n_2)^{-1}\{\widehat{\Sigma}_{S_j, S_j}\}^{-1}$. 

Therefore, 
 \[
    \mathbb{E}\left[ (\widehat{\bf q}_{j,S_j} - {\bf q}_{j,S_j}) (\widehat{\bf q}_{k,S_k} - {\bf q}_{k,S_k})^\intercal \mid  A_{S_j \cup S_k \cup \{k\}} \right] =  {\bf 0} \times (\widehat{\bf q}_{k,S_k} - {\bf q}_{k,S_k})^\intercal = {\bf 0},
    \]
which completes the proof.

\item[(3a)] Using the properties of the OLS variance estimator, conditioned on $A_E$, we have 
    $(n_1 + n_2 - |S_j| - 4) \widehat{r}_j r_j^{-1} \sim \chi^2_{n_1 + n_2 - 2 - |S_j|}$. This indicates that conditioned on $A_{E}$, $(n_1 + n_2 - |S_j| - 4)^{-1} \widehat{r}_j^{-1} r_j$ follows an inverse chi-square distribution with the degree of freedom $n_1 + n_2 - 2 - |S_j|$. Then, we know that $\mathbb{E}\left[(n_1 + n_2 - |S_j| - 4)^{-1} \widehat{r}_j^{-1} r_j \mid A_E \right] =  (n_1 + n_2 - |S_j| - 4)^{-1}$ and $\mbox{Var}\left[(n_1 + n_2 - |S_j| - 4)^{-1} \widehat{r}_j^{-1} r_j \mid A_E \right] = (n_1 + n_2 - |S_j| - 4)^{-2}(n_1 + n_2 - |S_j| - 6)^{-1}$. 
    
    This leads to 
    \[
    \mathbb{E}\left(\widehat{r}_j^{-1}\mid A_E \right) = {r}_j^{-1}, ~~~ \mbox{Var}\left(\widehat{r}_j^{-1}\mid A_E\right) = {r_j^{-2}(n_1 + n_2 - |S_j| - 6)^{-1}},
    \]
    which completes the proof.
    
    \item[(3b)] 
    Using similar arguments to (2c), we have for $k < j$,
    \[
    \mathbb{E}\left[ (\widehat{r}_j^{-1} - r_j^{-1})(\widehat{r}_k^{-1} - r_k^{-1}) \mid A_{S_j \cup S_k \cup \{k\}}\right] = \mathbb{E} \left[ \widehat{r}_j^{-1} - r_j^{-1} \mid A_{S_j \cup S_k \cup \{k\}} \right] \times (\widehat{r}_k^{-1} - r_k^{-1}) = 0,
    \]
    which completes the proof.
\end{itemize}

\end{proof}


With all the lemmas, we next prove Theorem 1.
\begin{proof}[Proof of Theorem 1]
Recall that $N = n_1n_2/(n_1 + n_2)$ and $A_{E} = \{ {\bf x}^{(l)}_{i, E}: l = 1,2; i = 1, \ldots, n_l\}$ for any set $E \subset \{1, \ldots, p\}$.
We write 
\begin{align}\label{pf:thm:1}
    T^2(\widehat{Q}, \widehat{R}) = & ~ N(\overline{\bf x}^{(1)} - \overline{\bf x}^{(2)})^\intercal \Sigma^{-1} (\overline{\bf x}^{(1)} - \overline{\bf x}^{(2)}) \nonumber \\
    + & ~ N(\overline{\bf x}^{(1)} - \overline{\bf x}^{(2)})^\intercal \left\{ \widehat{\Sigma}_{\text{DAG}}^{-1} - \Sigma^{-1} \right\} (\overline{\bf x}^{(1)} - \overline{\bf x}^{(2)}) \nonumber \\
    = &~ J_1 + J_2. 
\end{align}
It is easy to check that $J_1 \myeq \chi^2_p$.
Next, we provide an asymptotic bound for $J_2$. We first write
\begin{align}\label{add:1117:01}
J_2
    = &~ N\left\{\sum_{j \in \mathcal{J}} \widehat{r_j}^{-1} \{(\overline{\bf x}^{(1)} - \overline{\bf x}^{(2)})^\intercal({\bf e}_j - \widehat{\bf q}_j)\}^2 - \sum_{j \in \mathcal{J}} {r_j}^{-1} \{(\overline{\bf x}^{(1)} - \overline{\bf x}^{(2)})^\intercal({\bf e}_j - {\bf q}_j)\}^2 \right\}\nonumber \\
    + &~ N\left\{\sum_{j \notin \mathcal{J}} \widehat{r_j}^{-1} \{(\overline{\bf x}^{(1)} - \overline{\bf x}^{(2)})^\intercal{\bf e}_j\}^2 - \sum_{j \notin \mathcal{J}} {r_j}^{-1} \{(\overline{\bf x}^{(1)} - \overline{\bf x}^{(2)})^\intercal{\bf e}_j\}^2 \right\}\nonumber \\
    = &~ J_{21} + J_{22}.
\end{align}
For $J_{22}$, based on Lemma \ref{independence}, we have
\begin{align}\label{add:1117:02}
\mathbb{E}(J_{22}) = N\sum_{j \notin \mathcal{J}} \mathbb{E}\left[ \mathbb{E}\left[ \{(\overline{\bf x}^{(1)} - \overline{\bf x}^{(2)})^\intercal{\bf e}_j\}^2 \mid A_{S_j}\right] \mathbb{E}(\widehat{r}_j^{-1} - r_j^{-1} \mid A_{S_j} ) \right]= 0.
\end{align}
This double expectation trick will be used throughout the proof.
Also, we have
\begin{align}\label{add:1117:04}
    \mbox{Var}(J_{22}) = &~ N^2 \sum_{j \notin \mathcal{J}} \mathbb{E} \left[\{(\overline{\bf x}^{(1)} - \overline{\bf x}^{(2)})^\intercal{\bf e}_j r_j^{-1/2}\}^4 r_j^2(\widehat{r}_j^{-1} - r_j^{-1})^2 \right]\nonumber \\
    + &~ 2N^2 \sum_{j \notin \mathcal{J}}\sum_{k \notin \mathcal{J}, k < j} \mathbb{E} \bigg[ \{(\overline{\bf x}^{(1)} - \overline{\bf x}^{(2)})^\intercal{\bf e}_j r_j^{-1/2}\}^2 \{(\overline{\bf x}^{(1)} - \overline{\bf x}^{(2)})^\intercal{\bf e}_k r_k^{-1/2}\}^2  \times \nonumber\\
    &~ r_j(\widehat{r}_j^{-1} - r_j^{-1}) r_k (\widehat{r}_k^{-1} - r_j^{-1})\bigg] \nonumber \\
   = &~ J_{221} + J_{222}.
\end{align}
Since $|S_j| = 0$ for $j \notin \mathcal{J}$, using Lemma \ref{independence} and the double expectation trick, we have 
\[
J_{221} = N^2 \sum_{j \notin \mathcal{J}} \mathbb{E} [\{(\overline{\bf x}^{(1)} - \overline{\bf x}^{(2)})^\intercal{\bf e}_j r_j^{-1/2}\}^4] \mathbb{E} [r_j^2(\widehat{r}_j^{-1} - r_j^{-1})^2] = \frac{3(p - p_0)}{n_1 + n_2 - 6}. 
\]
Here we also use the fact that $\sqrt{N}(\overline{\bf x}^{(1)} - \overline{\bf x}^{(2)})^\intercal{\bf e}_j r_j^{-1/2} \sim N(0,1)$ for $j \notin \mathcal{J}$.
Similarly, we have 
\begin{align*}
    J_{222} = &~ 2N^2 \sum_{j \notin \mathcal{J}}\sum_{k \notin \mathcal{J}, k < j} 
\mathbb{E} \bigg[    
\mathbb{E} [(\overline{\bf x}^{(1)} - \overline{\bf x}^{(2)})^\intercal{\bf e}_j r_j^{-1/2}\}^2 r_j \mid B_{j,k}] \mathbb{E} [ (\widehat{r}_j^{-1} - r_j^{-1}) \mid B_{j,k} ]  \times \nonumber\\
  &~  \{(\overline{\bf x}^{(1)} - \overline{\bf x}^{(2)})^\intercal{\bf e}_k r_k^{-1/2}\}^2
   r_k (\widehat{r}_k^{-1} - r_k^{-1})] \bigg] \\
   =&~ 0.
\end{align*}
Combining results of $J_{221}$ and $J_{222}$, we have 
\[
\mbox{Var}(J_{22}) = 3(p - p_0)/(n_1 + n_2 - 6) = o(1).
\]

For $J_{21}$, we first write that
\begin{align}\label{add:1117:05}
    J_{21} \leq &~ 2N\left\{\sum_{j \in \mathcal{J}} \widehat{r_j}^{-1} \{(\overline{\bf x}^{(1)} - \overline{\bf x}^{(2)})^\intercal({\bf e}_j - {\bf q}_j)\}^2 - \sum_{j \in \mathcal{J}} {r_j}^{-1} \{(\overline{\bf x}^{(1)} - \overline{\bf x}^{(2)})^\intercal({\bf e}_j - {\bf q}_j)\}^2 \right\} + \nonumber \\
    ~& 2N\sum_{j \in \mathcal{J}} \widehat{r_j}^{-1} \{(\overline{\bf x}^{(1)} - \overline{\bf x}^{(2)})^\intercal(\widehat{\bf q}_j - {\bf q}_j)\}^2 \nonumber \\
    = &~ J_{211} + J_{212}.
\end{align}
For $J_{211}$, 
using Lemma 1.3, we have
\[
\mathbb{E}(J_{211}) = N\left\{\sum_{j \in \mathcal{J}} \mathbb{E} \left[ \mathbb{E}[\widehat{r_j}^{-1} - r_j^{-1} | A_{S_j}] \mathbb{E}[(\overline{\bf x}^{(1)} - \overline{\bf x}^{(2)})^\intercal({\bf e}_j - {\bf q}_j)\}^2 \mid A_{S_j}] \right] \right\}= 0. 
\]
Also, 
\begin{align*}
    \mbox{Var}(J_{211}) = &~ N^2 \sum_{j \in \mathcal{J}} \mathbb{E} \left[\{(\overline{\bf x}^{(1)} - \overline{\bf x}^{(2)})^\intercal({\bf e}_j - {\bf q}_j) r_j^{-1/2}\}^4 r_j^2(\widehat{r}_j^{-1} - r_j^{-1})^2\right] \nonumber \\
    + &~ 2N^2 \sum_{j \in \mathcal{J}}\sum_{k \in \mathcal{J}, k < j} \mathbb{E} \bigg[\{(\overline{\bf x}^{(1)} - \overline{\bf x}^{(2)})^\intercal({\bf e}_j - {\bf q}_j) r_j^{-1/2}\}^2 \{(\overline{\bf x}^{(1)} - \overline{\bf x}^{(2)})^\intercal({\bf e}_k - {\bf q}_k) r_k^{-1/2}\}^2 \nonumber\\
    &~ \times r_j(\widehat{r}_j^{-1} - r_j^{-1}) r_k (\widehat{r}_k^{-1} - r_j^{-1})\bigg].
\end{align*}
Noting that $\sqrt{N}(\overline{\bf x}^{(1)} - \overline{\bf x}^{(2)})^\intercal({\bf e}_j - {\bf q}_j) r_j^{-1/2} \sim N(0,1)$ for $j \in \mathcal{J}$, we use the same tricks as those for $J_{221}$ and $J_{222}$ to show
\[
\mbox{Var}(J_{221}) = \sum_{j \in \mathcal{J}}3/(n_1 + n_2 - |S_j| - 6) \leq 3p_0/(n_1 + n_2 - d - 6). 
\]

For $J_{212}$, we first note that  $(\overline{\bf x}^{(1)} - \overline{\bf x}^{(2)})^\intercal (\widehat{\bf q}_j - {\bf q}_j) = (\overline{\bf x}^{(1)} - \overline{\bf x}^{(2)})_{S_j}^\intercal \{\widehat{\Sigma}_{S_j, S_j}\}^{-1/2} \{\widehat{\Sigma}_{S_j, S_j}\}^{1/2} (\widehat{\bf q}_j - {\bf q}_j)_{S_j}$ for $j \in \mathcal{J}$. For ease of notation, we let $\ve \nu_j^\intercal =  (\overline{\bf x}^{(1)} - \overline{\bf x}^{(2)})_{S_j}^\intercal \{\widehat{\Sigma}_{S_j, S_j}\}^{-1/2}$ and $\ve \theta_{j} = \{\widehat{\Sigma}_{S_j, S_j}\}^{1/2}(\widehat{\bf q}_j - {\bf q}_j)_{S_j}$. It follows from Lemma 1.3 that $\mathbb{E}(\ve \theta_{j} | A_{S_j}) = {\bf 0}$ and $\mbox{Cov}(\ve \theta_{j} | A_{S_j}) = r_j(n_1 + n_2)^{-1} I_{S_j}$. Thus,
\begin{align*}
    &~ \mathbb{E}(J_{212}) =  N \sum_{j \in \mathcal{J}} \widehat{r}_j^{-1} \sum_{s \in S_j, t\in S_j} \mathbb{E} \left[\nu_{js} \nu_{jt} \theta_{js} \theta_{jt}\right]
    \nonumber \\
    = &~ N \sum_{j \in \mathcal{J}} \sum_{s \in S_j, t\in S_j} \mathbb{E} \left[ \mathbb{E}[\widehat{r_j}^{-1}\theta_{js} \theta_{jt}|A_{S_j}]  \nu_{js} \nu_{jt}\right] \nonumber \\
   = &~ N \sum_{j \in \mathcal{J}} r_j^{-1} \sum_{s \in S_j} \mathbb{E} \left[ \mathbb{E}[\theta_{js}^2|A_{S_j}] \nu_{js}^2\right] \nonumber \\
   = &~ \frac{N}{n_1 + n_2} \sum_{j \in \mathcal{J}} \mathbb{E}[ \ve \nu_j^\intercal \ve \nu_j ] \nonumber \\
   = &~ \sum_{j \in \mathcal{J}} \frac{|S_j|}{n_1 + n_2}  \left( 1 + N\mathbb{E} \bigg[(\overline{\bf x}^{(1)} - \overline{\bf x}^{(2)})_{S_j}^\intercal {\Sigma}_{S_j, S_j}^{-1/2}  \frac{{\Sigma}_{S_j, S_j}^{1/2}\{\widehat{\Sigma}_{S_j, S_j}\}^{-1}{\Sigma}_{S_j, S_j}^{1/2} - I_{|S_j|}}{|S_j|} {\Sigma}_{S_j, S_j}^{-1/2} (\overline{\bf x}^{(1)} - \overline{\bf x}^{(2)})_{S_j}) \bigg] \right)
\end{align*}
Since $d\log p_0 = o(n_1 + n_2)$, using Lemma \ref{covariance}, we have
\begin{align*}
    &~ N\mathbb{E} \left[ (\overline{\bf x}^{(1)} - \overline{\bf x}^{(2)})_{S_j}^\intercal {\Sigma}_{S_j, S_j}^{-1/2}  \frac{{\Sigma}_{S_j, S_j}^{1/2}\{\widehat{\Sigma}_{S_j, S_j}\}^{-1}{\Sigma}_{S_j, S_j}^{1/2} - I_{|S_j|}}{|S_j|} {\Sigma}_{S_j, S_j}^{-1/2} (\overline{\bf x}^{(1)} - \overline{\bf x}^{(2)})_{S_j}) \right] \nonumber \\
    \leq &~ N\mathbb{E} \left[ \max_{j \in \mathcal{J}} \frac{\left \|{\Sigma}_{S_j, S_j}^{1/2}\{\widehat{\Sigma}_{S_j, S_j}\}^{-1}{\Sigma}_{S_j, S_j}^{1/2} - I_{|S_j|} \right\|_2}{|S_j|} (\overline{\bf x}^{(1)} - \overline{\bf x}^{(2)})_{S_j}^\intercal \{{\Sigma}_{S_j, S_j}\}^{-1} (\overline{\bf x}^{(1)} - \overline{\bf x}^{(2)})_{S_j} \right] \nonumber \\
    \leq &~ O_p \left( \frac{d\log p_0}{n_1 + n_2} \right).
\end{align*}

Since $d\log p_0 (n_1 + n_2)^{-1} = O(1)$ and $N_{e} = \sum_{j \in \mathcal{J}} |S_j|$, we have
\[
{\mathbb{E}(J_{212})} = O\left( \frac{N_e }{n_1 + n_2} \right).
\]
Furthermore,
\begin{align*}
    \mathbb{E}(J_{212}^2) = &~ N^2 \mathbb{E}\left[ \left(\sum_{j \in \mathcal{J}} \widehat{r_j}^{-1} \{(\overline{\bf x}^{(1)} - \overline{\bf x}^{(2)})^\intercal(\widehat{\bf q}_j - {\bf q}_j)\}^2\right)^2 \right] \nonumber \\
    = &~ N^2 \sum_{j \in \mathcal{J},k \in \mathcal{J}} \mathbb{E} \left[ \widehat{r}_j^{-1} \sum_{s \in S_j, t \in S_j} \nu_{js} \theta_{js}
    \nu_{jt} \theta_{jt} \times 
    \widehat{r}_k^{-1} \sum_{u \in S_k, v \in S_k} \nu_{ku} \theta_{ku}
  \nu_{kv} \theta_{kv} \right] \nonumber \\ 
    = &~ N^2 \sum_{j \in \mathcal{J}} \mathbb{E} \left[ \widehat{r}_j^{-2} \{\sum_{s \in S_j, t \in S_j} \nu_{js} \theta_{js}
    \nu_{jt} \theta_{jt}\}^2 \right] \nonumber \\
    + &~ 2N^2 \sum_{j \in \mathcal{J},k \in \mathcal{J}, k < j} \mathbb{E} \left[ \widehat{r}_j^{-1} \sum_{s \in S_j, t \in S_j} \nu_{js} \theta_{js}
    \nu_{jt} \theta_{jt} \times 
    \widehat{r}_k^{-1} \sum_{u \in S_k, v \in S_k} \nu_{ku} \theta_{ku}
  \nu_{kv} \theta_{kv} \right] \nonumber \\ 
    = &~ J_{2121} + J_{2122}.
 \end{align*}
For $J_{2121}$, we have
\begin{align*}
    J_{2121} = &~ N^2 \sum_{j \in \mathcal{J}} \mathbb{E} \bigg[ \mathbb{E} \big( \widehat{r}_j^{-2}  \mid A_{S_j} \big) \mathbb{E} \big[ \big(\sum_{s \in S_j, t \in S_j} \nu_{js} \theta_{js}
    \nu_{jt} \theta_{jt}\big)^2 \mid A_{S_j} \big] \bigg] \nonumber \\
    = &~ N^2 \sum_{j \in \mathcal{J}} r_j^{-2} \left(1 + \frac{1}{n_1 + n_2 - |S_j| - 6}\right) \sum_{s,t,u,v \in S_j} \mathbb{E} \bigg[ \nu_{js}\nu_{jt}\nu_{ju}\nu_{jv} \mathbb{E} \big( \theta_{js}\theta_{jt}\theta_{ju}\theta_{jv} \mid A_{S_j} \big) \bigg].
\end{align*}
Since given $A_{S_j}$, 
$\theta_{js} \sim N(0, r_j(n_1 + n_2)^{-1})$ and 
$\theta_{js} \independent \theta_{jt}$ for all $s \neq t$ and $s, t \in S_j$, 
\begin{align*}
    J_{2121} = &~ N^2 \sum_{j \in \mathcal{J}} \frac{1}{(n_1 + n_2)^2} \left(1 + \frac{1}{n_1 + n_2 - |S_j| - 6}\right) \left( \sum_{s \in S_j} \mathbb{E} \big[ v_{js}^4 \big] 
    + \sum_{s, t \in S_j, s \neq t} \mathbb{E} \big[ \nu_{js}^2 \nu_{jt}^2 \big] \right).
\end{align*}
Since $\sqrt{N}\nu_{js} \sim N(0,1)$ for all $s$, and $\nu_{js} \independent \nu_{jt}$ for all $s \neq t$, we have 
\[
J_{2121}  = O\left(\frac{dN_e}{(n_1 + n_2)^2}\right).
\]

We can use similar tricks on $J_{2122}$. Specifically, 
\begin{align*}
    J_{2122} = &~ 2N^2 \sum_{j \in \mathcal{J},k \in \mathcal{J}, k < j} \mathbb{E} \bigg[ \mathbb{E} \big[ \widehat{r}_j^{-1} \sum_{s \in S_j, t \in S_j} \nu_{js} \theta_{js}
    \nu_{jt} \theta_{jt} \mid B_{j,k} \big]\times 
    \widehat{r}_k^{-1} \sum_{u \in S_k, v \in S_k} \nu_{ku} \theta_{ku}
  \nu_{kv} \theta_{kv} \bigg] \nonumber \\
  = &~ 2N^2 \sum_{j \in \mathcal{J},k \in \mathcal{J}, k < j} \mathbb{E} \bigg[ r_j^{-1} \sum_{s \in S_j} \nu_{js}^2 \mathbb{E} \big[ \theta_{js}^2 | B_{j,k} \big] \widehat{r}_k^{-1} \sum_{u \in S_k, v \in S_k} \nu_{ku} \theta_{ku}
  \nu_{kv} \theta_{kv} \bigg] \nonumber \\
  = &~ \frac{2N}{n_1 + n_2} \sum_{j \in \mathcal{J},k \in \mathcal{J}, k < j} |S_j| \mathbb{E} \big[ \widehat{r}_k^{-1} \sum_{u \in S_k, v \in S_k} \nu_{ku} \theta_{ku}
  \nu_{kv} \theta_{kv} \big] \nonumber \\
  = &~ \frac{2N}{(n_1 + n_2)} \sum_{j \in \mathcal{J},k \in \mathcal{J}, k < j} |S_j| \mathbb{E} \bigg[ \mathbb{E} \big[ r_k^{-1} | A_k \big]  \sum_{u \in S_k, v \in S_k} \nu_{ku} \nu_{kv} \mathbb{E} \big[ \theta_{ku} \theta_{kv} \mid A_k \big] \bigg] \nonumber \\
  = &~ \frac{2}{(n_1 + n_2)^2} \sum_{j \in \mathcal{J},k \in \mathcal{J}, k < j} |S_j| |S_k| \leq  \frac{N_e^2}{(n_1 + n_2)^2}.
\end{align*}
Since $d \leq N_e$, we have 
\[
\mbox{Var}(J_{212}) \leq \mathbb{E}(J_{212}^2) = O\left( \frac{N_e^2}{(n_1 + n_2)^2} \right).
\]
Combining all the assertions above, we have 
\[
\mathbb{E}(J_2) = O\left( \frac{N_e}{n_1 + n_2} \right), ~~~ \mbox{Var}(J_2) = O \left( \frac{p}{n_1 + n_2 - d - 6} + \frac{N_e^2}{(n_1 + n_2)^2}\right).
\]
Therefore, as $n_1, n_2, p \rightarrow \infty$,
\begin{itemize}
    \item if $N_e = o(n_1 + n_2)$ and $p = o(n_1 + n_2)$, then under $H_0$, $T^2(\widehat{Q}, \widehat{R}) \rightarrow \chi^2_p$ in distribution.
    \item if $N_e = o\left(\sqrt{p}(n_1 + n_2)\right)$, we have 
    \[
    \frac{T^2(\widehat{Q}, \widehat{R}) - p}{\sqrt{2p}} = \frac{J_1 - p}{\sqrt{2p}} + o_p(1) \rightarrow N(0, 1)
    \]
    in distribution. 
\end{itemize}
This completes the proof.

\end{proof}

\begin{proof}[Proof of Theorem 2]
 For ease of notation, let ${\bf d}_\mu = \ve \mu_1 - \ve \mu_2$ and ${\bf y}_d = \overline{\bf x}^{(1)} - \overline{\bf x}^{(2)} - {\bf d}_\mu$. Recall that $N = n_1n_2(n_1 + n_2)^{-1}$, $D_{\text{KL}} = \sum_{j=1}^p r_j^{-1}\{ {\bf d}_\mu^\intercal ({\bf e}_j - {\bf q}_j)\}^2$ and $D_{\text{DAG}} = \sum_{j \in \mathcal{J}} {\bf d}_{\mu, S_j}^\intercal \{\Sigma_{S_j, S_j}\}^{-1} {\bf d}_{\mu, S_j}$. 
We can then write
 \begin{align}\label{pf:thm2:1}
     T^2(\widehat{Q}, \widehat{R}) = N{\bf y}_d^\intercal \widehat{\Sigma}_{\text{DAG}}^{-1} {\bf y}_d + ND_{\text{KL}} + N {\bf d}_\mu^\intercal (\widehat{\Sigma}_{\text{DAG}}^{-1} - \Sigma^{-1}) {\bf d}_\mu + 
     2N{\bf d}_\mu^\intercal \widehat{\Sigma}_{\text{DAG}}^{-1}{\bf y}_d.
 \end{align}

Using (\ref{pf:thm2:1}), we can write 
\begin{align}\label{pf:thm2:5}
  \left|\frac{T^2(\widehat{Q}, \widehat{R}) - p}{\sqrt{2p}}\right| \geq   \frac{ND_{\text{KL}}}{\sqrt{2p}} - \left| \frac{ N {\bf y}_d^\intercal \widehat{\Sigma}^{-1}{\bf y}_d  - p}{\sqrt{2p}}\right| - \frac{2N}{\sqrt{2p}}\left|{\bf d}_\mu^\intercal \widehat{\Sigma}_{\text{DAG}}^{-1}{\bf y}_d\right| - \frac{N}{\sqrt{2p}} |{\bf d}_\mu^\intercal (\widehat{\Sigma}_{\text{DAG}}^{-1} - \Sigma^{-1}) {\bf d}_\mu|.
\end{align}
Since $d\log p_0 = O(n_1 + n_2)$ and $N_e = o(\sqrt{p}(n_1 + n_2))$, 
one can use the proof of Theorem 1 to show that 
\[
\frac{ N {\bf y}_d^\intercal \widehat{\Sigma}^{-1}{\bf y}_d  - p}{\sqrt{2p}} \stackrel{d}{\rightarrow} N(0,1) \mbox{ as } n_1, n_2 \rightarrow \infty.  
\]
We now show $\frac{N}{\sqrt{2p}}{\bf d}_\mu^\intercal (\widehat{\Sigma}_{\text{DAG}}^{-1} - \Sigma^{-1}) {\bf d}_{\mu} = o_p(1)$. Specifically, 
\begin{align*}
   {\bf d}_\mu^\intercal (\widehat{\Sigma}_{\text{DAG}}^{-1} - \Sigma^{-1}) {\bf d}_\mu =  \sum_{j = 1}^p \widehat{r}_j^{-1} \{ {\bf d}_\mu^\intercal ({\bf e}_j - \widehat{\bf q}_j)  \}^2 -   \sum_{j = 1}^p {r}_j^{-1} \{ {\bf d}_\mu^\intercal ({\bf e}_j - {\bf q}_j)  \}^2. 
\end{align*}
Again, using the double expectation trick, we have
\begin{align*}
   &~ \mathbb{E} \big[ \widehat{r}_j^{-1} \{ {\bf d}_\mu^\intercal ({\bf e}_j - \widehat{\bf q}_j)  \}^2 \big] = \mathbb{E} \bigg[ \mathbb{E} \big[\widehat{r}_j^{-1} \mid A_{S_j} \big] \mathbb{E} \big[ \{ {\bf d}_\mu^\intercal ({\bf e}_j - \widehat{\bf q}_j)  \}^2 \mid A_{S_j} \big] \bigg] \nonumber \\
   = &~ r_j^{-1} \mathbb{E} \bigg[  \{ {\bf d}_\mu^\intercal ({\bf e}_j - {\bf q}_j)  \}^2 + r_j (n_1 + n_2)^{-1} {\bf d}_{\mu, S_j}^\intercal \{\widehat{\Sigma}_{S_j, S_j}\}^{-1} {\bf d}_{\mu, S_j}\bigg].
\end{align*}
Since $d\log p_0 = o(n_1 + n_2)$, using Lemma \ref{covariance}, we have
\begin{align*}
   &~ \mathbb{E} \big[ {\bf d}_\mu^\intercal (\widehat{\Sigma}_{\text{DAG}}^{-1} - \Sigma^{-1}) {\bf d}_\mu\big] = \frac{1}{n_1 + n_2} \sum_{j \in \mathcal{J}} \mathbb{E} \big[ {\bf d}_{\mu, S_j}^\intercal \{\widehat{\Sigma}_{S_j, S_j}\}^{-1} {\bf d}_{\mu,S_j} \big] \nonumber \\
   = &~ \frac{D_{\text{DAG}}}{n_1 + n_2} \times O\left( 1 + \frac{d\log p_0}{(n_1 + n_2)}  \right).
\end{align*}
Since $D_{\text{DAG}} = O(p^{1/2}N^{-\gamma})$ for $\gamma > 1/2$ and $N/(n_1 + n_2) < 1$, we have
\[
\frac{N}{\sqrt{2p}} \mathbb{E} \big[ {\bf d}_\mu^\intercal (\widehat{\Sigma}_{\text{DAG}}^{-1} - \Sigma^{-1}) {\bf d}_\mu\big] = o(1).
\]

Also, 
\begin{align*}
     &~ \mathbb{E} \bigg[ \{ {\bf d}_\mu^\intercal (\widehat{\Sigma}_{\text{DAG}}^{-1} - \Sigma^{-1}) {\bf d}_\mu \} ^2\bigg] \leq 2 \mathbb{E}\left\{\sum_{j=1}^p (\widehat{r}_j^{-1} - r_j^{-1}) \{{\bf d}_u^\intercal ({\bf e}_j - \widehat{\bf q}_j) \}^2  \right\}^2 \nonumber \\
     + &~ 2 \mathbb{E}\left\{ \sum_{j=1}^p r_j^{-1} \left( \{{\bf d}_u^\intercal ({\bf e}_j - \widehat{\bf q}_j) \}^2  - \{{\bf d}_u^\intercal ({\bf e}_j - {\bf q}_j) \}^2 \right) \right\}^2 \nonumber \\
     = &~ 2(J_3 + J_4). 
\end{align*}
Note that 
\begin{align*}
    J_3 = &~ \sum_{j=1}^p \mathbb{E} \bigg[ (\widehat{r}_j^{-1} - r_j^{-1})^2 \{{\bf d}_u^\intercal ({\bf e}_j - \widehat{\bf q}_j) \}^4 \bigg] \nonumber \\
    + &~ 2\sum_{k < j} \mathbb{E}\bigg[ (\widehat{r}_j^{-1} - r_j^{-1}) \{{\bf d}_u^\intercal ({\bf e}_j - \widehat{\bf q}_j) \}^2 (\widehat{r}_k^{-1} - r_k^{-1}) \{{\bf d}_u^\intercal ({\bf e}_k - \widehat{\bf q}_k) \}^2 \bigg] \nonumber \\
    = &~ J_{31} + 2J_{32}.
\end{align*}
For $J_{31}$, we have 
\begin{align*}
    J_{31} = &~ \sum_{j=1}^p r_j^{-2}\frac{1}{n_1 + n_2 - |S_j| - 6} \mathbb{E} \bigg[ \mathbb{E} \big[ \{{\bf d}_u^\intercal ({\bf e}_j - \widehat{\bf q}_j) \}^4 \mid A_{S_j} \big] \bigg] \nonumber \\
   \leq &~ 8\sum_{j=1}^p r_j^{-2}\frac{1}{n_1 + n_2 - |S_j| - 6} \left( ({\bf d}_\mu^\intercal ({\bf e}_j - {\bf q}_j) )^4 + 3\mathbb{E}\left[\left(\frac{r_j {\bf d}_{\mu, S_j}^\intercal \{\widehat{\Sigma}_{S_j, S_j}\}^{-1} {\bf d}_{\mu, S_j}}{n_1 + n_2}\right)^2\right] I(j \in \mathcal{J}) \right) \nonumber \\
   \leq &~ \frac{8}{n_1 + n_2 - d - 6} \sum_{j=1}^p r_j^{-2} ({\bf d}_\mu^\intercal ({\bf e}_j - {\bf q}_j) )^4 + \frac{24}{n_1 + n_2 - d - 4}\sum_{j \in \mathcal{J}} \mathbb{E}\left[\left(\frac{{\bf d}_{\mu, S_j}^\intercal \{\widehat{\Sigma}_{S_j, S_j}\}^{-1} {\bf d}_{\mu, S_j}}{n_1 + n_2}\right)^2\right]. \nonumber \\
\end{align*}
Note that 
\begin{align*}
    \sum_{j \in \mathcal{J}} \mathbb{E}\left[\left(\frac{{\bf d}_{\mu, S_j}^\intercal \{\widehat{\Sigma}_{S_j, S_j}\}^{-1} {\bf d}_{\mu, S_j}}{n_1 + n_2}\right)^2\right] \leq \frac{2}{(n_1 + n_2)^2} \sum_{j \in \mathcal{J}} ({\bf d}_{\mu, S_j}^\intercal \{{\Sigma}_{S_j, S_j}\}^{-1} {\bf d}_{\mu, S_j})^2 \times O\left( 1 + \frac{d^2\log^2 p_0}{(n_1 + n_2)^2}\right).
\end{align*}
Also, 
\[
\sum_{j=1}^p r_j^{-2} ({\bf d}_\mu^\intercal({\bf e}_j - {\bf q}_j))^4 \leq D_{\text{KL}}^2;  \sum_{j \in \mathcal{J}} ({\bf d}_{\mu, S_j}^\intercal \{{\Sigma}_{S_j, S_j}\}^{-1} {\bf d}_{\mu, S_j})^2 \leq D_{\text{DAG}}^2.
\]
Since $d\log p_0 = O(n_1 + n_2)$, we then have
\[
J_{31} = O\left( \frac{D_{\text{KL}}^2}{n_1 + n_2 - d - 6} + \frac{D_{\text{DAG}}^2}{(n_1 + n_2)^2 (n_1 + n_2 - d - 6)}\right).
\]
Also, using the double expectation trick, it is easy to see that $J_{32} = 0$. Therefore, 
\[
J_3 = O\left( \frac{D_{\text{KL}}^2}{n_1 + n_2 - d - 4} + \frac{D_{\text{DAG}}^2}{(n_1 + n_2)^2 (n_1 + n_2 - d - 6)}\right).
\]
For $J_4$, we have 
\begin{align*}
    J_4 = &~ \sum_{j=1}^p \mathbb{E} \bigg[ r_j^{-2} \{\{{\bf d}_u^\intercal ({\bf e}_j - \widehat{\bf q}_j) \}^2  - \{{\bf d}_u^\intercal ({\bf e}_j - {\bf q}_j) \}^2\}^2  \bigg] \nonumber \\
    + &~ 2 \sum_{k < j} \mathbb{E} \bigg[ r_j^{-1}  [\{{\bf d}_u^\intercal ({\bf e}_j - \widehat{\bf q}_j) \}^2  - \{{\bf d}_u^\intercal ({\bf e}_j - {\bf q}_j) \}^2] r_k^{-1}  [\{{\bf d}_u^\intercal ({\bf e}_k - \widehat{\bf q}_k) \}^2  - \{{\bf d}_\mu^\intercal ({\bf e}_k - {\bf q}_k) \}^2] \bigg] \nonumber \\
    = &~ J_{41} + 2J_{42}.
\end{align*}
For $J_{41}$, using similar tricks, we can show that
\begin{align*}
    J_{41} \leq &~ \frac{8}{n_1 + n_2}\sum_{j \in \mathcal{J}} r_j^{-1}  \{{\bf d}_u^\intercal ({\bf e}_j - {\bf q}_j) \}^2 {\bf d}_{\mu, S_j}^\intercal \{{\Sigma}_{S_j, S_j}\}^{-1} {\bf d}_{\mu, S_j} \times O\left( 1+ \frac{d\log p_0}{n_1 + n_2}\right) \nonumber \\
    + &~ \frac{6}{(n_1 + n_2)^2} \sum_{j \in \mathcal{J}} ({\bf d}_{\mu, S_j}^\intercal \{{\Sigma}_{S_j, S_j}\}^{-1} {\bf d}_{\mu, S_j})^2 \times O\left( 1+ (\frac{d\log p_0}{n_1 + n_2})^2\right) \nonumber \\
    = &~ O \left( \frac{D_{\text{DL}} D_{\text{DAG}}}{n_1 + n_2} + \frac{D_{\text{DAG}}^2}{(n_1 + n_2)^2} \right).
\end{align*}
and 
\begin{align*}
    2J_{42} = &~ \frac{1}{(n_1 + n_2)^2} \sum_{k \neq j} {\bf d}_{\mu, S_j}^\intercal \{{\Sigma}_{S_j, S_j}\}^{-1} {\bf d}_{\mu, S_j} {\bf d}_{\mu, S_k}^\intercal \{{\Sigma}_{S_k, S_k}\}^{-1} {\bf d}_{\mu, S_k} \times O\left( (1+ \frac{d\log p_0}{n_1 + n_2})^2\right) \nonumber \\
    = &~ O\left( \frac{D_{\text{DAG}}^2}{(n_1 + n_2)^2} \right).
\end{align*}
Thus, since $D_{\text{DAG}} \asymp p^\beta, D_{\text{KL}} \asymp p^{1/2} N^{-\gamma}$ and $d = o(n_1 + n_2)$ for $\gamma > 1/2$ and $0 < \beta < 1/2$, combining the results for $J_3$ and $J_4$, 
one can see that 
\[
\frac{N^2}{2p}\mathbb{E} \bigg[ \{ {\bf d}_\mu^\intercal (\widehat{\Sigma}_{\text{DAG}}^{-1} - \Sigma^{-1}) {\bf d}_\mu \} ^2\bigg] = o(1).
\]
Therefore, 
\[
\frac{N}{\sqrt{2p}} \left( {\bf d}_\mu^\intercal (\widehat{\Sigma}_{\text{DAG}}^{-1} - \Sigma^{-1}) {\bf d}_\mu \right) = o_p(1).
\]

We now focus on ${\bf d}_\mu^\intercal \widehat{\Sigma}_{\text{DAG}}^{-1}{\bf y}_d$. Similarly, we can write 
\[
{\bf d}_\mu^\intercal \widehat{\Sigma}_{\text{DAG}}^{-1}{\bf y}_d = {\bf d}_\mu^\intercal {\Sigma}^{-1}{\bf y}_d + {\bf d}_\mu^\intercal (\widehat{\Sigma}_{\text{DAG}}^{-1} - {\Sigma}^{-1}){\bf y}_d = J_{5} + J_{6}
\]
For $J_{5}$, it is easy to see that $J_5 \sim N(0, N^{-1}D_{\text{KL}})$; this indicates that $NJ_5/\sqrt{2p} = o_p(1)$, since $D_{\text{KL}} = O( p^{1/2} N^{-\gamma})$ for $\gamma > 1/2$.
For $J_6$, we first write
\begin{align*}
    J_6 = \sum_{j=1}^p \widehat{r}_j^{-1} {\bf d}_\mu^\intercal ({\bf e}_j - \widehat{\bf q}_j) ({\bf e}_j - \widehat{\bf q}_j)^\intercal {\bf y}_d - \sum_{j=1}^p {r}_j^{-1} {\bf d}_\mu^\intercal ({\bf e}_j - {\bf q}_j) ({\bf e}_j - {\bf q}_j)^\intercal {\bf y}_d.
\end{align*}
Again, using the double expectation trick and Lemmas \ref{covariance} and \ref{independence}, we have
\begin{align*}
    \mathbb{E}[J_6] = &~ \sum_{j=1}^p \mathbb{E} \bigg[ \mathbb{E}\big[\widehat{r}_j^{-1} \mid A_{S_j} \big] {\bf d}_\mu^\intercal \mathbb{E} \big[ ({\bf e}_j - \widehat{\bf q}_j) ({\bf e}_j - \widehat{\bf q}_j)^\intercal {\bf y}_d \mid A_{S_j} \big] \bigg] \nonumber \\
    = &~ -2 \sum_{j=1}^p r_j^{-1} \mathbb{E} \bigg[ {\bf d}_\mu^\intercal ({\bf e}_j - {\bf q}_j) (\widehat{\bf q}_j - {\bf q}_j)^\intercal {\bf y}_d \bigg] + \sum_{j=1}^p r_j^{-1} \mathbb{E} \bigg[ {\bf d}_\mu^\intercal (\widehat{\bf q}_j - {\bf q}_j)(\widehat{\bf q}_j - {\bf q}_j)^\intercal {\bf y}_d \bigg] \nonumber \\
    = &~ \frac{1}{n_1 + n_2} \sum_{j \in \mathcal{J}} \mathbb{E} \bigg[ {\bf d}_{\mu, S_j}^\intercal \{\widehat{\Sigma}_{S_j, S_j}\}^{-1} {\bf y}_{d,S_j} \bigg] \nonumber \\
    = &~  \frac{1}{(n_1 + n_2)N} \sum_{j \in \mathcal{J}} {\bf d}_{\mu, S_j}^\intercal \{\Sigma_{S_j, S_j}\}^{-1} {\bf d}_{\mu, S_j} \times O\left( 1+ \frac{d\log p_0}{n_1 + n_2} \right) \nonumber \\
    = &~ O\left( \frac{D_{\text{DAG}}}{(n_1 + n_2)N} \right).
\end{align*}
Since $D_{\text{DAG}} = O(p^{\beta})$ for $0 < \beta < 1/2$, one can see that 
$\sqrt{N}\mathbb{E}(J_6)/\sqrt{p} = o(1)$. 
Similarly, 
\begin{align*}
    \mathbb{E} [J_6^2] \leq &~ 2\mathbb{E}\bigg[ \{\sum_{j=1}^p (\widehat{r}_j^{-1} - r_j^{-1}) {\bf d}_\mu^\intercal ({\bf e}_j - {\bf q}_j) ({\bf e}_j - {\bf q}_j)^\intercal {\bf y}_d \}^2  \bigg] \nonumber \\
    +&~ 16 \mathbb{E} \bigg[  \{\sum_{j=1}^p \widehat{r}_j^{-1}  {\bf d}_\mu^\intercal ({\bf e}_j - {\bf q}_j) (\widehat{\bf q}_j - {\bf q}_j)^\intercal {\bf y}_d \}^2
    \bigg] \nonumber \\
    + &~ 4 \mathbb{E} \bigg[ \{\sum_{j=1}^p \widehat{r}_j^{-1}  {\bf d}_\mu^\intercal (\widehat{\bf q}_j - {\bf q}_j) (\widehat{\bf q}_j - {\bf q}_j)^\intercal {\bf y}_d \}^2 \bigg]. \nonumber \\
    = &~ 2J_{61} + 16J_{62} + 4J_{63}.
\end{align*}
For $J_{61}$, we have 
\begin{align*}
    J_{61} \leq &~ \frac{1}{N(n_1 + n_2 - d - 6)}\sum_{j=1}^p r_j^{-2} \{{\bf d}_\mu^\intercal ({\bf e}_j - {\bf q}_j)\}^2 ({\bf e}_j - {\bf q}_j)^\intercal \Sigma ({\bf e}_j - {\bf q}_j) \nonumber \\
\leq &~ \frac{D_{\text{KL}}} {N(n_1 + n_2 - d - 6)}.
\end{align*}
Here, we use the fact that $({\bf e}_j - {\bf q}_j)^\intercal \Sigma ({\bf e}_j - {\bf q}_j)= r_j$.
Using similar tricks, we have
\begin{align*}
    J_{62} \leq \frac{D_{\text{KL}} d}{(n_1 + n_2)N} \times O\left( 1+ \frac{d\log p_0}{n_1 + n_2} \right). 
\end{align*}
and
\begin{align*}
    J_{63} = &~ \sum_{j=1}^p r_j^{-2} (1 + \frac{1}{n_1 + n_2 - |S_j| - 6}) \mathbb{E} \bigg[  \{ {\bf d}_\mu^\intercal (\widehat{\bf q}_j - {\bf q}_j) (\widehat{\bf q}_j - {\bf q}_j)^\intercal {\bf y}_d \}^2
    \bigg] \nonumber \\
    + &~ \frac{2}{(n_1 + n_2)^2} \sum_{k < j} \mathbb{E} \bigg[ {\bf d}_{\mu, S_j}^\intercal \{\widehat{\Sigma}_{S_j, S_j}\}^{-1} {\bf y}_{d, S_j}  {\bf d}_{\mu, S_k}^\intercal \{\widehat{\Sigma}_{S_k, S_k}\}^{-1} {\bf y}_{d, S_k} \bigg] \nonumber \\
    = &~ J_{631} + J_{632}.
\end{align*}
For $J_{631}$, recall that $\ve \theta_{S_j} = \{\widehat{\Sigma}_{S_j, S_j}\}^{1/2} (\widehat{\bf q}_j - {\bf q}_j)_{S_j}$. For ease of notations, let $\ve \psi_{j} =\{\widehat{\Sigma}_{S_j, S_j}\}^{-1/2} {\bf d}_{\mu,S_j} $ and $\ve \phi_j = \{\widehat{\Sigma}_{S_j, S_j}\}^{-1/2}{\bf y}_{d, S_j}$. Then, using the fact that $\ve \theta_{S_j} \mid A_{S_j} \sim N(0, r_j (n_1 + n_2)^{-1})$, we have
    \begin{align*}
       &~  \mathbb{E} \bigg[  \{ {\bf d}_\mu^\intercal (\widehat{\bf q}_j - {\bf q}_j) (\widehat{\bf q}_j - {\bf q}_j)^\intercal {\bf y}_d \}^2 \bigg] =  \mathbb{E} \bigg[ \sum_{s, t, u, v \in \mathcal{S}_j} \phi_{js}\phi_{jt} \psi_{ju}\psi_{jv} \theta_{js}\theta_{jt}\theta_{ju}\theta_{jv}.
       \bigg] \nonumber \\
       \leq &~ \frac{4r_j^2}{(n_1 + n_2)^2} \mathbb{E} \bigg[  {\bf d}_{\mu, S_j}^\intercal \{\widehat{\Sigma}_{S_j, S_j}\}^{-1} {\bf d}_{\mu, S_j} \times 
       {\bf y}_{d, S_j}^\intercal \{\widehat{\Sigma}_{S_j, S_j}\}^{-1} {\bf y}_{d, S_j}\bigg] \nonumber \\
       + &~ \frac{2r_j^2}{(n_1 + n_2)^2} \mathbb{E} \bigg[ \left( {\bf d}_{\mu, S_j}^\intercal \{\widehat{\Sigma}_{S_j, S_j}\}^{-1} {\bf y}_{d, S_j} \right)^2\bigg].
    \end{align*}
  Since ${\bf y}_{d,S_j} \sim N(0, N^{-1} \Sigma_{S_j, S_j})$, by some algebra, one can see that 
  \begin{align*}
      J_{631} \leq &~ \frac{4d}{N (n_1 + n_2)^2} (1 + \frac{1}{n_1 + n_2 - d - 4}) \sum_{j \in \mathcal{J}}{\bf d}_{\mu, S_j}^\intercal \{\Sigma_{S_j, S_j}\}^{-1} {\bf d}_{\mu, S_j} \times O\left( (1 + \frac{d\log p_0}{n_1 + n_2})^2 \right) \nonumber \\
      + &~ \frac{2}{N(n_1 + n_2)^2} \sum_{j \in \mathcal{J}}{\bf d}_{\mu, S_j}^\intercal \{\Sigma_{S_j, S_j}\}^{-1} {\bf d}_{\mu, S_j} \times O\left( (1 + \frac{d\log p_0}{n_1 + n_2})^2 \right) \nonumber \\
      = &~ O\left( \frac{4d D_{\text{DAG}}  }{N(n_1 + n_2)^2} \right).
  \end{align*}
  Similarly, 
  \begin{align*}
      J_{632} \leq &~ \frac{1}{N(n_1 + n_2)^2} \sum_{j \in \mathcal{J}}{\bf d}_{\mu, S_j}^\intercal \{\Sigma_{S_j, S_j}\}^{-1} {\bf d}_{\mu, S_j} \times O\left( (1 + \frac{d\log p_0}{n_1 + n_2})^2 \right) \nonumber \\
      = &~ \frac{D_{\text{DAG}}}{N(n_1 + n_2)^2}. 
  \end{align*}
Since $D_{\text{KL}} = O(p^{1/2} N^{-\gamma})$, $D_{\text{DAG}} = O(p^\beta)$ and $d = o(n_1 + n_2)$ for $\gamma > 1/2$ and $0 < \beta < 1/2$, one can see that 
  $N^2\mathbb{E}(J_6^2)/p = o(1)$. 
Combining all the assertions above, as $n_1, n_2 \rightarrow \infty$,
 we have
  \begin{align*}
  \mathbb{P}\left( \left|\frac{T^2(\widehat{Q}, \widehat{P}) - p}{\sqrt{2p}}\right| \geq z_{\alpha/2} \right) \geq \mathbb{P} \left( |Z| \leq \frac{ND_{\text{KL}}}{\sqrt{2p}} - z_{\alpha/2} \right). 
\end{align*}
Furthermore, if $1/2 < \gamma < 1$, then
\[
 \mathbb{P} \left( |Z| \leq \frac{ND_{\text{KL}}}{\sqrt{2p}} - z_{\alpha/2} \right) = 1 \mbox{ as } n_1, n_2 \rightarrow \infty.
\]
This completes the proof.   
\end{proof}

\section{More information on the simulation studies}\label{supp2}

\subsection{Additional simulation results assuming accurate edge information}\label{sim.sec1}
\begin{figure}[ht!]
\centering
    \includegraphics[width=\textwidth]{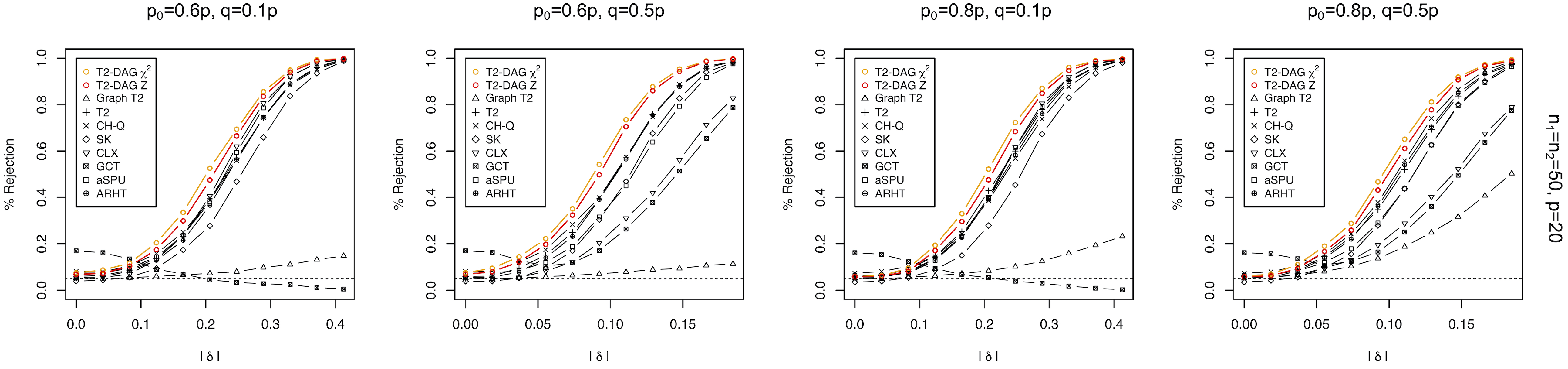}
	\caption*{Figure S1. Empirical type I error rates ($\delta = 0$) and powers ($|\delta| > 0)$ under $\bf{n_1=n_2=50}$ and $\bf{p=20}$, with the number of children nodes in the DAG set to $p_0=0.4p$ (columns 1 and 2) or $p_0=0.8p$ (columns 3 and 4), and the number of non-zero signals set to $q = 0.1p$ (columns 1 and 3) or $0.5p$ (columns 2 and 4).}
\label{suppfig1}
\end{figure}
\FloatBarrier
\noindent

\begin{figure}[ht!]
\centering
    \includegraphics[width=\textwidth]{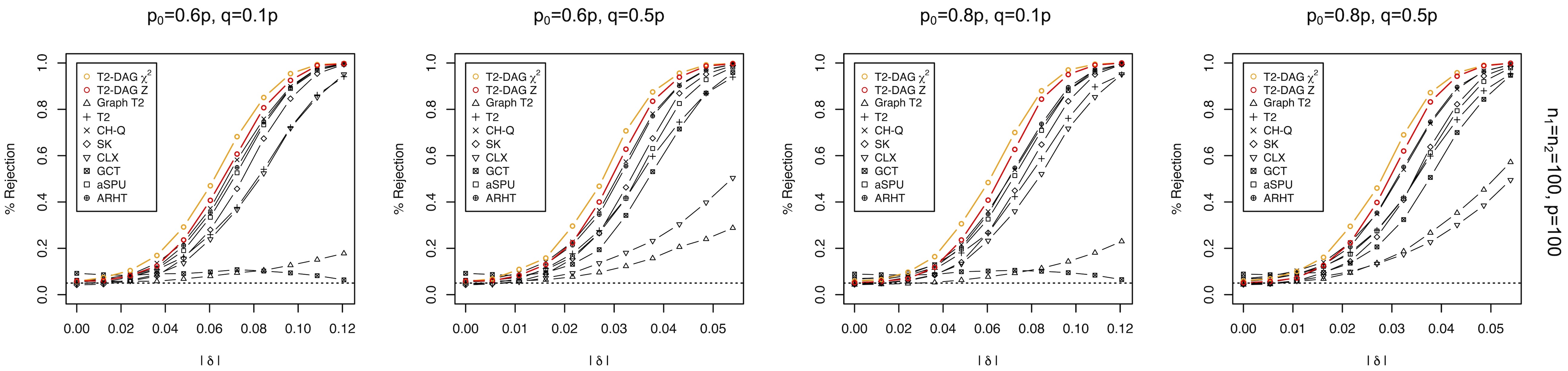}
	\caption*{Figure S2. Empirical type I error rates ($\delta = 0$) and powers ($|\delta| > 0)$ under $\bf{n_1=n_2=100}$ and $\bf{p=100}$, with 
	$p_0=0.4p$ (columns 1 and 2) or $p_0=0.8p$ (columns 3 and 4), and 
	$q = 0.1p$ (columns 1 and 3) or $0.5p$ (columns 2 and 4).}
\label{suppfig2}
\end{figure}
\FloatBarrier
\noindent

\begin{figure}[ht!]
\centering
    \includegraphics[width=\textwidth]{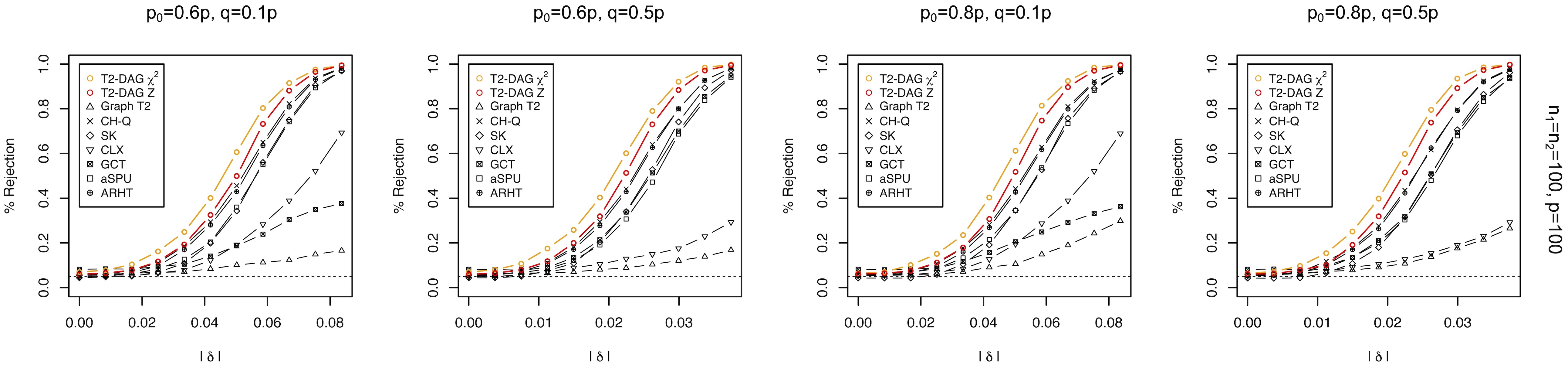}
	\caption*{Figure S3. Empirical type I error rates ($\delta = 0$) and powers ($|\delta| > 0)$ under $\bf{n_1=n_2=100}$ and $\bf{p=300}$, with 
	$p_0=0.4p$ (columns 1 and 2) or $p_0=0.8p$ (columns 3 and 4), and 
	$q = 0.1p$ (columns 1 and 3) or $0.5p$ (columns 2 and 4).}
\label{suppfig3}
\end{figure}
\FloatBarrier

\subsection{Additional simulation results given mis-specified edge information}
\begin{figure}[ht!]
\centering
\includegraphics[width=0.7\textwidth]{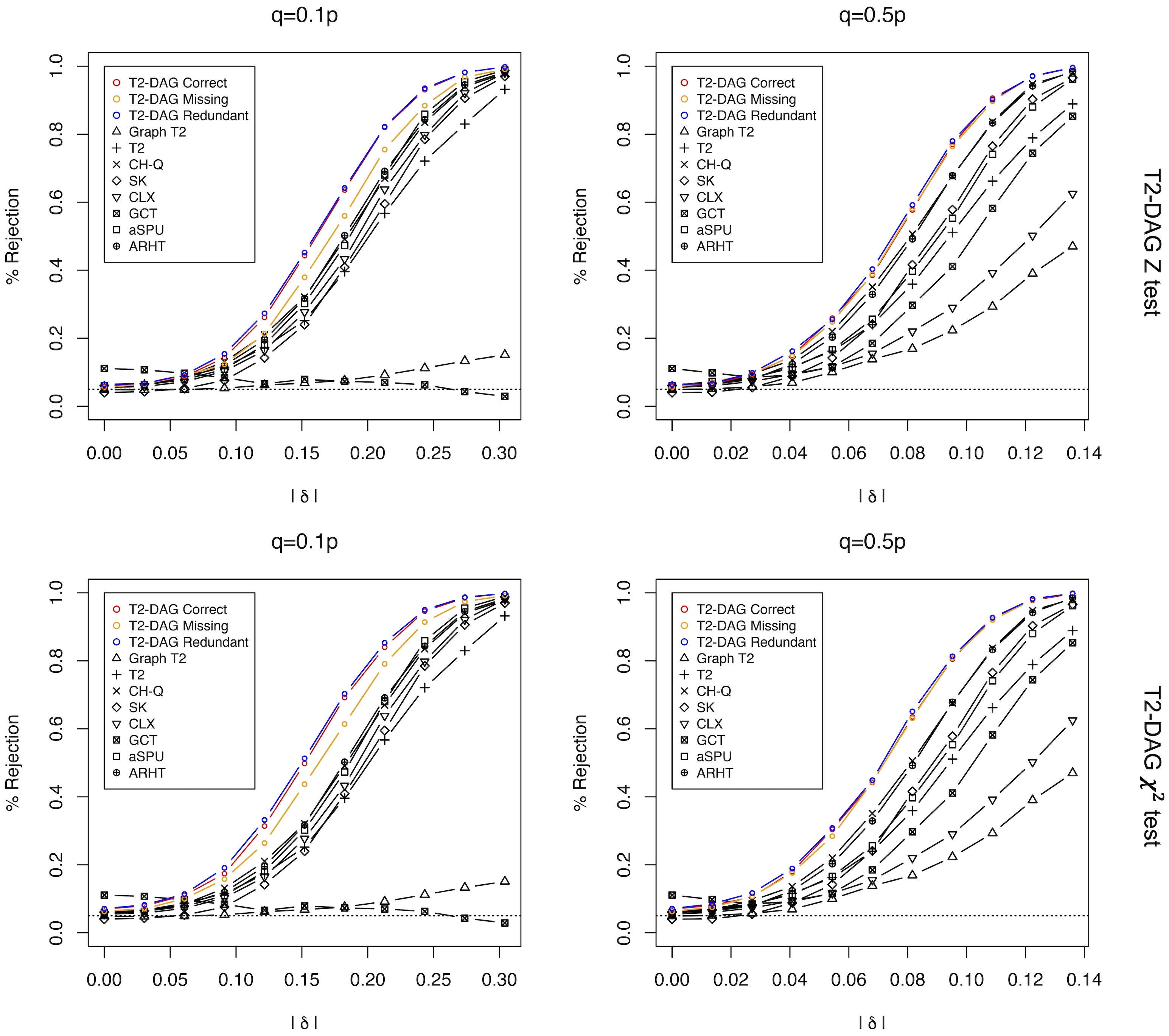}
	\caption*{Figure S4. Empirical type I error rates ($\delta = 0$) and powers  ($|\delta| > 0)$ of T2-DAG $Z$ (top panel) and T2-DAG $\chi^2$ (bottom panel) under $n_1=n_2=50$ and $p=50$, with 
	$p_0=0.8p$, and 
	$q = 0.1p$ (left panel) or $0.5p$ (right panel).
	The T2-DAG was applied based on $A$ (``T2-DAG Correct''), $A_1$ (``T2-DAG Missing''), or $A_2$ (``T2-DAG Redundant''). The Graph T2 was applied based on the true adjacency matrix $A$.}
\label{suppfig4}
\end{figure}
\FloatBarrier
\noindent

\begin{figure}[ht!]
\centering
\includegraphics[width=0.7\textwidth]{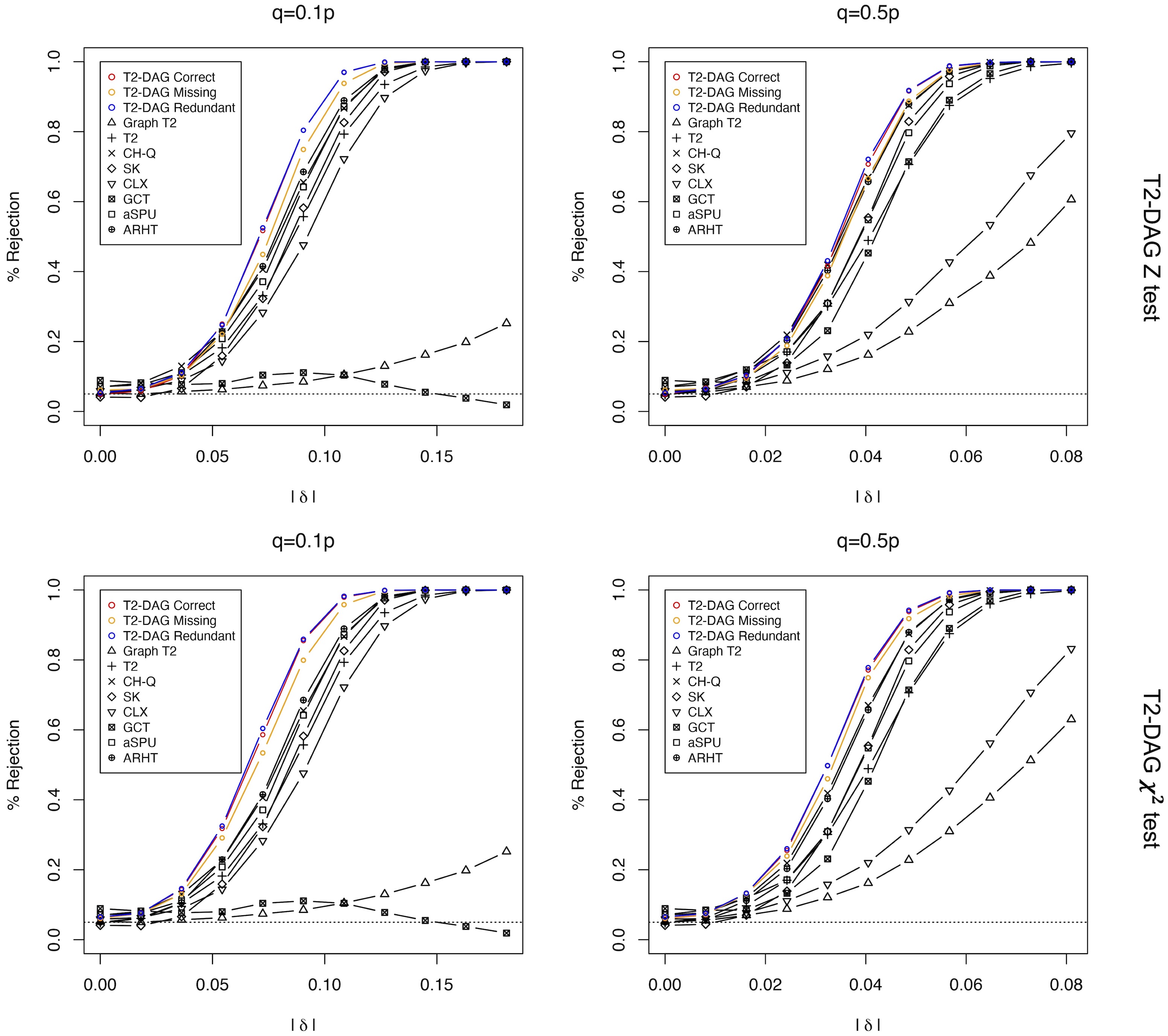}
	\caption*{Figure S5. Empirical type I error rates ($\delta = 0$) and powers ($|\delta| > 0)$ of T2-DAG $Z$ (top panel) and T2-DAG $\chi^2$ (bottom panel) under $n_1=n_2=100$ and $p=100$, with 
	$p_0=0.8p$, and 
	$q = 0.1p$ (left panel) or $0.5p$ (right panel).  
	The T2-DAG was applied based on $A$ (``T2-DAG Correct''), $A_1$ (``T2-DAG Missing''), or $A_2$ (``T2-DAG Redundant''). The Graph T2 was applied based on the true adjacency matrix $A$.}
\label{suppfig5}
\end{figure}
\FloatBarrier

\begin{figure}[ht!]
\centering
\includegraphics[width=0.7\textwidth]{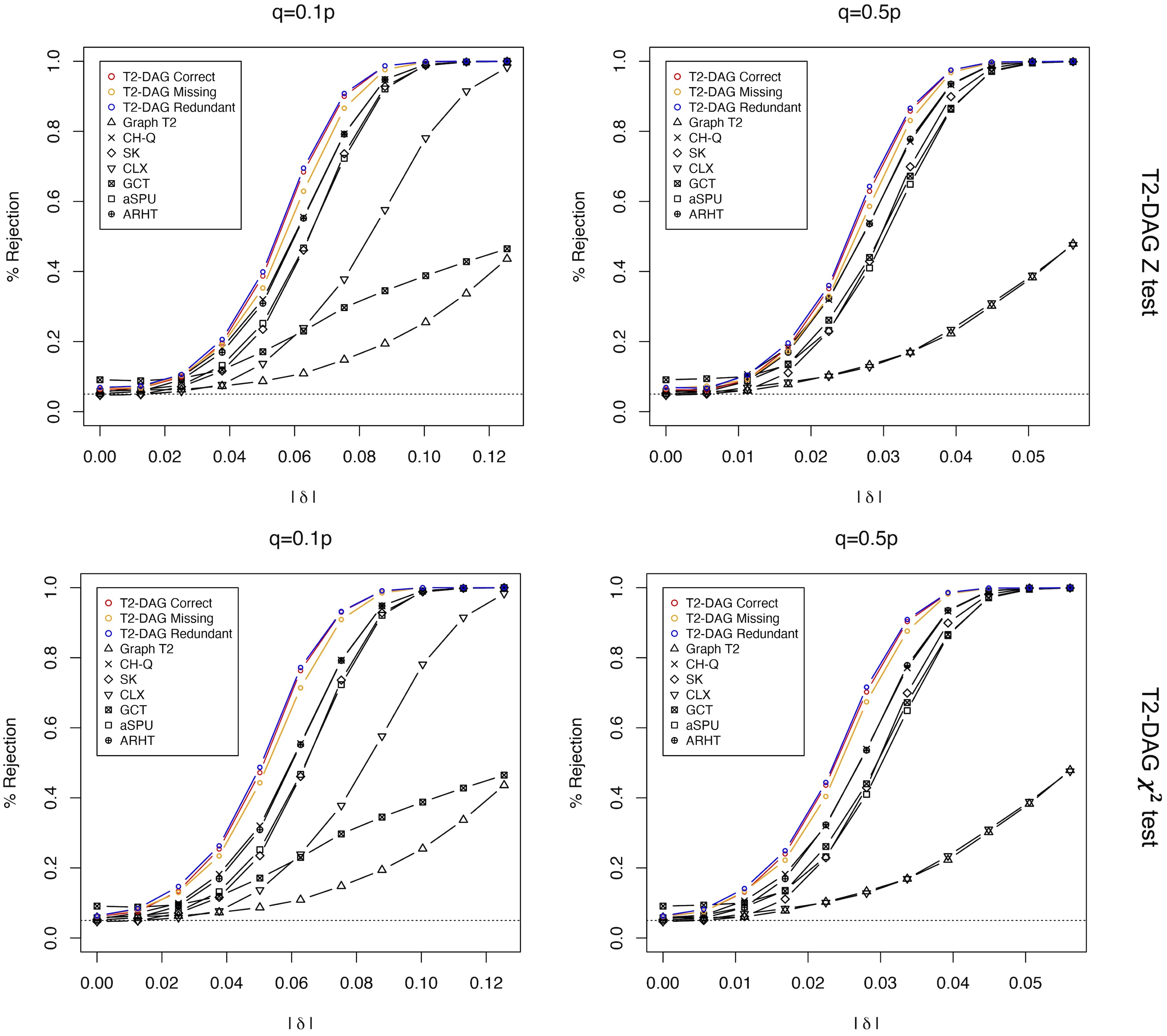}
	\caption*{Figure S6. Empirical type I error rates ($\delta = 0$) and powers ($|\delta| > 0)$ of T2-DAG $Z$ (top panel) and T2-DAG $\chi^2$ (bottom panel) under $n_1=n_2=100$ and $p=300$, with 
	$p_0=0.8p$, and 
	$q = 0.1p$ (left panel) or $0.5p$ (right panel).  
	The T2-DAG was applied based on $A$ (``T2-DAG Correct''), $A_1$ (``T2-DAG Missing''), or $A_2$ (``T2-DAG Redundant''). The Graph T2 was applied based on the true adjacency matrix $A$.}
\label{suppfig6}
\end{figure}
\FloatBarrier
\noindent

\subsection{Additional simulation results in the presence of unadjusted confounding effects}
\begin{figure}[ht!]
\centering
    \includegraphics[width=0.9\textwidth]{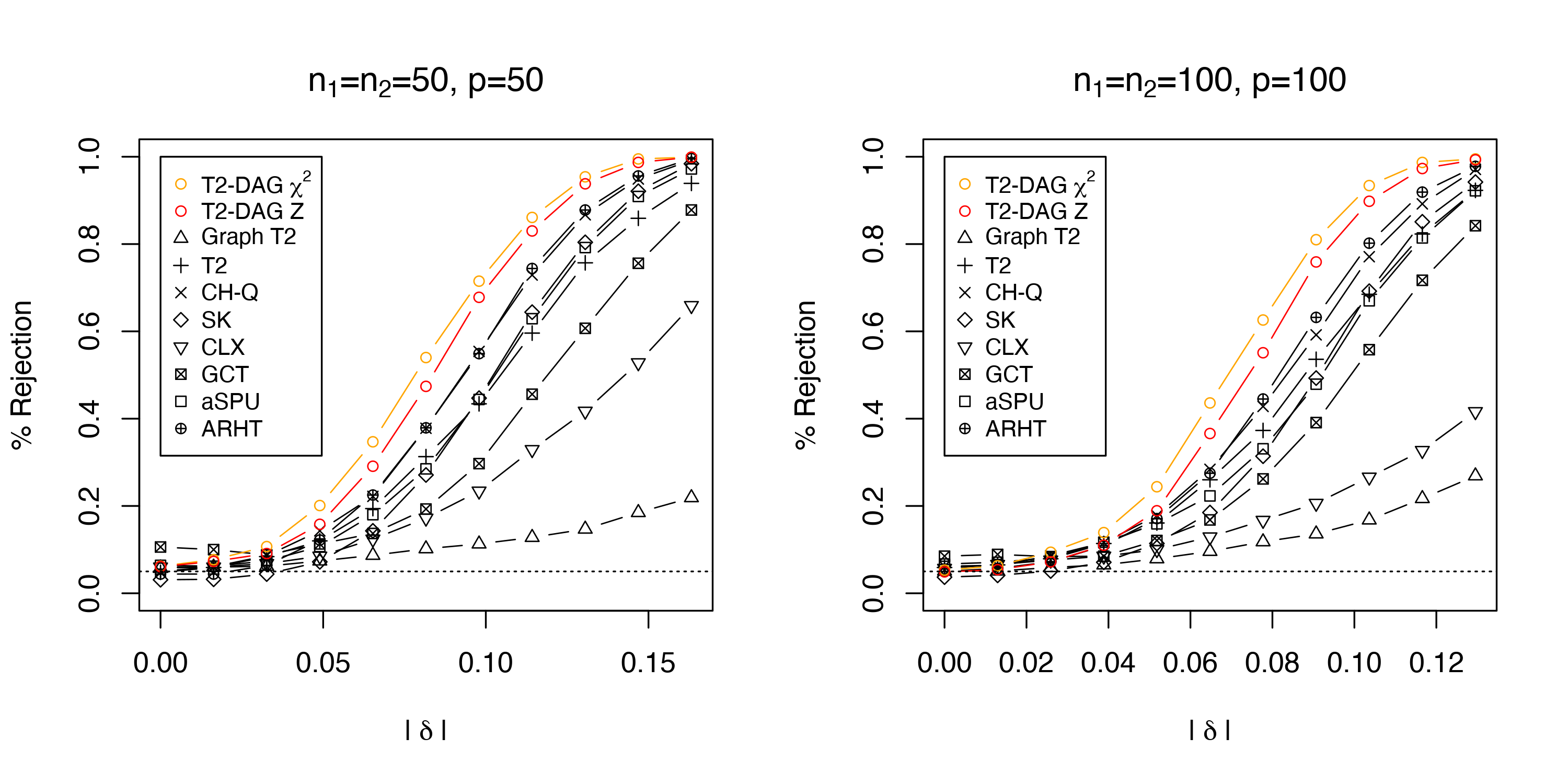}
\caption*{Figure S7. Empirical type I error rates ($\delta = 0$) and powers ($|\delta| > 0)$ of the various tests under $n_1=n_2=50$ and $p=50$ (column 1) or $n_1=n_2=100$ and $p=100$ (column 2), assuming two continuous confounders, with 
$p_0=0.8p$, and 
$q = 0.5p$.}
\label{fig.confounder}
\end{figure}
\FloatBarrier
\noindent

\subsection{Additional simulation results given non-Gaussian error terms}
We conducted a sensitivity analysis on T2-DAG in terms of the distribution of the error terms. Specifically, we considered three alternative non-Gaussian distributions: (1) uniform distribution, where $\epsilon_{ij}$ was generated from $\text{Unif}(-\sqrt(3r_j), \sqrt(3r_j))$, with a zero mean and variance equal to $r_j$ as in the simulation setting in Section 3.2 of the main manuscript; (2) Gamma distribution, where $\epsilon_{ij}$ was generated from a Gamma distribution with a shape parameter $s=10$ and a scale parameter $\sqrt{r_j/s}$ and then subtracted by its mean, $s\sqrt{r_j/s}$, which gave a zero mean and variance equal to $r_j$; and (3) log-normal distribution, where $\epsilon_{ij}$ was generated from a log-normal distribution with zero mean and variance $\tau^2=0.16$ on the log scale and then subtracted by its mean, $\exp(\tau^2/2)$.

\begin{figure}[ht!]
\centering
    \includegraphics[width=\textwidth]{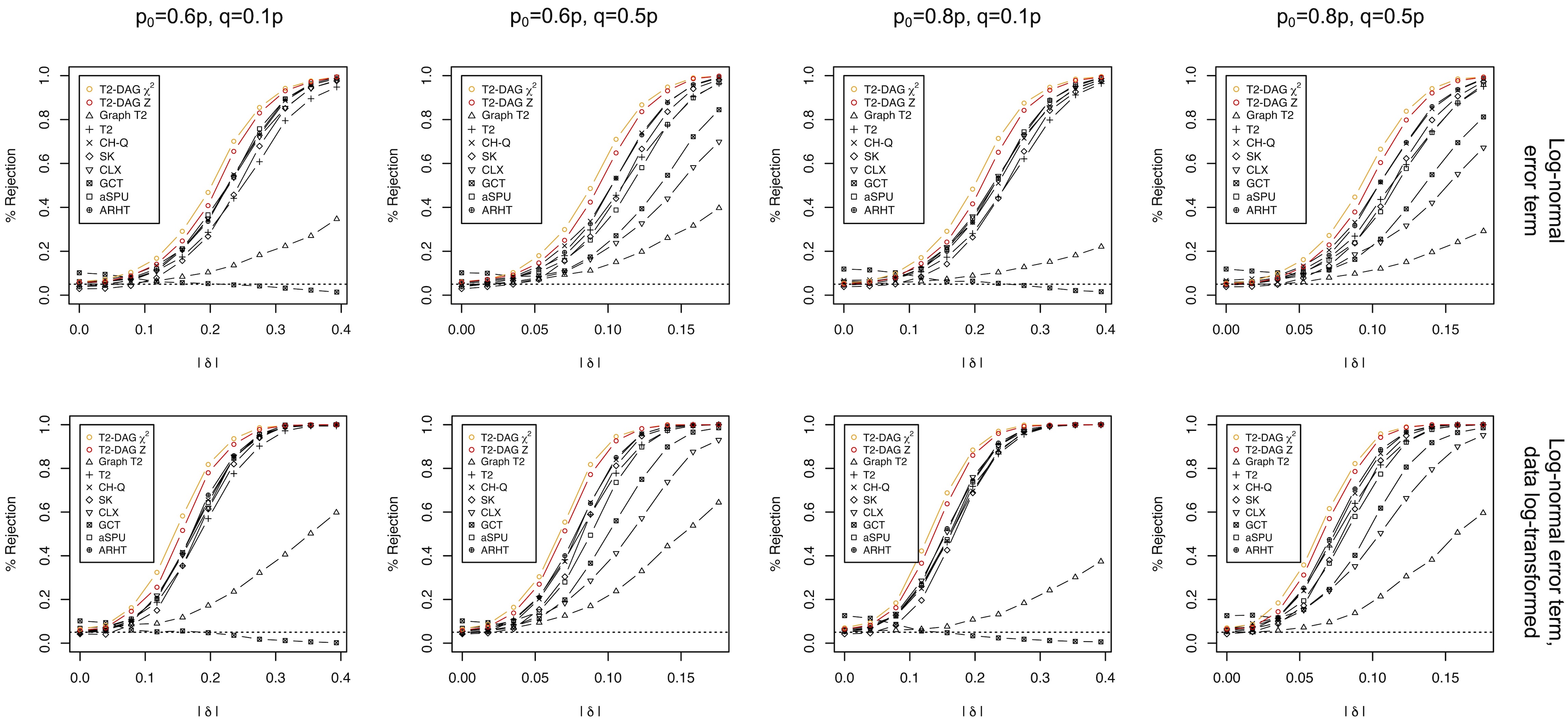}
\caption*{Figure S8. Empirical type I error rates ($\delta = 0$) and powers ($|\delta| > 0)$ of the various tests under $\bf{n_1=n_2=50}$ and $\bf{p=40}$, with error terms generated from \textbf{log-normal} distributions based on the raw data (row 1) or log-transformed data (row 2).
The number of children nodes in the DAG is set to $p_0=0.6p$ (columns 1 and 2) or $p_0=0.8p$ (columns 3 and 4), and the number of non-zero signals is set to $q = 0.1p$ (columns 1 and 3) or $0.5p$ (columns 2 and 4).}
\end{figure}
\FloatBarrier
\noindent

\begin{figure}[ht!]
\centering
    \includegraphics[width=\textwidth]{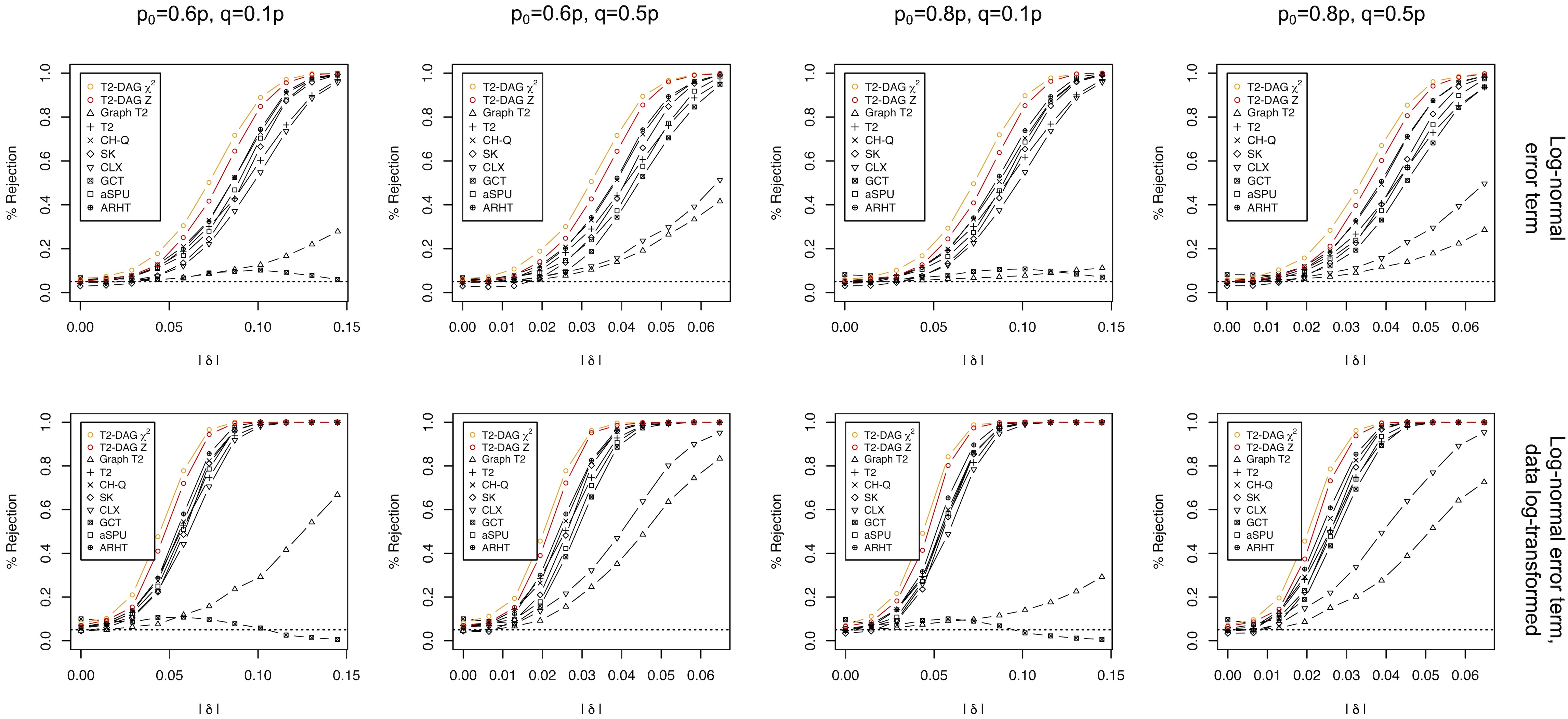}
\caption*{Figure S9. Empirical type I error rates ($\delta = 0$) and powers ($|\delta| > 0)$ of the various tests under $\bf{n_1=n_2=100}$ and $\bf{p=100}$, with error terms generated from \textbf{log-normal} distributions based on the raw data (row 1) or log-transformed data (row 2).
The number of children nodes in the DAG is set to $p_0=0.6p$ (columns 1 and 2) or $p_0=0.8p$ (columns 3 and 4), and the number of non-zero signals is set to $q = 0.1p$ (columns 1 and 3) or $0.5p$ (columns 2 and 4).}
\end{figure}
\FloatBarrier
\noindent


\begin{figure}[ht!]
\centering
    \includegraphics[width=\textwidth]{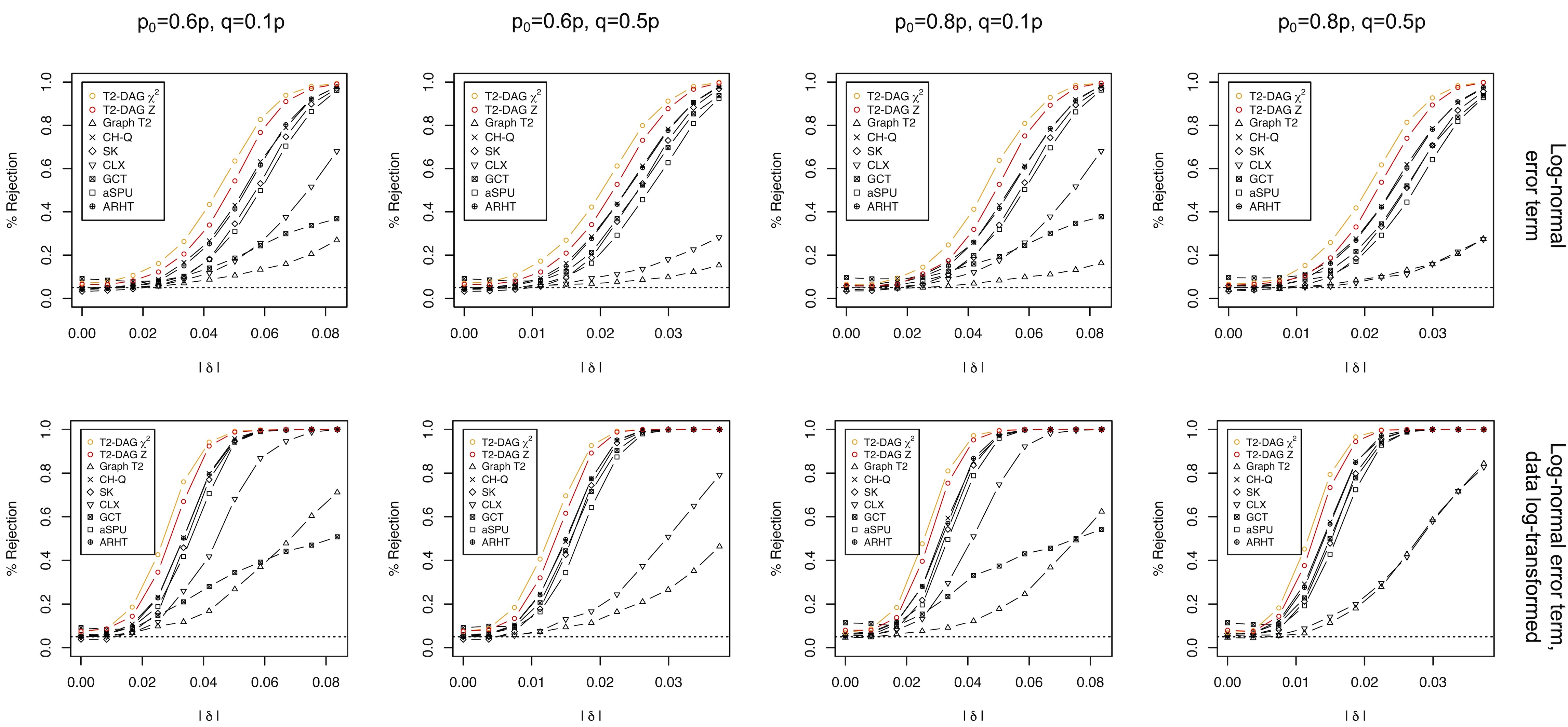}
\caption*{Figure S10. Empirical type I error rates ($\delta = 0$) and powers ($|\delta| > 0)$ of the various tests under $\bf{n_1=n_2=100}$ and $\bf{p=300}$, with error terms generated from \textbf{log-normal} distributions based on the raw data (row 1) or log-transformed data (row 2).
The number of children nodes in the DAG is set to $p_0=0.6p$ (columns 1 and 2) or $p_0=0.8p$ (columns 3 and 4), and the number of non-zero signals is set to $q = 0.1p$ (columns 1 and 3) or $0.5p$ (columns 2 and 4).}
\end{figure}
\FloatBarrier
\noindent

\begin{figure}[ht!]
\centering
    \includegraphics[width=\textwidth]{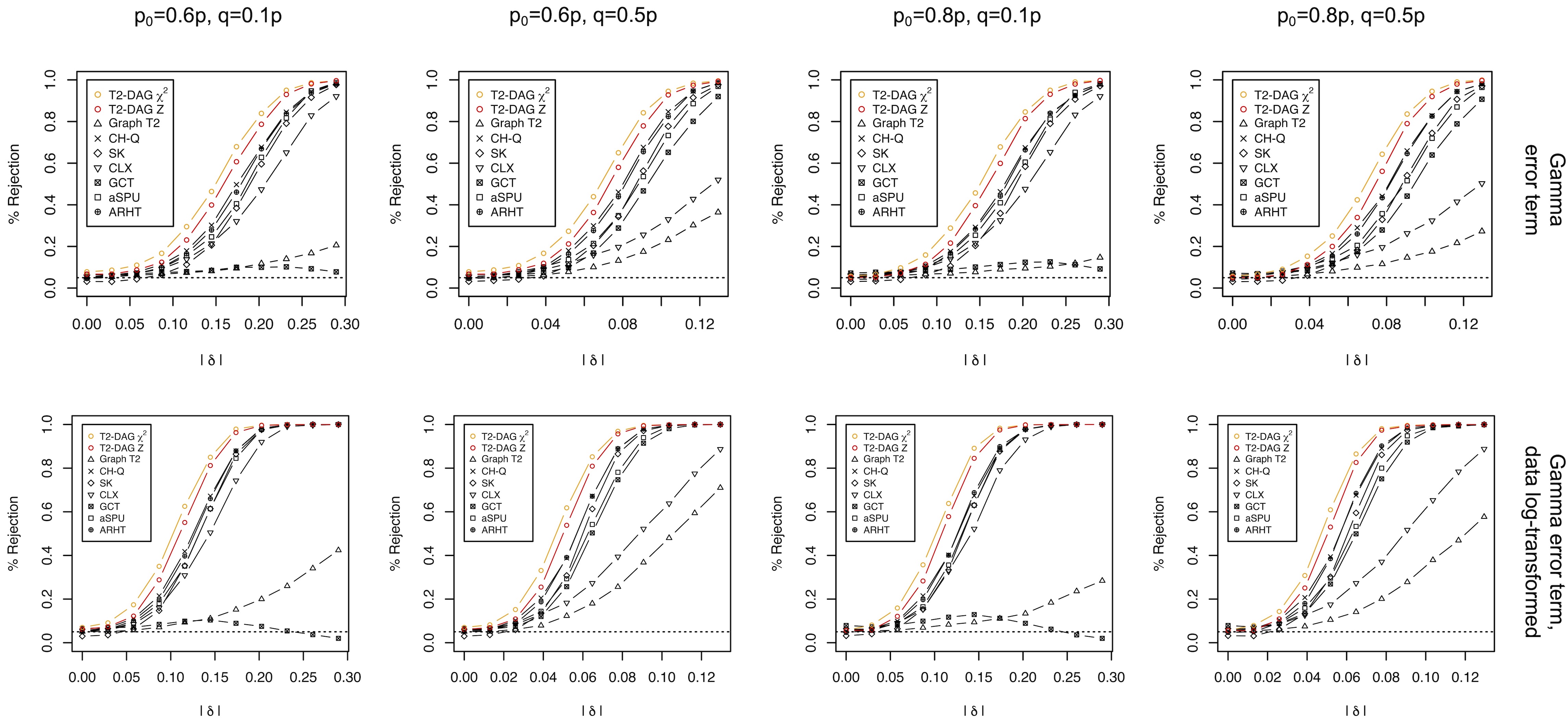}
\caption*{Figure S11. Empirical type I error rates ($\delta = 0$) and powers ($|\delta| > 0)$ of the various tests under $\bf{n_1=n_2=50}$ and $\bf{p=40}$, with error terms generated from \textbf{gamma} distributions based on the raw data (row 1) or log-transformed data (row 2).
The number of children nodes in the DAG is set to $p_0=0.6p$ (columns 1 and 2) or $p_0=0.8p$ (columns 3 and 4), and the number of non-zero signals is set to $q = 0.1p$ (columns 1 and 3) or $0.5p$ (columns 2 and 4).}
\end{figure}
\FloatBarrier
\noindent


\begin{figure}[ht!]
\centering
    \includegraphics[width=\textwidth]{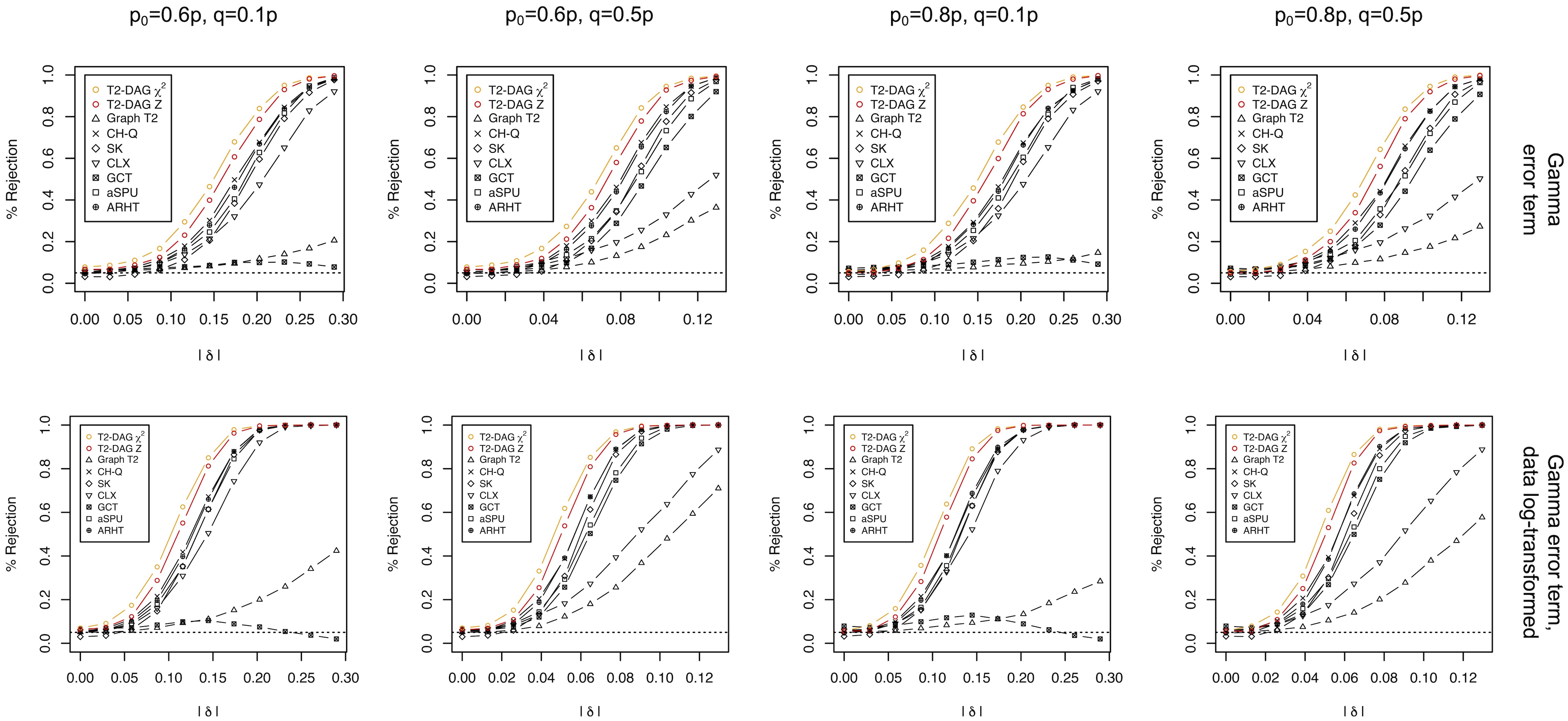}
\caption*{Figure S12. Empirical type I error rates ($\delta = 0$) and powers ($|\delta| > 0)$ of the various tests under $\bf{n_1=n_2=50}$ and $\bf{p=100}$, with error terms generated from \textbf{gamma} distributions based on the raw data (row 1) or log-transformed data (row 2).
The number of children nodes in the DAG is set to $p_0=0.6p$ (columns 1 and 2) or $p_0=0.8p$ (columns 3 and 4), and the number of non-zero signals is set to $q = 0.1p$ (columns 1 and 3) or $0.5p$ (columns 2 and 4).}
\end{figure}
\FloatBarrier
\noindent

\begin{figure}[ht!]
\centering
    \includegraphics[width=\textwidth]{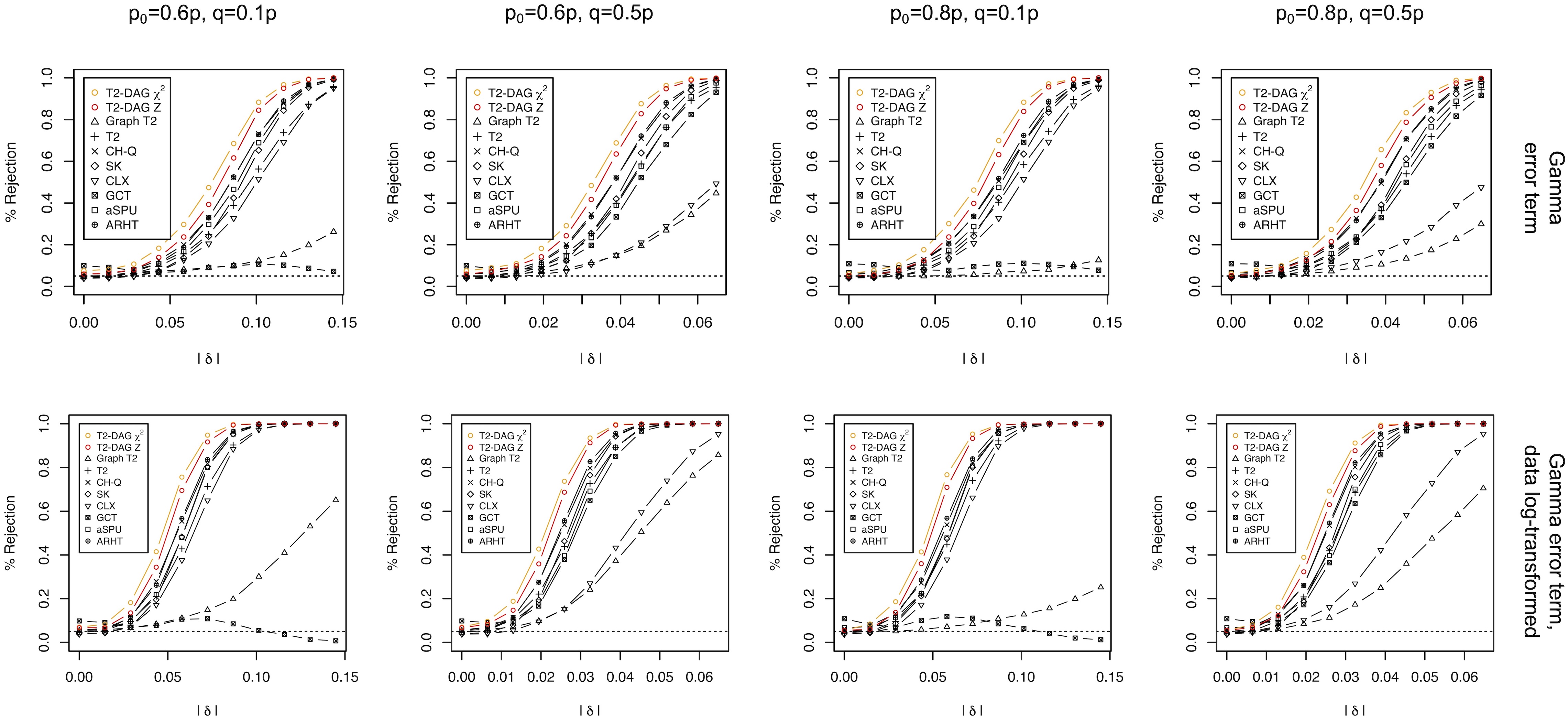}
\caption*{Figure S13. Empirical type I error rates ($\delta = 0$) and powers ($|\delta| > 0)$ of the various tests under $\bf{n_1=n_2=100}$ and $\bf{p=100}$, with error terms generated from \textbf{gamma} distributions based on the raw data (row 1) or log-transformed data (row 2).
The number of children nodes in the DAG is set to $p_0=0.6p$ (columns 1 and 2) or $p_0=0.8p$ (columns 3 and 4), and the number of non-zero signals is set to $q = 0.1p$ (columns 1 and 3) or $0.5p$ (columns 2 and 4).}
\end{figure}
\FloatBarrier
\noindent

\begin{figure}[ht!]
\centering
    \includegraphics[width=\textwidth]{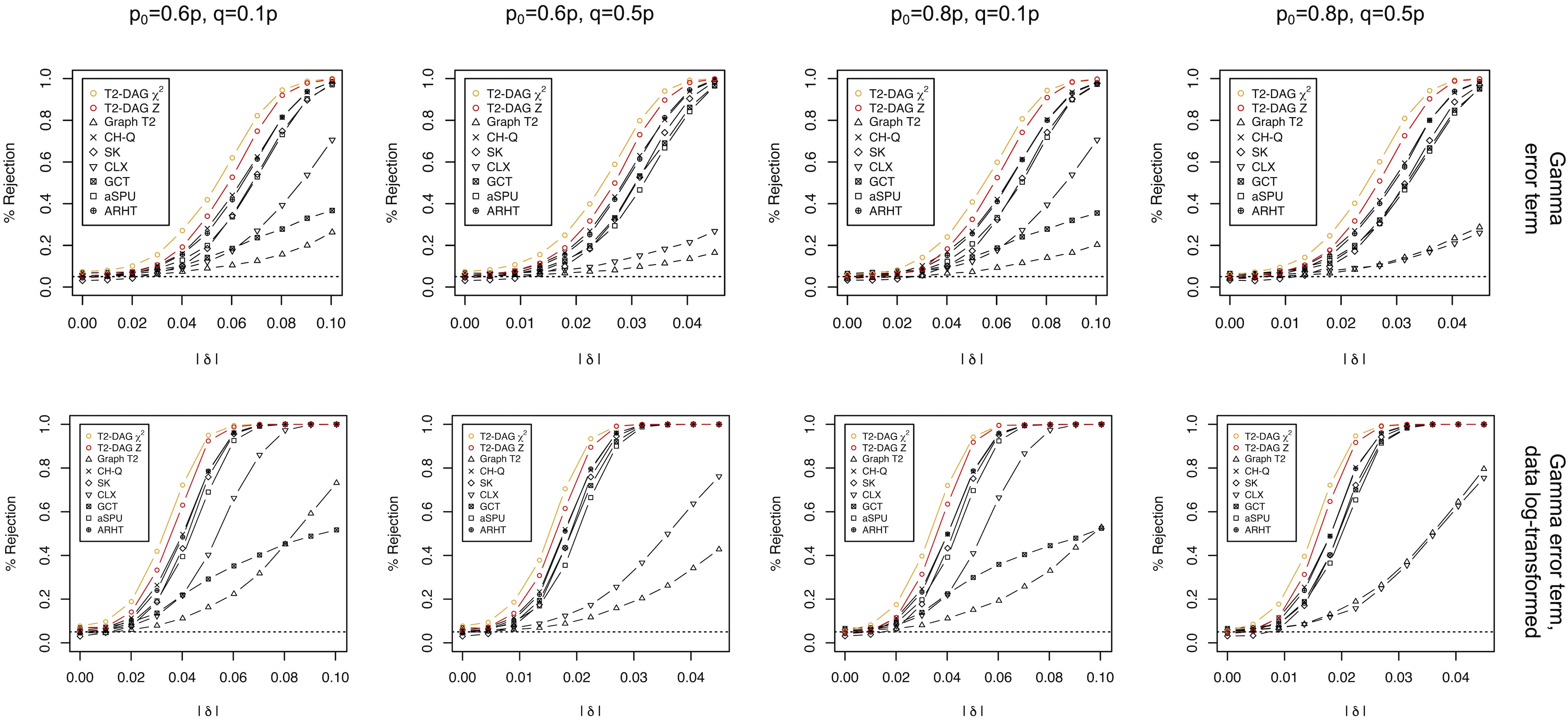}
\caption*{Figure S14. Empirical type I error rates ($\delta = 0$) and powers ($|\delta| > 0)$ of the various tests under $\bf{n_1=n_2=100}$ and $\bf{p=300}$, with error terms generated from \textbf{gamma} distributions based on the raw data (row 1) or log-transformed data (row 2).
The number of children nodes in the DAG is set to $p_0=0.6p$ (columns 1 and 2) or $p_0=0.8p$ (columns 3 and 4), and the number of non-zero signals is set to $q = 0.1p$ (columns 1 and 3) or $0.5p$ (columns 2 and 4).}
\end{figure}
\FloatBarrier
\noindent


\begin{figure}[ht!]
\centering
    \includegraphics[width=\textwidth]{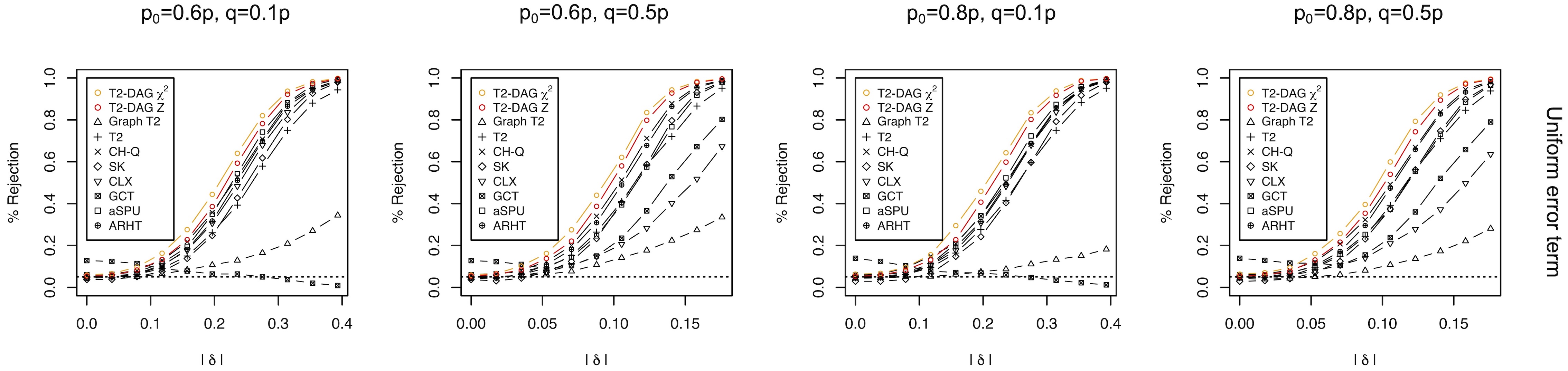}
\caption*{Figure S15. Empirical type I error rates ($\delta = 0$) and powers ($|\delta| > 0)$ of the various tests under $\bf{n_1=n_2=50}$ and $\bf{p=40}$, with error terms generated from \textbf{uniform} distributions.
The number of children nodes in the DAG is set to $p_0=0.6p$ (columns 1 and 2) or $p_0=0.8p$ (columns 3 and 4), and the number of non-zero signals is set to $q = 0.1p$ (columns 1 and 3) or $0.5p$ (columns 2 and 4).}
\end{figure}
\FloatBarrier
\noindent

\begin{figure}[ht!]
\centering
    \includegraphics[width=\textwidth]{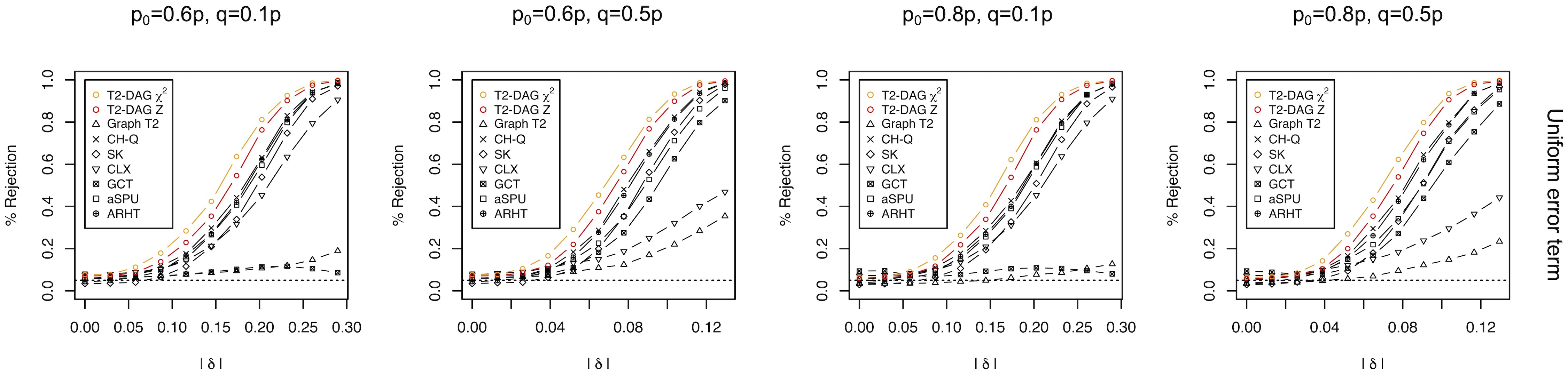}
\caption*{Figure S16. Empirical type I error rates ($\delta = 0$) and powers ($|\delta| > 0)$ of the various tests under $\bf{n_1=n_2=50}$ and $\bf{p=100}$, with error terms generated from \textbf{uniform} distributions.
The number of children nodes in the DAG is set to $p_0=0.6p$ (columns 1 and 2) or $p_0=0.8p$ (columns 3 and 4), and the number of non-zero signals is set to $q = 0.1p$ (columns 1 and 3) or $0.5p$ (columns 2 and 4).}
\label{fig.confounder}
\end{figure}
\FloatBarrier
\noindent

\begin{figure}[ht!]
\centering
    \includegraphics[width=\textwidth]{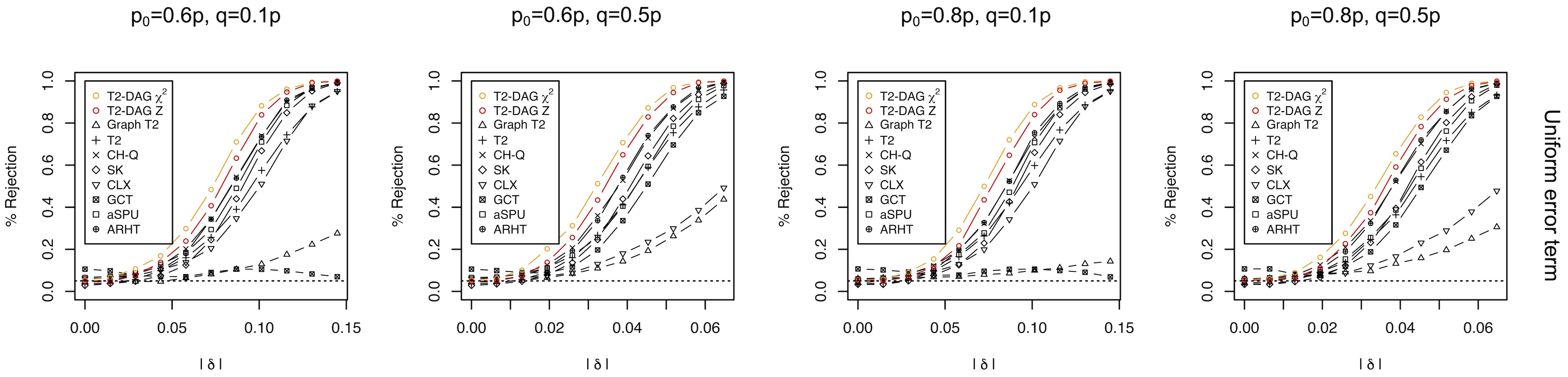}
\caption*{Figure S17. Empirical type I error rates ($\delta = 0$) and powers ($|\delta| > 0)$ of the various tests under $\bf{n_1=n_2=100}$ and $\bf{p=100}$, with error terms generated from \textbf{uniform} distributions.
The number of children nodes in the DAG is set to $p_0=0.6p$ (columns 1 and 2) or $p_0=0.8p$ (columns 3 and 4), and the number of non-zero signals is set to $q = 0.1p$ (columns 1 and 3) or $0.5p$ (columns 2 and 4).}
\end{figure}
\FloatBarrier
\noindent

\begin{figure}[ht!]
\centering
    \includegraphics[width=\textwidth]{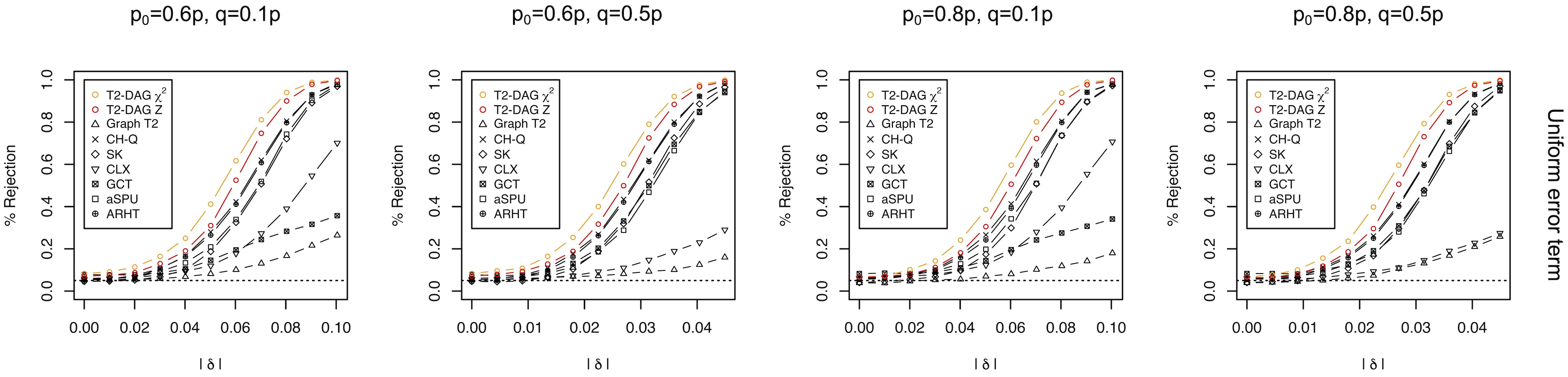}
\caption*{Figure S18. Empirical type I error rates ($\delta = 0$) and powers ($|\delta| > 0)$ of the various tests under $\bf{n_1=n_2=100}$ and $\bf{p=300}$, with error terms generated from \textbf{uniform} distributions.
The number of children nodes in the DAG is set to $p_0=0.6p$ (columns 1 and 2) or $p_0=0.8p$ (columns 3 and 4), and the number of non-zero signals is set to $q = 0.1p$ (columns 1 and 3) or $0.5p$ (columns 2 and 4).}
\end{figure}
\FloatBarrier
\noindent

\section{Additional information on real data analysis}

\begin{figure}[ht!]
\centering
	\includegraphics[width=\textwidth]{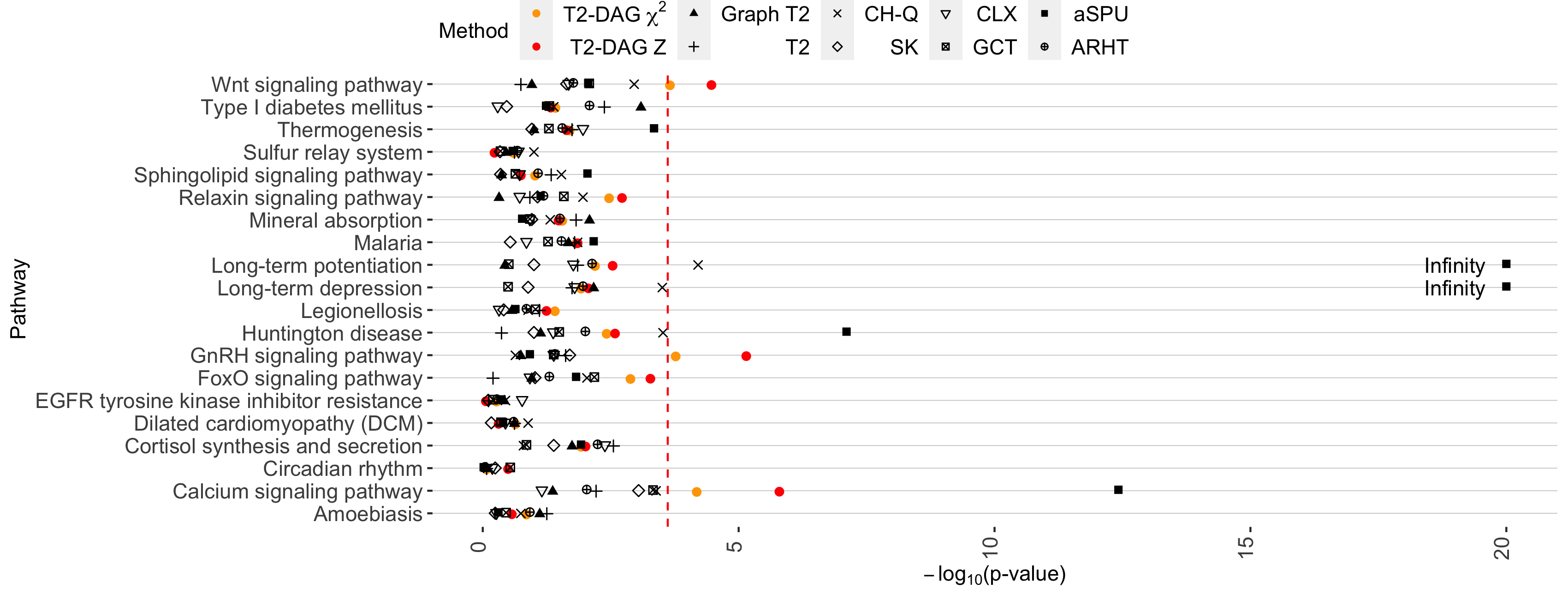}
	\caption*{Figure S16: $-{\log}_{10}(p)$ of the various tests for 20 randomly selected  pathways between stage I and stage II lung cancer.
\label{fig.I.II}}
\end{figure}
\FloatBarrier
\noindent
\begin{figure}[ht!]
\centering
	\includegraphics[width=\textwidth]{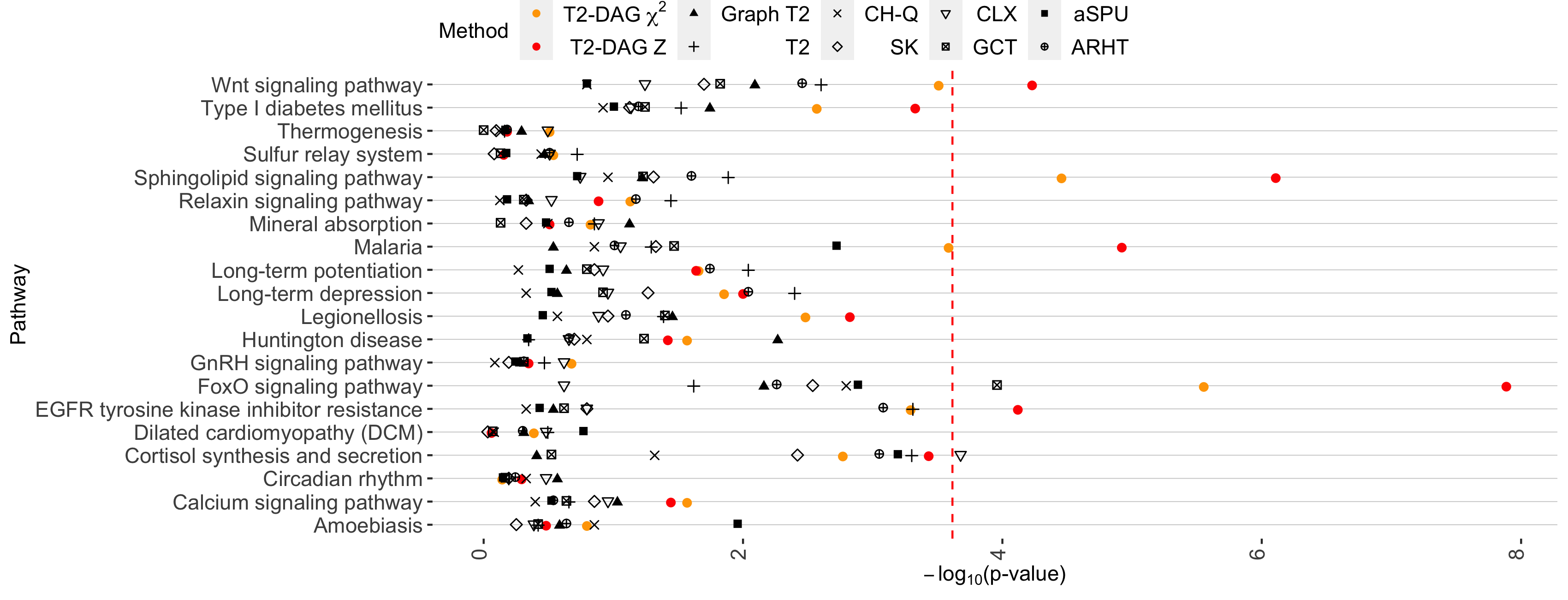}
	\caption*{Figure S17: The $-{\log}_{10}(p)$ of the various tests for 20 randomly selected  pathways between stage I and stage IV lung cancer.
\label{fig.I.IV}}
\end{figure}
\FloatBarrier
\noindent

\begin{figure}[ht!]
\centering
	\includegraphics[width=\textwidth]{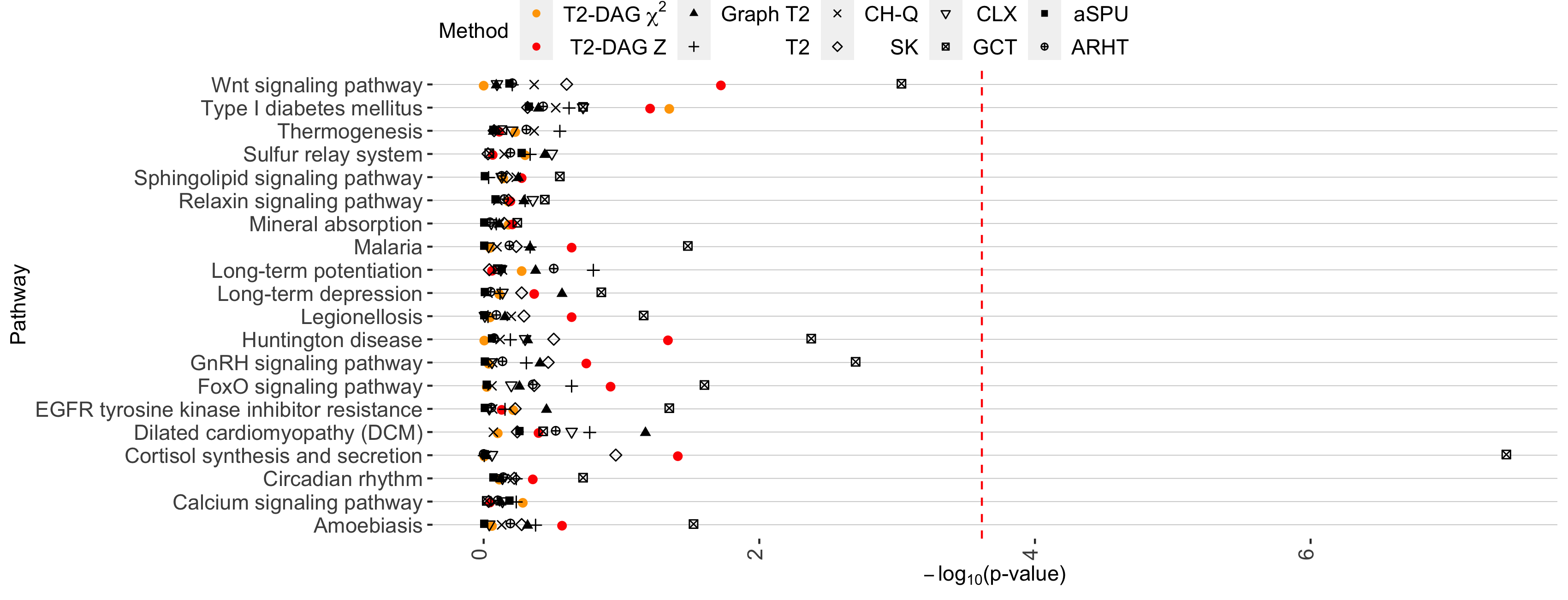}
	\caption*{Figure S18: The $-{\log}_{10}(p)$ of the various tests for 20 randomly selected  pathways between stage II and stage III lung cancer.
\label{fig.II.III}}
\end{figure}
\FloatBarrier
\noindent

\begin{figure}[ht!]
\centering
	\includegraphics[width=\textwidth]{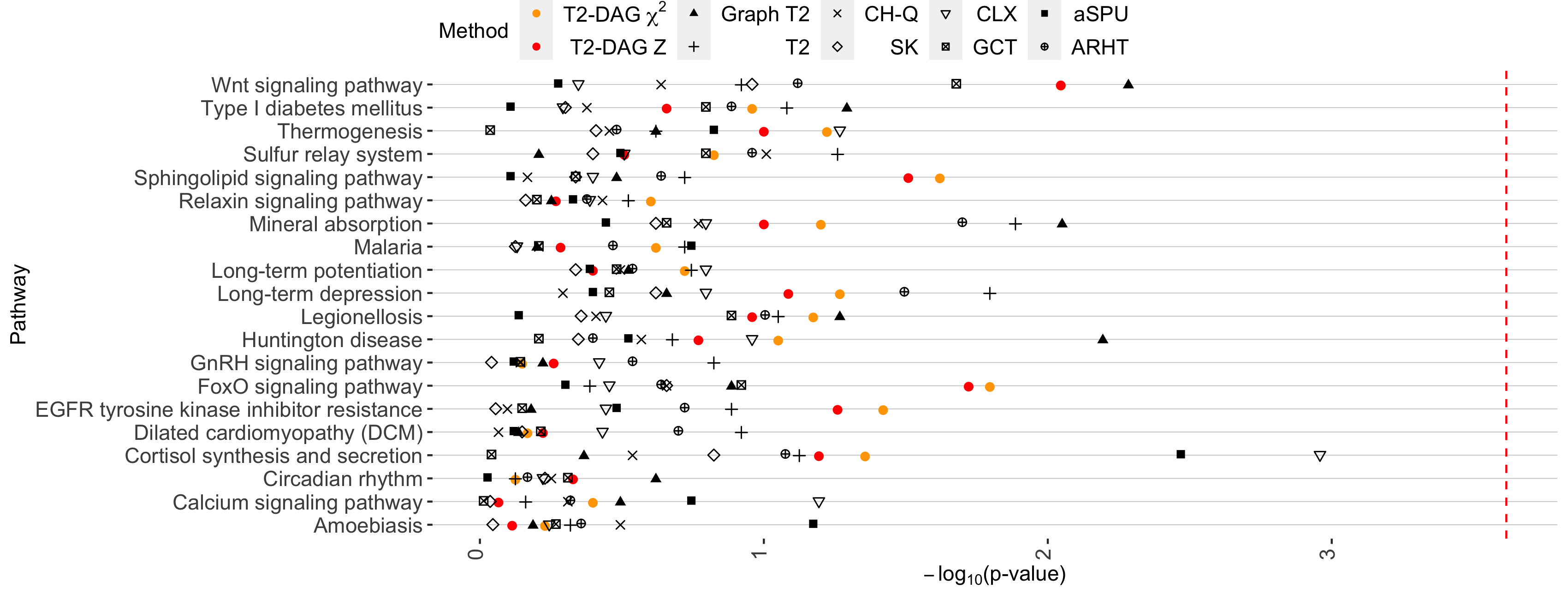}
	\caption*{Figure S19: The $-{\log}_{10}(p)$ of the various tests for 20 randomly selected  pathways between stage II and stage IV lung cancer.
\label{fig.II.IV}}
\end{figure}
\FloatBarrier
\noindent
\begin{figure}[ht!]
\centering
	\includegraphics[width=\textwidth]{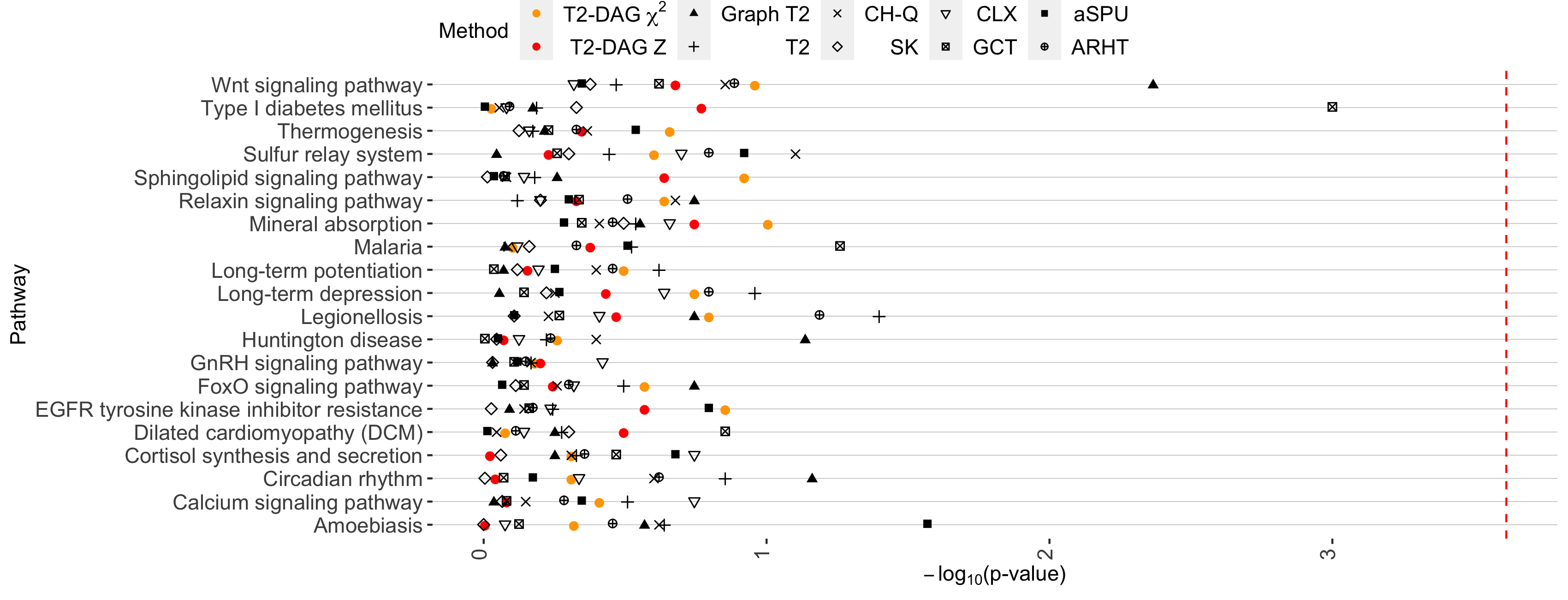}
	\caption*{Figure S20: The $-{\log}_{10}(p)$ of the various tests for 20 randomly selected  pathways between stage III and stage IV lung cancer.
\label{fig.III.IV}}
\end{figure}
\FloatBarrier
\noindent

\begin{figure}[ht!]
\centering
	\includegraphics[width=\textwidth]{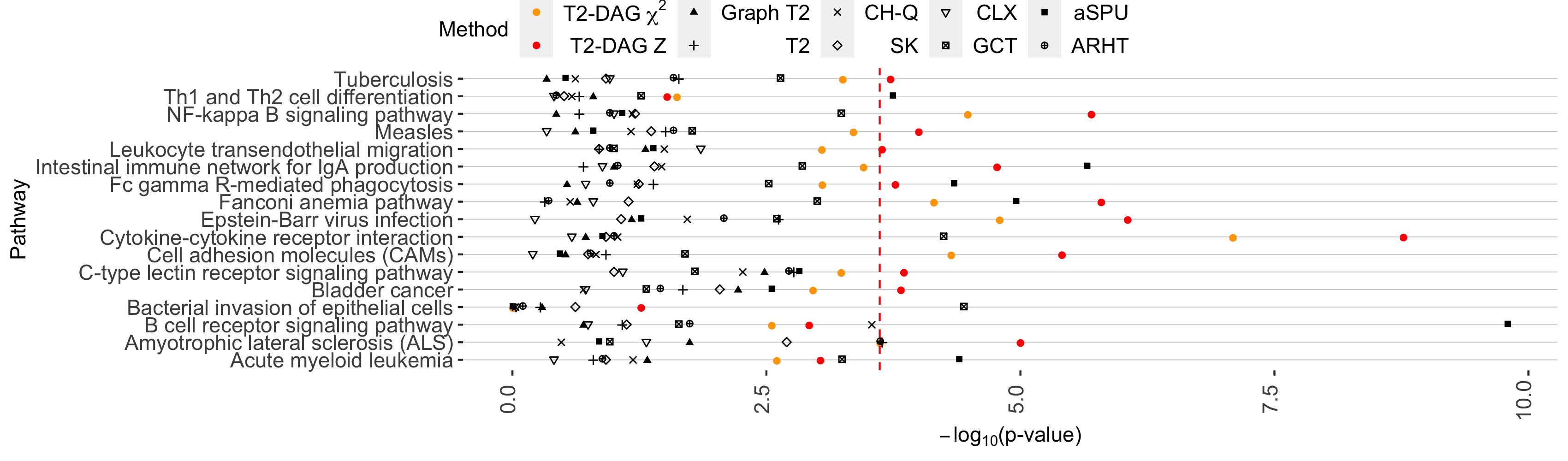}
	\caption*{Figure S21: The $-{\log}_{10}(p)$ of the comparison between stage II and stage IV lung cancer ($N_2=124$, $N_4=24$) for the pathways identified by any test to have significantly different expression levels between the two cancer stages.
\label{fig.II.III.significant}}
\end{figure}
\FloatBarrier
\noindent

\begin{figure}[ht!]
\centering
	\includegraphics[width=\textwidth]{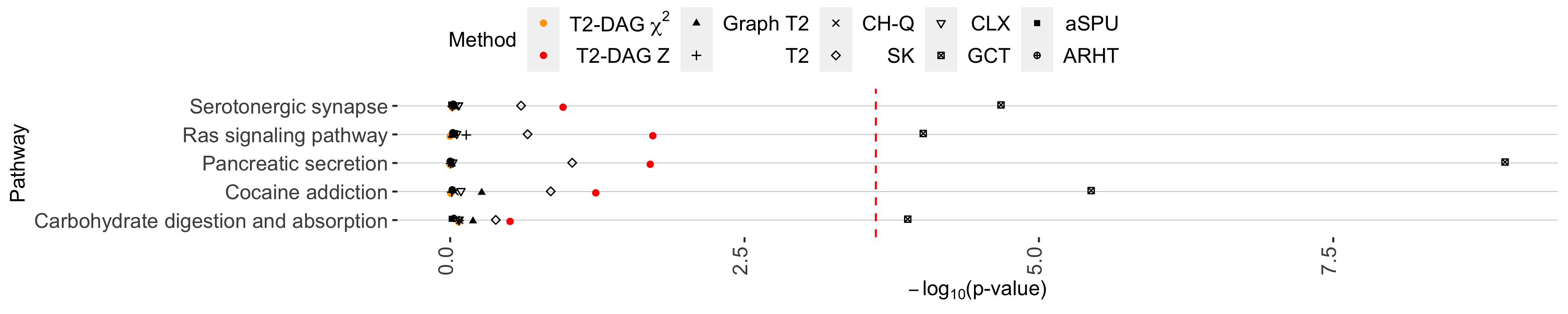}
	\caption*{Figure S22: The $-{\log}_{10}(p)$ of the comparison between stage III and stage IV lung cancer ($N_3=84$, $N_4=27$) for the pathways identified by any test to have significantly different expression levels between the two cancer stages.
\label{fig.III.IV.significant}}
\end{figure}
\FloatBarrier
\noindent

Detailed information on the pathways and testing results of data analysis  are summarized in Supplementary Tables S1 - S10 which are provided in a separate file named ``Supplemental Tables S1 - S10''.

\begin{table}[ht!]
\footnotesize{}
\def~{\hphantom{0}}
\centering
\caption*{Table S11: Results for testing differential expression of 206 pathways between stage I and stage II lung cancer ($N_1=278$, $N_2=124$). 
In each panel, the numbers in the diagonal cells are the numbers of significant pathways identified by the corresponding methods; the numbers in the $(i,j)$-th cell in the upper and lower triangle respectively indicate the number of significant pathways detected by both method $i$ and method $j$, and the correlation in $-{\log}_{10}(p)$ between the two methods. 
\label{table.I.II}}
\begin{tabular}{lcccccccccc}
& T2-DAG $\chi^2$ & T2-DAG $z$ & Graph T2 & T2 & CH-Q & SK & CLX & GCT & aSPU & ARHT\\ 
T2-DAG $\chi^2$ & 39 & 39 & 0 &   0 &  11 &   4 &   0 &   10 &   13 & 3 \\ 
T2-DAG $z$  & 0.98 & 57 & 0 &   0 &  13 &   4 &   0 &   10 &   18 & 3 \\ 
Graph T2 & 0.43 & 0.39 & 1 & 1 &   0 &   0 &   0 &   0 &   0 & 0\\
T2 & 0.39 & 0.34 & 0.70 &   1 &   0 &   0 &   0 &   0 &   0 & 0\\ 
CH-Q & 0.71 & 0.71 & 0.30 & 0.34 &  16 &   4 &   0 &   6 &   13 & 3\\ 
SK & 0.72 & 0.71 & 0.34 & 0.43 & 0.83 &   4 &   0 &   3 &   4 & 2\\ 
CLX & 0.40 & 0.34 & 0.23 & 0.35 & 0.40 & 0.48 &   0 &   0 &   0 & 0\\ 
GCT & 0.86 & 0.80 & 0.37 & 0.38 & 0.60 & 0.65 & 0.26 &   10 &   6 & 3\\
aSPU & 0.28 & 0.30 & 0.21 & 0.22 & 0.60 & 0.42 & 0.35 & 0.13 &   26 & 3\\
ARHT & 0.63 & 0.57 & 0.56 & 0.84 & 0.71 & 0.69 & 0.46 & 0.59 & 0.44 & 3\\
\end{tabular}
\end{table}

\begin{table}[ht!]
\footnotesize{}
\def~{\hphantom{0}}
\centering
\caption*{Table S12: Results for testing potential differential expression of 206 pathways between stage I and stage IV lung cancer ($N_1=278$, $N_4=27$).
In each panel, the numbers in the diagonal cells are the numbers of significant pathways identified by the corresponding methods; the numbers in the $(i,j)$-th cell in the upper and lower triangle respectively indicate the number of significant pathways detected by both method $i$ and method $j$, and the correlation in $-{\log}_{10}(p)$ between the two methods.
\label{table.I.IV}}
\begin{tabular}{lccccccccccc}
& T2-DAG $\chi^2$ & T2-DAG $z$ & Graph T2 & T2 & CH-Q & SK & CLX & GCT & aSPU & ARHT\\ 
  T2-DAG $\chi^2$ & 55 &  55 &   4 &  9 &   4 &   4 &   1 &  16 &   16 & 8\\ 
  T2-DAG $z$ &0.98 &   74 &   4 &  9 &   4 &   4 &   2 &  16 &   16 & 8\\ 
  Graph T2 & 0.54 & 0.55 &   4 &   2 &   3 &   2 &   0 &   2 &   3 & 3\\ 
  T2 & 0.43 &  0.41 & 0.49 &  12 &   2 &   2 &   0 &   0 &   3 & 7\\ 
  CH-Q & 0.44 & 0.54 & 0.43 & 0.32 &   5 &   2 &   0 &   1 &   5 & 3\\ 
  SK & 0.72 & 0.74 & 0.60 & 0.59 & 0.59 &  4 &   1 &   0 &   3 & 2\\ 
  CLX & 0.36 & 0.36 & 0.26 & 0.31 & 0.22 & 0.57 &  3 &   0 &   0 & 0\\ 
  GCT & 0.85 & 0.82 & 0.51 & 0.34 & 0.31 & 0.62 & 0.14 &  16 &   6 & 2\\ 
  aSPU & 0.39 & 0.45 & 0.47 & 0.27 & 0.69 & 0.52 & 0.33 & 0.33 &   19 & 5\\ 
  ARHT & 0.39 & 0.48 & 0.40 & 0.47 & 0.95 & 0.55 & 0.22 & 0.23 & 0.54 & 9\\
\end{tabular}
\end{table}
\FloatBarrier

\begin{table}[ht!]
\footnotesize{}
\def~{\hphantom{0}}
\centering
\caption*{Table S13: Results for testing potential differential expression of 206 pathways between stage II and stage IV lung cancer ($N_2=124$, $N_4=27$). In each panel, the numbers in the diagonal cells are the numbers of significant pathways identified by the corresponding methods; the numbers in the $(i,j)$-th cell in the upper and lower triangle respectively indicate the number of significant pathways detected by both method $i$ and method $j$, and the correlation in $-{\log}_{10}(p)$ between the two methods.
\label{table.II.IV}}
\begin{tabular}{llccccccccc}
& T2-DAG $\chi^2$ & T2-DAG $z$ & Graph T2 & T2 & CH-Q & SK & CLX & GCT & aSPU & ARHT\\ 
T2-DAG $\chi^2$ & 6 & 6 & 0 &   1 &   0 &   0 &   0 &   1 &   1 & 1\\ 
T2-DAG $z$ & 0.97 & 13 & 0 &   1 &   0 &   0 &   0 &   1 &   3 & 1\\ 
Graph T2  & 0.38 & 0.34 &   0 &   0 &   0 &   0 &   0 &   0 &   0 & 0\\ 
T2   & 0.45 & 0.40 & 0.49 &   1 &   0 &   0 &   0 &   0 &   0 & 1\\ 
CH-Q  & 0.60 & 0.58 & 0.36 & 0.30 &   0 &   0 &   0 &   0 &   0 & 0\\ 
SK   & 0.74 & 0.77 & 0.45 & 0.45 & 0.55 &   0 &   0 &   0 &   0 & 0\\ 
CLX   & 0.29 & 0.25 & 0.13 & 0.27 & 0.19 & 0.44 &   0 &   0 &   0 & 0\\ 
GCT   & 0.58 & 0.67 & 0.20 & 0.03 & 0.40 & 0.50 & -0.11 &   2 &   0 & 0\\ 
aSPU   & 0.43 & 0.45 & 0.16 & 0.18 & 0.64 & 0.47 & 0.38 & 0.29 &   6 & 0\\ 
ARHT   & 0.50 & 0.44 & 0.40 & 0.91 & 0.42 & 0.51 & 0.27 & 0.11 & 0.22 & 1\\ 
\end{tabular}
\end{table}

\begin{table}[ht]
\footnotesize{}
\def~{\hphantom{0}}
\centering
\caption*{Table S14: Results for testing potential differential expression of 206 pathways between stage III and stage IV lung cancer ($N_3=84$, $N_4=24$). In each panel, the numbers in the diagonal cells are the numbers of significant pathways identified by the corresponding methods; the numbers in the $(i,j)$-th cell in the upper and lower triangle respectively indicate the number of significant pathways detected by both method $i$ and method $j$, and the correlation in $-{\log}_{10}(p)$ between the two methods.
\label{table.III.IV}}
\begin{tabular}{lcccccccccc}
& T2-DAG $\chi^2$ & T2-DAG $z$ & Graph T2 & T2 & CH-Q & SK & CLX & GCT & aSPU & ARHT\\ 
  T2-DAG $\chi^2$ & 0 &  0 & 0 & 0 & 0 & 0 & 0 & 0 & 0 & 0\\ 
  T2-DAG $z$ & 0.55 & 0 & 0 & 0 & 0 & 0 & 0 & 0 & 0 & 0\\ 
  Graph T2 & 0.33 & 0.01 & 0 & 0 & 0 &  0 &  0 & 0 &  0 & 0\\ 
  T2 & 0.35 & -0.01 & 0.34 & 0 & 0 &  0 &   0 &   0 &  0 & 0\\ 
  CH-Q & 0.50 & 0.13 & 0.35 & 0.22 & 0 & 0 &  0 & 0 &  0 & 0\\ 
  SK & 0.37 & 0.89 & -0.01 & -0.07 & 0.12 & 0 & 0 &   0 &   0 & 0\\ 
  CLX & 0.45 & 0.10 & 0.08 & 0.30 & 0.22 & 0.07 & 0 &   0 &   0 & 0\\ 
  GCT & -0.21 & 0.53 & -0.20 & -0.26 & -0.21 & 0.62 & -0.27 & 5 &  0 & 0\\ 
  aSPU & 0.38 & 0.11 & 0.12 & 0.15 & 0.47 & 0.05 & 0.41 & -0.16 &  0 & 0\\ 
  ARHT  & 0.40 & 0.01 & 0.41 & 0.86 & 0.44 & -0.04 & 0.25 & -0.25 & 0.18 & 0\\ 
\end{tabular}
\end{table}
\FloatBarrier

\bibliographystyle{plainnat}
\bibliography{c-ref}